\documentclass[english,cleveref, autoref, thm-restate, a4paper]{socg-lipics-v2021}

\hideLIPIcs

\usepackage{algorithm2e}
\usepackage{enumitem}
\usepackage{blkarray}

\newcommand{\squeezelist}{\setlength{\itemsep}{0pt}}

\newtheorem{thm}{Theorem}
\newtheorem{obs}[theorem]{Observation}
\newtheorem{lem}[theorem]{Lemma}
\newtheorem{coro}[theorem]{Corollary}

\newtheorem{assumption}{Assumption}
\newtheorem{assumptions}{Assumptions}

\newcommand{\changed}[1]{{#1}}
\newcommand{\newchanged}[1]{{#1}}
\newcommand{\nnewchanged}[1]{{#1}}
\newcommand{\reviewerchange}[1]{{#1}}

\def\defn#1{\textit{\textbf{\boldmath #1}}}

\title{The Geodesic Edge Center of a Simple Polygon}

\author{Anna Lubiw}
{David R.~Cheriton School of Computer Science, University of Waterloo, Waterloo, Canada}
{alubiw@uwaterloo.ca}
{https://orcid.org/0000-0002-2338-361X}{}

\author{Anurag Murty Naredla}{David R.~Cheriton School of Computer Science, University of Waterloo, Waterloo, Canada \and Institut für Informatik, University of Bonn, Bonn, Germany}{amnaredla@uwaterloo.ca}{}{}

\authorrunning{A.~Lubiw and A.M.~Naredla}

\Copyright{Anna Lubiw and Anurag Murty Naredla} %

\ccsdesc[500]{Theory of computation~Computational geometry}

\keywords{geodesic center of polygon, farthest edges, farthest-segment Voronoi diagram} %

\relatedversion{A full version of the paper is available at \url{https://arxiv.org/abs/papernumber}.}

\acknowledgements{ For helpful comments we thank Boris Aronov, Therese Biedl, John Hershberger, Joseph Mitchell, and the SoCG reviewers.
}

\funding{Research supported by the Natural Sciences and Engineering Research Council of Canada (NSERC).}

\begin{document}
\maketitle

\begin{abstract}
The \emph{geodesic edge center} of a polygon is a point $c$ inside the polygon that minimizes the maximum geodesic distance 
from $c$ to any edge of the polygon, where \emph{geodesic distance} is the shortest path distance inside the polygon.
We give a \reviewerchange{linear-time algorithm} to find a geodesic edge center of a simple polygon. 
This improves on the previous $O(n \log n)$ time algorithm by Lubiw and Naredla [European Symposium on Algorithms, 2021].
The algorithm builds on an  
algorithm to find the geodesic \emph{vertex} center of a simple polygon due to Pollack, Sharir, and Rote [Discrete \& Computational Geometry, 1989] and
an improvement to linear time by Ahn, Barba, Bose, De Carufel, Korman, and Oh [Discrete \& Computational Geometry, 2016].

The
geodesic edge center 
can easily be found from the
\reviewerchange{geodesic farthest-edge} Voronoi diagram of the polygon.
Finding that Voronoi diagram
in linear time is an open question, although
the geodesic \emph{nearest} edge Voronoi diagram 
(the medial axis) can be found in linear time.
As a first step of our geodesic edge center algorithm,  we give a \reviewerchange{linear-time algorithm} to find the \reviewerchange{geodesic farthest-edge} Voronoi diagram  
restricted to the polygon boundary. 
\end{abstract}

\section{Introduction}\label{section:introduction}

The most basic ``center'' problem is Sylvester's problem: given $n$ points in the plane, find the smallest disc that encloses the points. The center of this disc is a point that minimizes the maximum distance to any of the given points.
We consider a center problem that differs in two ways from Sylvester's problem. First, the domain is a simple polygon and the distance measure is not Euclidean distance, but rather the shortest path, or ``geodesic'' distance inside the polygon.
Second, the sites are not points but rather the edges of the polygon.  More precisely, the problem is to find, given a simple polygon in the plane, the \defn{geodesic edge center}, which is a point in the polygon that minimizes the maximum geodesic distance to a polygon edge. 
See Figure~\ref{fig:center-example}.
More formally, let $E$ be the set of edges of the polygon $P$, and for point $p \in P$ and edge $e \in E$, define \defn{$d(p,e)$} to be the geodesic distance from $p$ to $e$.  
\changed{Define the \defn{geodesic radius} of a point $p \in P$ to be $r(p) := \max \{d(p,e): e \in E\}$.
Then the \defn{geodesic edge center} is a point $p \in P$ that minimizes
$r(p)$.}

\begin{figure}[htb]
  \centering
  \includegraphics[width=.65\textwidth]{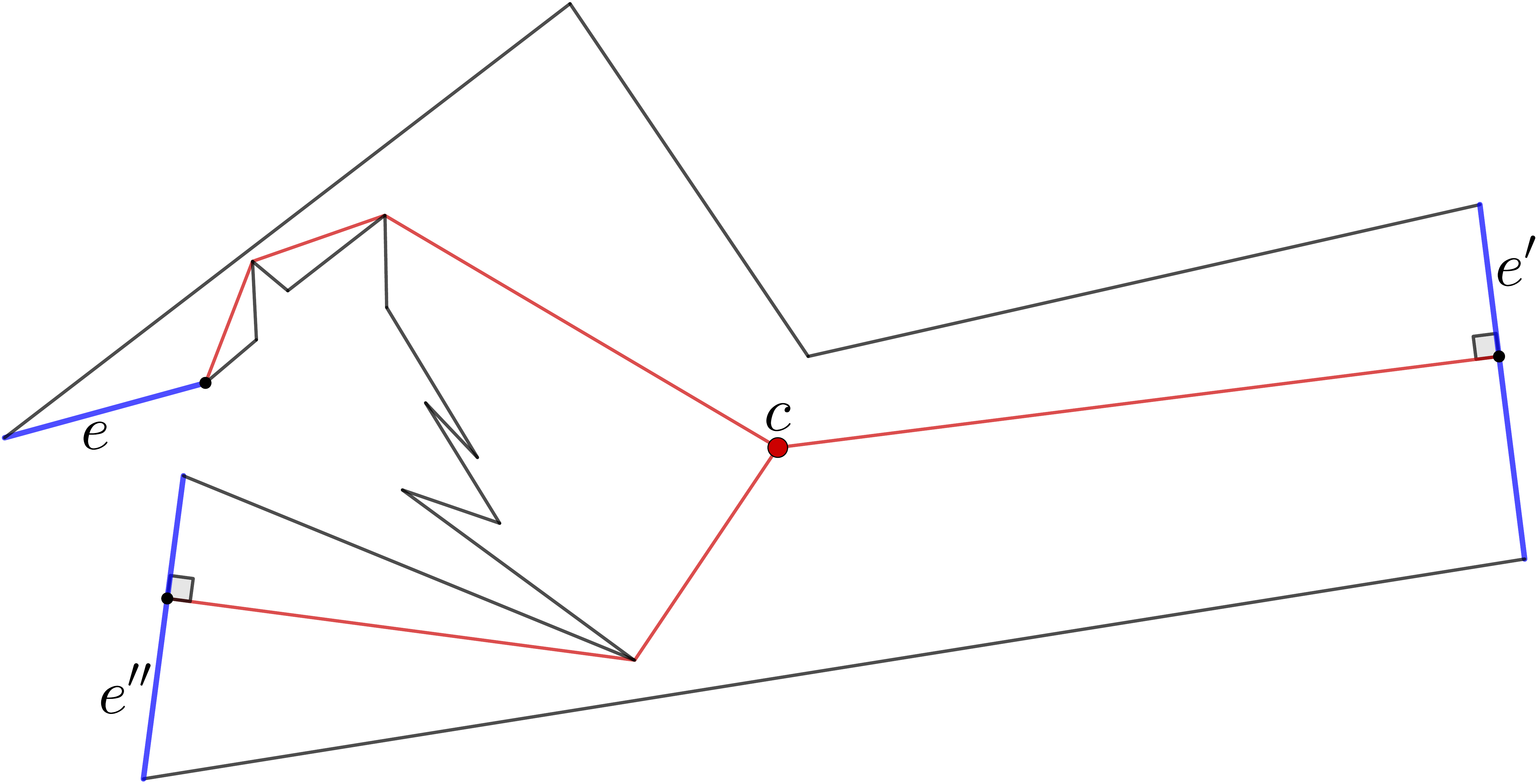}
  \caption{Point $c$ is the edge center of this polygon.
  Edges $e, e', e''$ (in blue)
  are geodesically farthest
  from $c$---the geodesic paths (in red) from $c$ to these edges all
have the same length.}
\label{fig:good_four_cell}
\label{fig:center-example}
\end{figure}

Our main result is a linear-time algorithm to find the geodesic edge center of a simple $n$-vertex polygon.  This improves %
\reviewerchange{our} previous 
$O(n \log n)$ time algorithm~\cite{lubiw2021visibility}.
The algorithm follows the strategies used to find the geodesic \emph{vertex} center, which is a point in the polygon that minimizes the maximum geodesic distance to a polygon vertex.  
In 1989, Pollack, Sharir and Rote~\cite{pollack_sharir} gave an $O(n \log n)$ time algorithm for the geodesic vertex center problem.  A main tool---which is used in all subsequent algorithms---is
a \reviewerchange{linear-time} \emph{chord oracle} that finds, given a chord, which side of the chord contains the center. 
In 2016, Ahn, Barba, Bose, De Carufel, Korman, and Oh~\cite{linear_time_geodesic} improved the runtime for the geodesic vertex center to $O(n)$.  
Their most important new contribution is the use of $\epsilon$-nets to perform a divide-and-conquer search. 
Our algorithm follows the approach of Ahn et al., modified to deal with farthest edges rather than farthest vertices.  We simplify some aspects and we repair some  
errors in their approach.  
\newchanged{The \reviewerchange{edge-center problem} is more general than the vertex center problem via the reduction of splitting each vertex into  two vertices joined by a very short edge.}

In general, the center of a set of sites can be determined from the farthest Voronoi diagram of those sites, but computing the Voronoi diagram can be more costly. 
As the first step of 
our center algorithm we give a \reviewerchange{linear-time algorithm} to compute the \reviewerchange{geodesic farthest-edge} Voronoi diagram restricted to the boundary of the polygon. 
\changed{Computing the whole \reviewerchange{geodesic farthest-edge} Voronoi diagram in linear time is an open problem.}

\subparagraph*{Background on Centers and Farthest Voronoi Diagrams.} 

\changed{
Megiddo~\cite{megiddo_linear} gave a \reviewerchange{linear-time} algorithm to find the center of a set of points in the plane (Sylvester's problem) using his ``prune-and-search'' technique, which is used in the final stages of all geodesic center algorithms.
However, computing the farthest Voronoi diagram of points in the plane takes $\Theta(n \log n)$ time~\cite{shamos1975closest}.
}

Our problem involves distances that are geodesic rather that Euclidean, and sites that are segments (edges) rather than points. \changed{These have been studied separately, although there is almost no work
combining them.}

For Euclidean distances, 
Megiddo's method extends to \reviewerchange{linear-time algorithm}s to find the center of %
line segments or lines in the plane~\cite{bhattacharya1994optimal}.
The farthest Voronoi diagram of segments in the plane was considered by Aurenhammer et al.~\cite{aurenhammer2006farthest},
who called it a
``stepchild in
the vast Voronoi diagram literature''.  
They gave an $O(n \log n)$ time algorithm which was improved to output-sensitive time $O(n \log h)$, where $h$ is the number of faces of the diagram~\cite{papadopoulou2013farthest}.

For geodesic 
distances with point sites Ahn et al. gave 
a \reviewerchange{linear-time algorithm} to find the geodesic center of the vertices of a polygon~\cite{linear_time_geodesic}.
The corresponding farthest Voronoi diagram 
can be found in time $O(n \log \log n)$~\cite{oh2020geodesic}, and in expected linear time~\cite{barba2019optimal}. 
More generally, for $m$ points inside an $n$-vertex polygon, an algorithm to find their farthest Voronoi diagram 
was first given by Aronov et al.~\cite{aronov1993farthest} with run-time $O((n+m)\log (n+m))$, and improved in a sequence of papers~\cite{oh2020geodesic,barba2019optimal,oh2020voronoi}, culminating in an optimal run time of   
$O(n + m \log m)$~\cite{wang2021optimal}. \changed{This is also the best-known bound for finding the center of $m$ points in a simple polygon.}
For sites more general than point sites inside a polygon, the only result we are aware of is our 
$O((n+m) \log (n+m))$ time algorithm to find the geodesic center of $m$ \emph{half-polygons}~\cite{lubiw2021visibility}, 
with edges being a special case. 

Finally, we mention a curious difference between nearest and farthest site Voronoi diagrams of edges
in a polygon.
The nearest Voronoi diagram of the edges of a polygon is the medial axis, 
one of the most famous and useful Voronoi diagrams.  The medial axis can be found in linear time~\cite{chin1999finding}.
By contrast,
the \emph{farthest} Voronoi diagram of edges in a polygon has received virtually no attention,
\changed{except for a convex polygon (which avoids geodesic issues)
where there is an $O(n \log n)$ time algorithm~\cite{drysdale2008nlogn}, and a recent 
expected \reviewerchange{linear-time algorithm}~\cite{khramtcova2014expected}.
}

\section{Overview of the Algorithm}
\label{sec:overview}

\changed{Before giving the overview of our algorithm, we outline the previous work that our algorithm builds upon, and explain what is novel about our contributions.}

Pollack et al.~\cite{pollack_sharir} gave an $O(n \log n)$ time algorithm to find the geodesic vertex center of a simple polygon. 
A main ingredient is to solve the problem one dimension down.
In particular, they develop an $O(n)$ time 
\emph{chord oracle}
that, given
a chord of the polygon, finds the \emph{relative center} restricted to the chord
and from that, determines whether the 
center of the polygon lies to left or right of the chord.  
By applying the chord oracle $O(\log n)$ times, they limit the 
search to a convex subpolygon where Euclidean distances can be used.
This reduces the problem to finding 
a minimum disc that encloses some disks, which Megiddo~\cite{megiddo_spanned_ball} solved in linear time using the same approach as  for his linear programming algorithm.
\reviewerchange{We extended the chord oracle to handle farthest \emph{edges} instead of vertices~\cite{lubiw2021visibility}.}

The idea used in the chord oracle algorithm is central to further developments. Expressed in general terms, 
the goal is to find a point in a domain (either a chord or the whole polygon) that minimizes the maximum distance to a %
site (a vertex or edge of the polygon). The idea is to first find what we will call a \defn{coarse cover} of the domain by a linear number of elementary regions $R$ (intervals or triangles), 
\changed{each with an associated easy-to-compute convex function $f_R$ that captures the geodesic distance to a potential farthest edge, and with the property that
the upper envelope of the functions is 
the geodesic radius function.} 
Thus, the goal is to find the point $x$ that minimizes the upper envelope of the  functions $f_R$. 
When the domain is a chord,
the chord oracle solves this in linear time.

\changed{When the domain is the whole polygon, and the sites are vertices, 
Ahn et al.~\cite{linear_time_geodesic} gave a \reviewerchange{linear-time algorithm}.
They find a coarse cover of the whole polygon starting from 
Hershberger and Suri's algorithm~\cite{hershberger1997matrix} to find the farthest vertex from each vertex.  
They then use divide-and-conquer based on $\epsilon$-nets---their big innovation---to reduce the domain to a triangle. After that, the vertex center is found using  
Megiddo-style prune-and-search techniques 
like those used by 
Pollack et al.}

Our algorithm uses a similar approach, modified to deal with farthest \emph{edges} rather than vertices. 
Another difference is that 
we give a simpler method of finding a coarse cover of the polygon by first finding 
the \reviewerchange{geodesic farthest-edge} Voronoi diagram on the polygon boundary.
There is a \reviewerchange{linear-time algorithm} to find the geodesic farthest \emph{vertex} Voronoi diagram on the polygon boundary by Oh, Barba, and Ahn~\cite{oh2020voronoi}.
Our algorithm is considerably simpler, and 
\reviewerchange{it is a novel idea to use the boundary Voronoi diagram to find the center.}

Other differences between our approach and that of Ahn et al. are introduced in order to 
repair 
some flaws in 
their paper.
\changed{
They use $\epsilon$-net techniques, but their range space does not have the necessary properties for finding $\epsilon$-nets in deterministic linear time. We remedy this by using a different range space, thereby repairing and generalizing their result.}

\paragraph*{Algorithm Overview}

\paragraph*{Phase I: Finding the \reviewerchange{Farthest-Edge} Voronoi Diagram Restricted to
the Polygon Boundary (Section~\ref{section:Phase-I})}
We first show that the \reviewerchange{linear-time algorithm} of Hershberger and Suri~\cite{hershberger1997matrix} that finds the farthest \emph{vertex} from each vertex can be modified to find the farthest \emph{edge} from each vertex.
A polygon edge $e$ whose endpoints have the same farthest edge $g$ is then part of the farthest Voronoi region of $g$.
To find the Voronoi diagram on a 
\emph{transition edge}
$e$ 
that has different farthest edges at its endpoints, we must find
the upper envelope of the coarse cover of $e$.   
\changed{We use the fact that the coarse cover of $e$ is 
constructed from two shortest path trees inside a smaller subpolygon called the \emph{hourglass} of $e$. The hourglasses of all transition edges can be found in linear time. 
In each hourglass, 
the shortest path trees \reviewerchange{allow} us to construct the upper envelope incrementally in linear time---this is a main new aspect of our work.}

\paragraph*{Phase II: Finding the Geodesic Edge Center (Section~\ref{section:PhaseII})}
We first find a coarse cover of the polygon by triangles, 
each bounded by two polygon chords plus a segment of an edge, and 
each with an associated convex function that captures the 
geodesic distance to a potential farthest edge---the potential farthest edges are 
those that have 
non-empty Voronoi regions
on the boundary of $P$.
The problem of finding the edge center is then reduced to the problem of finding
the point that minimizes the upper envelope of the coarse cover functions.

To find this point we use divide-and-conquer, reducing in each step to a smaller subpolygon with a constant fraction of the coarse cover elements. 
There are two stages.
In Stage 2, once the subpolygon is a triangle, the prune-and-search approach of Megiddo's can be applied.  
In Stage 1
every coarse cover triangle that intersects the subpolygon has a boundary chord crossing the subpolygon, and $\epsilon$-net techniques are used to reduce the number of such chords, and hence the number of coarse cover elements.
Our approach follows that of Ahn et al.~\cite{linear_time_geodesic} but 
we repair some flaws---another main new aspect of our work.
Ahn et al.~recurse on subpolygons called ``4-cells'' that are the intersection of four half-polygons (a \defn{half-polygon} is the subpolygon to one side of a chord).
\newchanged{We instead recurse on ``$3$-anchor hulls'' that are the geodesic convex hulls of at most three points or subchains of the polygon boundary.} 
These define a 
range space 
whose ground set is 
a set of chords and whose ranges are subsets of chords that cross a
$3$-anchor hull. 
We prove that our range space has finite VC-dimension, which repairs the faulty proof in Ahn et al.~for 4-cells.
Even more crucially, we give a ``subspace oracle'' that permits an $\epsilon$-net to be found in \emph{deterministic} linear time, 
\newchanged{something missing from their approach.}

\section{Preliminaries}
\label{section:prelims}

Although our algorithm follows the pattern of the geodesic vertex center algorithm by Ahn et al.~\cite{linear_time_geodesic}, we must re-do everything from the ground up to deal with farthest edges. 
In this section we summarize
\reviewerchange{some basic results,}
deferring proofs and details to the
appendix.

\subparagraph*{Notation and Definitions.}
\label{sec:notation}
A \defn{chord} of a polygon is a straight line segment in the (closed) polygon with both endpoints on the boundary, $\partial P$. 
For a point $p \in P$ and a point or line segment $s$ in $P$, \defn{$\pi(p,s)$} is the unique \defn{shortest} (or \defn{geodesic}) path from $p$ to $s$, and \defn{$t(p,s)$} is its \defn{terminal point}.
The length of $\pi(p,s)$, denoted \defn{$d(p,s)$}, is the \defn{geodesic distance} from $p$ to $s$.
\reviewerchange{For point $p$ in $P$, a \defn{farthest edge}, \defn{$F(p)$}, is an edge for which $d(p, F(p)) \geq d(p,e)$, for %
every edge $e$ of $P$.}
The \defn{geodesic radius} of point $p$ is \defn{$r(p)$} $:= d(p,F(p))$.
The \defn{geodesic edge center} is a point $p \in P$ that minimizes
$r(p)$.

\subparagraph*{\reviewerchange{General Position} Assumptions.}
\label{sec:general-position}

\changed{As is standard for Voronoi diagrams of segments, e.g., see~\cite{aurenhammer2006farthest}, we use the following tie-breaking rule %
to prevent 2-dimensional Voronoi regions with more than one farthest edge.}

\subparagraph*{Tie-Breaking Rule.} 
Suppose that $p$ is a point of $P$, $e$ and $f$ are two edges that meet at reflex vertex $u$, and
$\pi(p,e) = \pi(p,f) = \pi(p,u)$. 
\changed{Let line $b$ be the angle bisector of $u$.
For $p$ not on $b$, break the tie $d(p,e)=d(p,f)$
by saying that the distance to the edge on the opposite side of $b$ is greater.}

\medskip
\changed{We make the following \reviewerchange{general position} assumptions, which we claim can be effected by perturbing vertices.}

\changed{
\begin{assumptions}
\label{assumptions}
(1) No three vertices of 
    $P$ are collinear.
(2) After imposing the tie-breaking rule, no vertex is equidistant from two or more edges.
(3) No point on the polygon boundary has more than two farthest edges and no point in the interior of the polygon has more than a constant number
\reviewerchange{(six)}
of farthest edges.
\end{assumptions}
}
It follows that the set of points with more than one farthest edge 
is 1-dimensional, does not contain any vertex of $P$, and intersects $\partial P$ in isolated points; see
Lemma~\ref{lem:1D-Voronoi-edges} in Appendix~\ref{appendix:general-position}.

\subparagraph*{Properties of Farthest Edges.}
\label{sec:farthest-properties}
We need the following basic ``triangle property'' (proved in
Appendix~\ref{appendix:farthest-properties}) 
about shortest paths that cross. 

\begin{lem}
\label{lem:new_order}
Suppose the points $p$, $q$ and the edges $e$, $f$ occur in the order $p,q,e,f$ along the polygon boundary $\partial P$.
Then $d(p,e) + d(q,f) \geq d(p,f) + d(q,e)$.
\end{lem}

We then characterize what two paths to farthest edges are like, see Figure~\ref{fig:correct_orderings} and
Appendix~\ref{appendix:farthest-properties}.  
In particular, we generalize the \reviewerchange{farthest-vertex} Ordering Property~\cite{aronov1993farthest} as follows.

\subparagraph*{The Ordering Property.}
As $p$ moves clockwise around $\partial  P$, so does $F(p)$.

\begin{figure}
  \centering
  \includegraphics[width=0.6\textwidth]{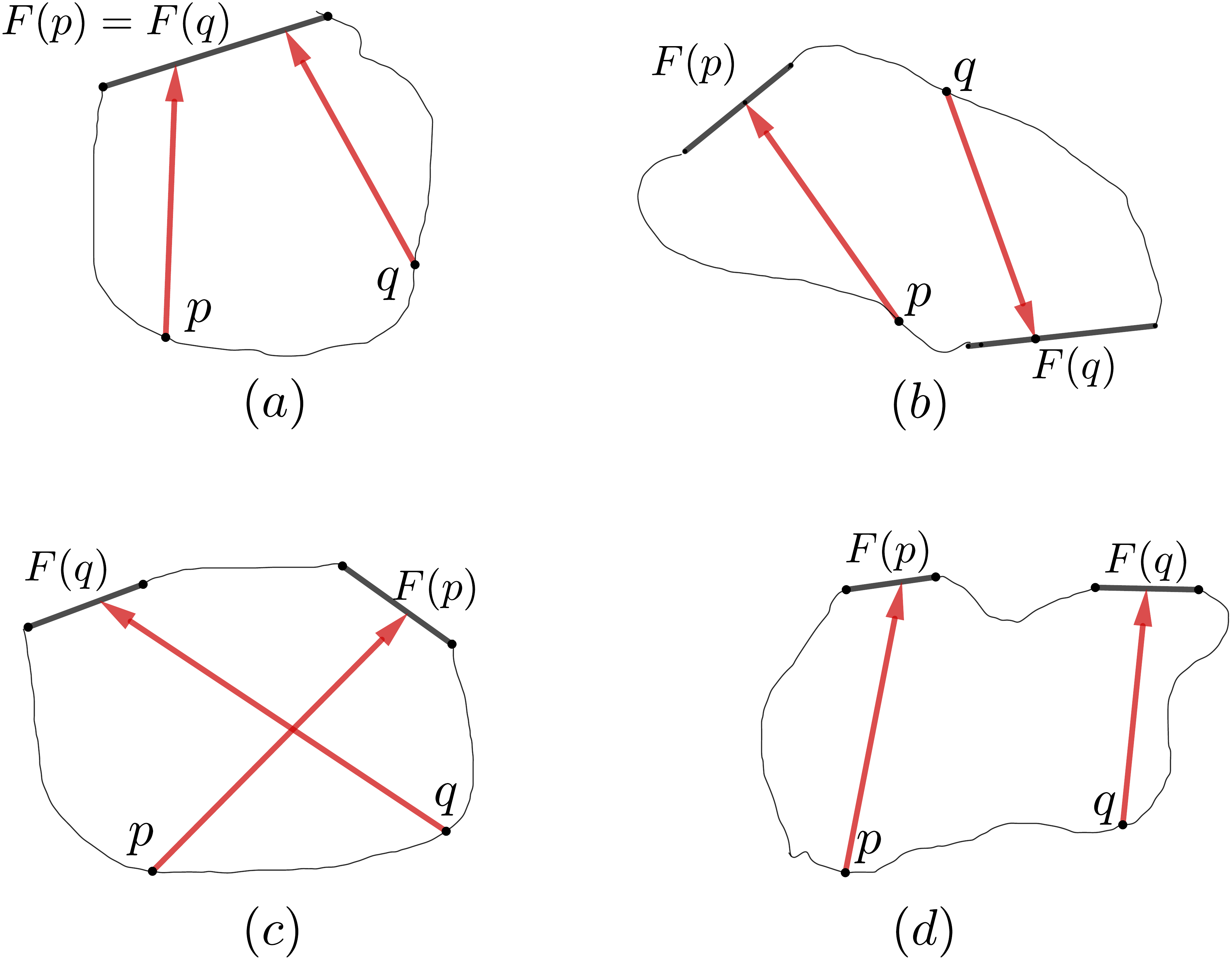}
  \caption{
Schematics for possible and impossible orderings of points $p$, $q$ and their farthest edges $F(p)$ and $F(q)$.  (a) The only possible ordering if $F(p) = F(q)$. 
(b),(c) The two possible orderings if $F(p) \ne F(q)$.
(d) The impossible ordering if $F(p) \ne F(q)$.  
}
\label{fig:correct_orderings}
\end{figure}

\subparagraph*{Shortest Paths To/From Edges.}
\label{sec:shortest-paths}

As basic tools, we need \reviewerchange{linear-time algorithm}s to find  shortest paths from a  given  point to all edges of the polygon, and to find shortest paths from a  given  edge to all  vertices of the polygon.
See 
Appendix~\ref{appendix:shortest-paths}.

\subparagraph*{Separators and Funnels.}
\label{sec:separators}

A geodesic path between two vertices of $P$ separates $\partial P$ into two parts, and when we focus on which vertices/edges are in opposite parts,
we call the geodesic path 
a ``separator''.
Separators, first introduced by Suri~\cite{suri1989}, are a
main tool for finding all farthest vertices in a polygon. 
In 
Appendix~\ref{appendix:separators} 
we extend the %
basic properties of separators to the case of farthest \emph{edges} and prove:
(1) 
If vertex $v$ and edge $e$
\reviewerchange{($e$ need not be farthest from $v$)}
are separated by a geodesic path $\pi(a,b)$, then the shortest path from $v$ to $e$ is contained, except for one edge, in the shortest path trees of $a$ and $b$;
(2) A constant number of separators suffice to separate every vertex from its farthest edge.

\subsection{Chord Oracles and Coarse Covers}\label{section:chord_oracle}

In this section we 
describe the chord oracle results that we need from previous work, 
and we give a unified explanation of those algorithms and our current algorithm in terms of coarse covers.
The basic function of a chord  oracle is to decide, given a chord $K$, whether the center lies to the left or right (or on) the chord.
Pollack et al.~\cite{pollack_sharir} gave a \reviewerchange{linear-time} chord oracle for the geodesic vertex center, 
which is at the heart of all further geodesic center algorithms.
\reviewerchange{We extended the chord oracle to the case of the geodesic \emph{edge} center~\cite{lubiw2021visibility}.}

In both cases, 
a main 
step is 
the ``one-dimension down''
problem of 
finding 
the \defn{relative center}, which is a point $c_K$ on $K$ that minimizes the 
geodesic radius function $r(x)$.
The directions of the first segments of the paths from $c_K$ to its farthest sites determine
whether the center of $P$ lies 
left/right/on $K$
(see 
Appendix~\ref{appendix:oracles}).

Algorithms to find the relative 
center of a chord 
or the 
center of a polygon
rely on a basic convexity property of the geodesic radius function (see
Lemma~\ref{lem:geodesically-convex} in Appendix~\ref{appendix:oracles}), 
and all
follow the same pattern, which can be 
formalized via the concept of a \emph{coarse cover} of the chord/polygon. 
The idea is that 
a \emph{coarse cover} for a domain (a chord/polygon)  
is a set of elementary regions $R$ (intervals/triangles) covering the domain,
where each region $R$ has 
an associated 
easy-to-compute convex function $f_R$, 
such that the upper envelope of the $f_R$'s is the geodesic radius function.
We give a precise definition 
for the case of farthest edges
(following~\cite{lubiw2021visibility} and specialized for our Assumptions~\ref{assumptions}).

\begin{definition}
\label{defn:coarse-cover}
A \defn{coarse cover} of chord $K$ [or polygon $P$] is  
a set of 
triples $(R,f,e)$
where 
\begin{enumerate}
\item 
$R$ is a subinterval of $K$ [or a triangle of $P$], $f$ is a function defined on domain $R$, and $e$ is an edge of $P$.

\item For all $x \in R$, $f(x) = d(x,e)$ and
either: $f(x) = d_2(x,v) + \kappa$ where $d_2$ is Euclidean distance, $\kappa$ is a constant
and
$v$ is 
a vertex of $P$; or
$f(x) = d_2(x, {\bar e})$, where $d_2$ is Euclidean distance and $\bar e$ is the line through $e$.

\item 
For any point $x \in K$ [or $P$], and any edge $e$ that is farthest from $x$, there is a triple $(R,f,e)$ in the coarse cover with $x \in R$.

\end{enumerate}
\end{definition}

Condition (3) 
implies that the upper envelope of the functions of the coarse cover is the geodesic radius function. Thus the
[relative] center problem breaks into two subproblems:
(1) find a coarse cover; and (2)  find the point $x$ that minimizes the upper envelope of the coarse cover functions.
The high-level idea for solving step (2) in linear time  (for a chord or polygon domain) is to 
recursively reduce the domain (the search space) to a subinterval or subpolygon while eliminating elements 
of the coarse cover whose functions are strictly dominated by others.
As for step (1)---constructing a coarse cover---see Section~\ref{section:Phase-I} for a chord and Section~\ref{section:Phase-II} for a polygon.

We call the chord oracle in Phase II when we use divide-and-conquer to search for the center in successively smaller subpolygons.  We actually need two variations of the basic chord oracle.  
First, we need a \emph{geodesic oracle} that tests which side of a geodesic contains the center.  Secondly, we 
do not construct a coarse cover of a chord/geodesic from scratch; rather, we intersect the triangles of the coarse cover of the subpolygon with the chord/geodesic, 
thus avoiding runtime dependence on $n$.
These variations are described in
Appendix~\ref{appendix:oracles}.

\section{Phase I: Finding the \reviewerchange{Farthest-Edge} Voronoi Diagram Restricted to the Polygon Boundary}
\label{section:Phase-I}

\changed{Based on Assumptions~\ref{assumptions}, the boundary of $P$ consists of chains with a single farthest edge, separated by %
points (not vertices) that have two farthest edges
\reviewerchange{(see Figure~\ref{fig:coarse-cover})}. 
Our goal is to find these points.}
The first step of the algorithm is to find the farthest edge from each vertex of the polygon in linear time.  To do this, we extend the algorithm of Hershberger and Suri~\cite{hershberger1997matrix} that finds the farthest \emph{vertex} from each vertex.  Details are in
Appendix~\ref{appendix:Hershberger-Suri}.
The next step is to fill in the Voronoi diagram along the polygon edges.
For an edge $ab$ where
vertices $a$ and $b$ have the same farthest edge, i.e., $F(a) = F(b)$,  all points on the edge $ab$ have the same farthest edge, by
the Ordering Property. 
An edge $ab$ with 
$F(a) \ne F(b)$ is a \defn{transition edge}. 
We will find the \reviewerchange{farthest-edge} Voronoi diagram on one transition edge in linear time.  To handle \emph{all} the transition edges in linear time, we will show that for each transition edge $ab$ we can restrict our attention to 
\reviewerchange{the \defn{hourglass} $H(a,b)$ which is
the subpolygon of $P$ bounded by $ab$, $\pi(a,F(b))$, $\pi(b,F(a))$ and the %
portion of $\partial P$ between the terminals $t(a,F(b))$ and $t(b,F(a))$.}
In 
Appendix~\ref{appendix:hourglasses}
we show that the hourglasses of all transition edges can be found in linear time and that the sum of their sizes is linear.

In this section, we show how to construct the 
\reviewerchange{farthest-edge} Voronoi diagram along one polygon edge $ab$ 
in time linear in the size of the polygon.
We do not assume that the polygon is an hourglass.
For purposes of description, imagine $ab$ horizontal with $a$ at the left, and the polygon interior above $ab$.
We use the \defn{coarse cover} (Definition~\ref{defn:coarse-cover}) of the edge $ab$,
which can be found in linear time (Lubiw and Naredla~\cite{lubiw_et_al:LIPIcs.ESA.2021.65}). 
Elements of the coarse cover are triples $(I,f,e)$ where $I$ is a subinterval of $ab$ and $f(x) = d(x,e)$ for any $x \in I$. 
By resolving overlaps of coarse cover intervals $I$, we find
the upper envelope of the coarse cover functions $f$, which immediately gives the Voronoi diagram on $ab$. 
This is easy if we  
sort the endpoints of the intervals $I$, but 
we cannot afford to sort. Instead, we will insert the coarse cover elements %
one by one, maintaining a list $M$ of [pairwise internally] disjoint subintervals of $ab$ together with an associated distance function $f_M(x)$. 
\changed{An efficient insertion order depends on the fact that}
elements of the coarse cover of edge $ab$ are associated with edges of the shortest path trees ${T}_a$ and ${T}_b$ (that consist of the shortest paths from $a$ and $b$, respectively, to all the edges of $P$). 
We will use the ordering of the trees as embedded in the plane.

Oh, Barba, Ahn~\cite{oh2020geodesic} gave a \reviewerchange{linear-time algorithm} to find the farthest \emph{vertex} Voronoi diagram on the boundary of $P$. 
\changed{The approach is similar, but they add coarse cover elements by iterating over the sites (the vertices in their case), which involves a complicated algorithm to sweep back and forth along $M$ maintaining a shortest path to the current vertex, and a tricky amortized analysis (see~\cite[Lemma 7]{oh2020geodesic}).
Our approach is simpler and more general.}

\subsection{\reviewerchange{Farthest-edge} Voronoi Diagram on One Edge}
\label{sec:one-edge-Vor}

\reviewerchange{In previous work~\cite{lubiw2021visibility, lubiw_et_al:LIPIcs.ESA.2021.65}
we constructed}
a coarse cover
(see Definition~\ref{defn:coarse-cover})
of an edge $ab$ from the shortest path trees $T_a$ and $T_b$.  
\reviewerchange{The trees are first augmented}
with $0$-length edges so that the paths to every 
polygon edge $e$ end with  
a tree edge perpendicular to $e$. In particular, every polygon edge corresponds to a leaf in each tree.

Direct edges of $T_a$ and $T_b$ away from their roots. 
Each edge $uv$ of $T_a \setminus T_b$ with $u \ne a$
corresponds to an \defn{$a$-side} coarse cover element $(I,f,e)$ where $e$ corresponds to the farthest leaf of $T_a$ descended from $v$.
For example, in Figure~\ref{fig:example_2_e4}, see edge $a_3$ of $T_a$ and interval $I_{a_3}$.
There are symmetrically defined \defn{$b$-side} coarse cover elements.
Each edge $uv$ of $T_a \cap T_b$ with $u$ visible from $ab$ corresponds to a \defn{central triangle} coarse cover element $(I,f,e)$ where $e$ corresponds to the farthest leaf of $T_a$ descended from $v$.
For example, see edge $a_5=b_5$ and interval $I_{a_5}$.
Each polygon edge $e$ that has an interior point visible from $ab$ corresponds to a \defn{central trapezoid} coarse cover element $(I,f,e)$ where $I$ consists of the points on $ab$ whose shortest paths to $e$ arrive perpendicularly.
For example, see edge $e_4$ and interval $I_{e_4}$.

\begin{lem} 
(proved in
Appendix~\ref{appendix:one-edge-Vor})
\label{lem:consecutive-coarse-cover}
\label{cor:consecutive-coarse-cover}
For any edge $e$ of $P$, let $C(e)$ be the set of coarse cover elements $(I,f,e)$ for $e$.
If $C(e)$ is nonempty, then its elements correspond to a (possibly empty) path in $T_a$ directed towards a leaf, followed by a central triangle or trapezoid, followed by a (possibly empty) path in $T_b$ directed towards the root.  Furthermore, the corresponding intervals on $ab$ appear in order, are [internally] disjoint, and their union is an interval.
\end{lem}

\changed{We next construct a single tree $T$ whose edges correspond to coarse cover elements of $ab$.
Then we incrementally construct the \reviewerchange{farthest-edge} Voronoi diagram on $ab$ by adding coarse cover elements in a depth first search (DFS) order of $T$.
}

\subparagraph*{Constructing 
tree $T$.}
Starting with $T_a$,
attach an edge for each central trapezoid element to the associated leaf vertex of $T_a$; add the path of $b$-side triangle elements for each polygon edge $e$ after the central triangle or trapezoid for $e$; and contract original edges of $T_a$ that are not associated with coarse cover elements.
See Figure~\ref{fig:example_2_e4}(right).
We give more detail of these steps in 
Appendix~\ref{appendix:one-edge-Vor}.
The resulting tree $T$ can be constructed in linear time and its edges are in one-to-one correspondence with the coarse cover elements. 

\begin{obs}
\label{obs:consecutive-coarse-cover-T}
If $uv$ and $vw$ are edges of $T$, then the
corresponding coarse cover intervals
$I_1$ and $I_2$
appear in that order along $ab$ and intersect in a single point. 
\end{obs}

\begin{figure}
    \centering
    \begin{subfigure}[t]{0.43\textwidth}
        \centering
        \includegraphics[width=\textwidth]{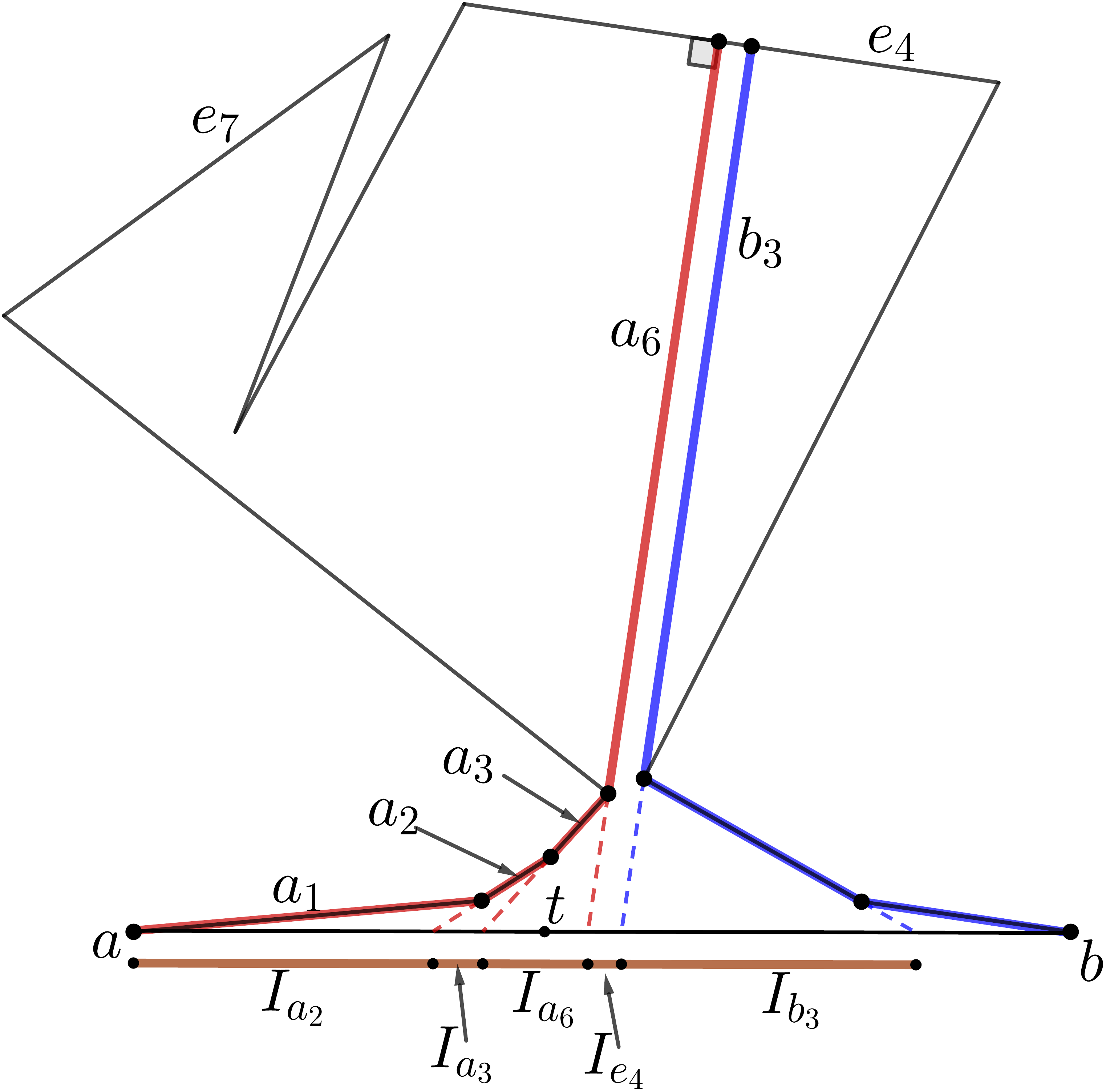}
    \end{subfigure}%
    \begin{subfigure}[t]{0.43\textwidth}
        \centering
        \includegraphics[width=\textwidth]{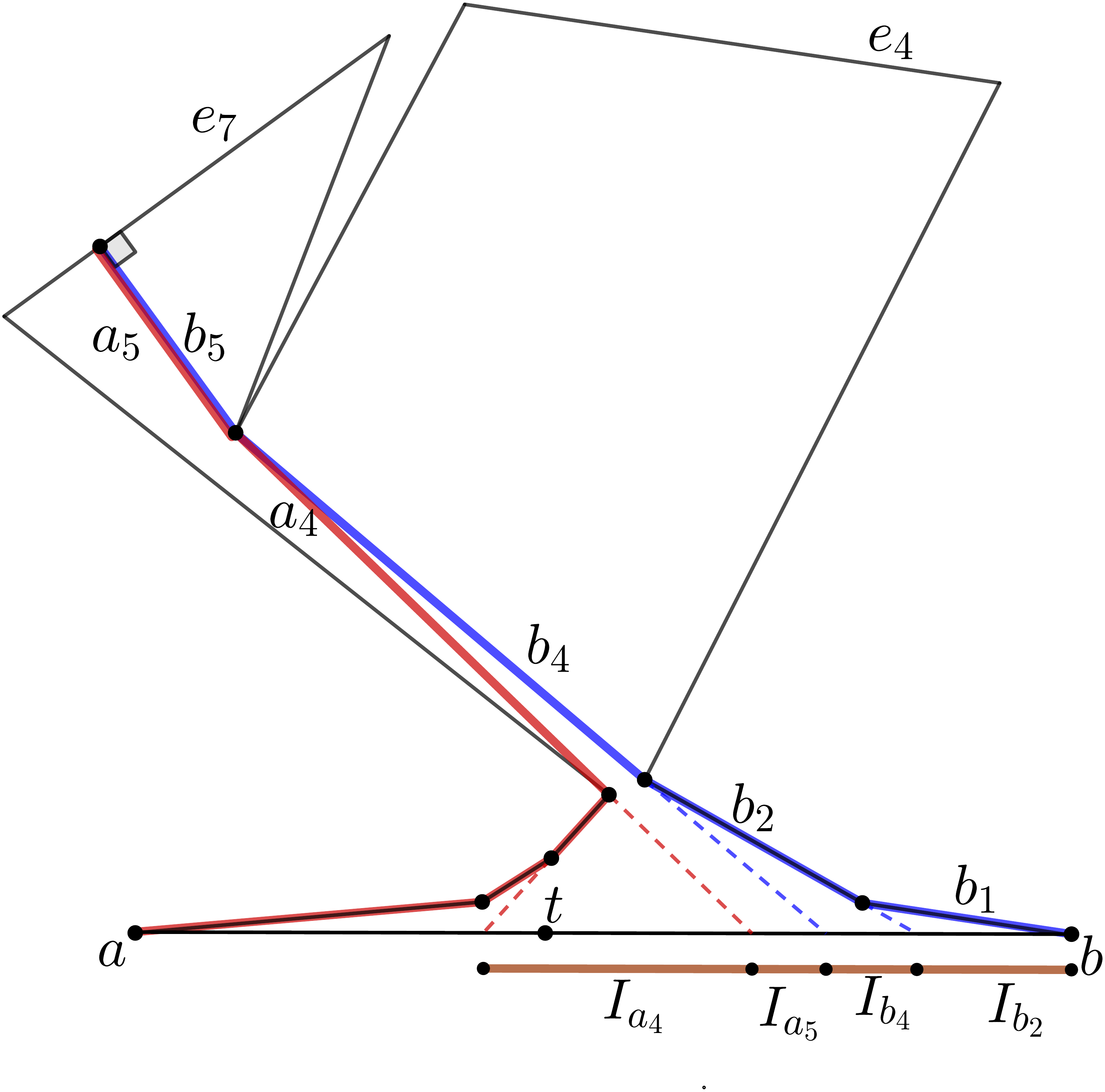}
    \end{subfigure}
    \begin{subfigure}[t]{0.13\textwidth}
        \centering
        \includegraphics[width=\textwidth]{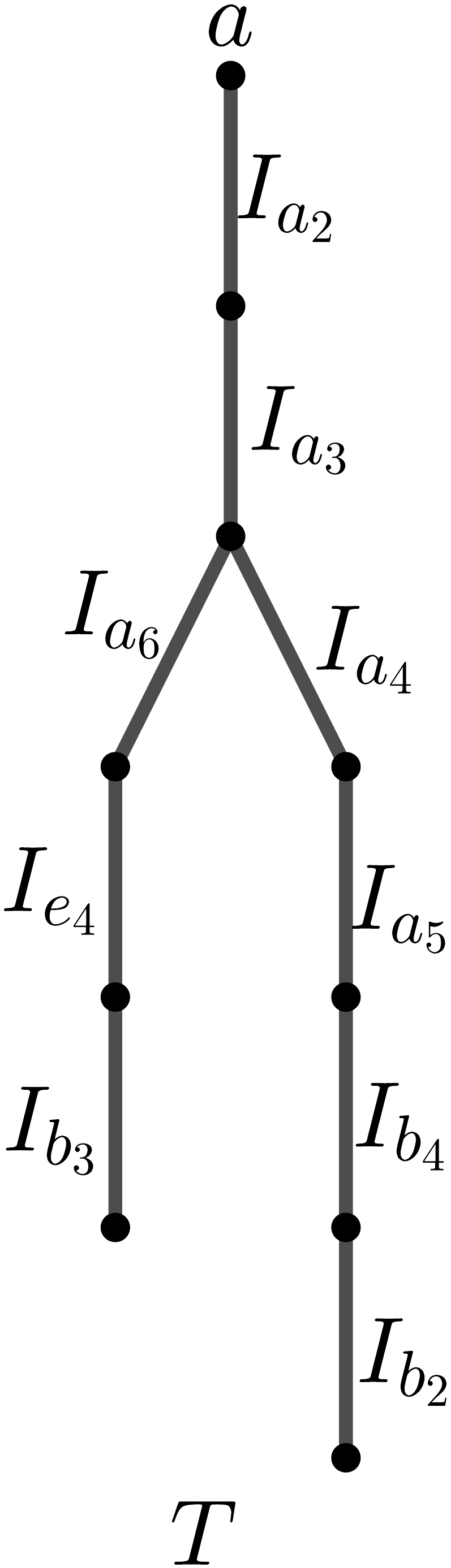}
    \end{subfigure}
    \caption{Coarse cover elements corresponding to some (not all) edges of $T_a$ (red) and $T_b$ (blue):  (left) coarse cover elements for $e_4$; 
(middle) coarse cover elements for $e_7$;
(right) the corresponding part of tree $T$.
\reviewerchange{When $I_{a_6}$ is handled by {\tt Insert} it wins the comparison with $I_{a_4}$ so it replaces $I_{a_4}$ up to the cross-over point $t$, and the algorithm discards the rest of $I_{a_6}$, together with $I_{e_4}$, and $I_{b_3}$.
}
}
\label{fig:example_2_e4}
\end{figure}

\subparagraph*{DFS Algorithm for the Voronoi Diagram.}
We add the coarse cover elements following a DFS of $T$ with children of a node in clockwise order. 
We maintain a list $M$ of interior disjoint subintervals of $ab$ whose union  is an interval starting at $a$. Each subinterval in $M$ records the coarse cover element it came from. 
Define $f_M$ to be the distance function determined by the intervals of $M$.
\changed{Initially, $M$ is the single point $a$, and $f_M$ is $- \infty$.}
At the end $M$ will be the upper envelope of the coarse cover functions (though this property is not guaranteed throughout). To handle edge $uv$ of $T$ with associated coarse cover element $(I,f,e)$, we compare $f$ to $f_M$ beginning at the left endpoint of $I$. We maintain a pointer $p_u$ that gives an interval of $M$ containing this endpoint. 
The recursive routine 
{\tt Insert}$(u,p_u)$ inserts into $M$ the portions of coarse cover elements that are associated with $u$'s subtree and that define the upper envelope.  
At the top level, we call {\tt Insert}$(a,p_a)$, where $p_a$ points to $a$.

\begin{algorithm}[H]
{\tt Insert}$(u, p_u)$ \ \ \ \# $u$ is a node of $T$ and $p_u$ is a pointer to an interval of $M$\\
\ \ \ \ {\bf for} each child $v$ of $u$ in clockwise order {\bf do}\\
\ \ \ \ \ \ $(I,f,e) :=$ the coarse cover element associated with the edge $uv$ of $T$\\
\ \ \ \ \ \ $l :=$ left endpoint of $I$; \ $r :=$ right endpoint of $I$\\
\ \ \ \ \ \ {\bf Invariant:} $p_u$ points to an interval of $M$ that contains $l$\\
\ \ \ \ \ \ {\bf if} $f(l^+) > f_M(l^+)$ where $l^+$ is just to the right of $l$  {\bf then}\\
\ \ \ \ \ \ \ \ replace intervals of $M$ starting at $p_u$ with a subinterval of $I$ ending at\\
\ \ \ \ \ \ \ \ \ \ the ``cross-over'' point $t < r$ where $f_M$ starts to dominate $f$, or at $r$\\
\ \ \ \ \ \ \ \  {\bf if} $f$ dominates until $r$ and $v$ is not a leaf of $T$ {\bf then}\\
\ \ \ \ \ \ \ \ \ \ \ {\bf call} {\tt Insert}$(v,p_v)$, where $p_v$ is a pointer to interval $I$ in $M$
\end{algorithm}

\subparagraph*{Runtime:}
Each edge of ${T}$ is handled once, and causes at most one new interval to be inserted into $M$, so the total number of endpoints inserted into $M$ is $O(n)$.
We can access $f_M(l^+)$ in constant time using the pointer $p_u$. 
Then 
the endpoints of intervals of $M$ that we traverse as we do the insertion vanish from $M$. 
Thus the runtime is $O(n)$.

\subparagraph*{Correctness:}
\changed{The following lemma implies that the final 
$M$ is the upper envelope of the coarse cover functions.}
\begin{lem}
\label{lem:DFS-correct}
The algorithm only discards pieces of coarse cover elements that do not form part of the final upper envelope.
\end{lem}

\begin{proof}
We examine the behaviour of the algorithm for edge $uv$ of $T$ with associated coarse cover element $(I,f,e)$, where $I=[l,r]$.   
We insert the subinterval $[l,t]$ into $M$ (or no subinterval). 
Because  $f(x) \ge f_M(x)$ for $x \in [l,t]$, any subintervals of $M$ that are removed due to the insertion do not determine the upper envelope, so their removal is correct. 

If we insert all of interval $I$ into $M$ and recursively call Insert$(v,p_v)$, then this is correct by induction.  
So suppose we insert a proper subinterval of $I$ or none of $I$.  
We must prove that no later part of $I$, and no element of the coarse cover associated with edges of the subtree rooted at $v$ determines the upper envelope. 
Let $t^+$ be a point just to the right of $t$ (or just to the right of $l$ if we insert no part of $I$).  Then $f_M(t^+) > f(t^+)$.
Number the polygon edges $e_1, e_2, \ldots, e_m$ clockwise from $a$ to $b$.  
Suppose that $e = e_i$, so $f(x) = d(x,e_i)$ for $x \in I$, in particular, $f(t^+) = d(t^+,e_i)$.
Suppose that 
$f_M(t^+) = d(t^+,e_k)$. 
Then $d(t^+,e_k) > d(t^+, e_i)$.  

Now consider the edges of $T$ descended from $v$ plus the edge $uv$.  Consider the corresponding coarse cover elements, $C_v$, and let $e_j$ be 
any polygon edge associated with any element in $C_v$. 
Note that the intervals on $ab$ associated with coarse cover elements of $C_v$ lie to the right of $r$, except for $I$ associated with $uv$.
We will prove that for any point $x \in ab$ to the right of $t^+$, $d(x,e_k) > d(x,e_j)$, which implies that none of the coarse cover elements in $C_v$ determines the upper envelope, nor does any part of $I$ to the right of $t$. Thus the algorithm is correct to discard them.

\changed{We first prove the result for $x = t^+$.}
If $uv$ corresponds to a central triangle/trapezoid for $e_i$
or a $b$-side triangle, then $T$ has a single path descending from $v$, all of whose edges are associated with $e_i$, i.e., $j=i$.  Otherwise, 
\changed{by the definition of
an $a$-side coarse cover element, $e_i$ corresponds to the farthest leaf of $T_a$ descended from $v$, which implies that  
}
$d(x,e_i) \ge d(x,e_j)$ for all $x \in I$, and in particular for $x=t^+$. 
Thus, in either case we have 
$d(t^+,e_k) > d(t^+, e_i) \ge d(t^+,e_j)$.

We next claim that $k < j$.  
The current $f_M$ values arise from tree edges already processed. 
These consist of: (1) edges on the path from $a$ to $u$; and (2)  edges of $T$ counterclockwise from this path.  
Edges on the path from $a$ to $u$ have coarse cover intervals on $ab$ to the left of $l$, by Observation~\ref{obs:consecutive-coarse-cover-T}.
Thus type (1) edges do not determine $f_M(t^+)$. 
By the depth-first-search order, type (2) edges have coarse cover elements corresponding to polygon edges counterclockwise from $e_j$.  Thus $k<j$.

To complete the proof of Lemma~\ref{lem:DFS-correct},
consider any point $x \in ab$ to the right of $t^+$. The clockwise ordering around the polygon boundary is $x,t^+, e_k, e_j$, so by
Lemma~\ref{lem:new_order}
and the fact that  $d(t^+,e_k) > d(t^+,e_j)$, we get $d(x,e_k) >  d(x,e_j)$, as required.  
\end{proof}

\section{Phase II: Finding the Geodesic Edge Center}
\label{section:final_edge_center}
\label{section:PhaseII}
\label{section:Phase-II}

The first step of Phase II is 
to construct a 
\defn{coarse cover} (Definition~\ref{defn:coarse-cover}) 
of the polygon
in linear time.
As shown in Figure~\ref{fig:coarse-cover} the
\reviewerchange{\defn{funnel} $Y(e)$ that consists}
of shortest paths between a chain on $\partial P$ with farthest edge $e$ and $e$ itself can be partitioned into its shortest path map.  If the result includes trapezoids, we partition each one into two triangles\footnote{Thus our triangles are not necessarily ``apexed'' triangles as in~\cite{linear_time_geodesic}.}. 
Each triangle is bounded by two polygon chords and a segment of a polygon edge, and the distance to $e$ has the form required by Definition~\ref{defn:coarse-cover}
(see 
the 
Appendix~\ref{appendix:coarse-cover} 
for details).
We seek the point inside $P$ that minimizes the upper envelope of the functions of the coarse cover. 
\changed{Note that Phase I can detect if the edge center lies \emph{on} $\partial P$ so we may assume that the center is interior to $P$.}

\begin{figure}
\centering
\includegraphics[trim={0 0cm 0 0}, width=.5\textwidth]{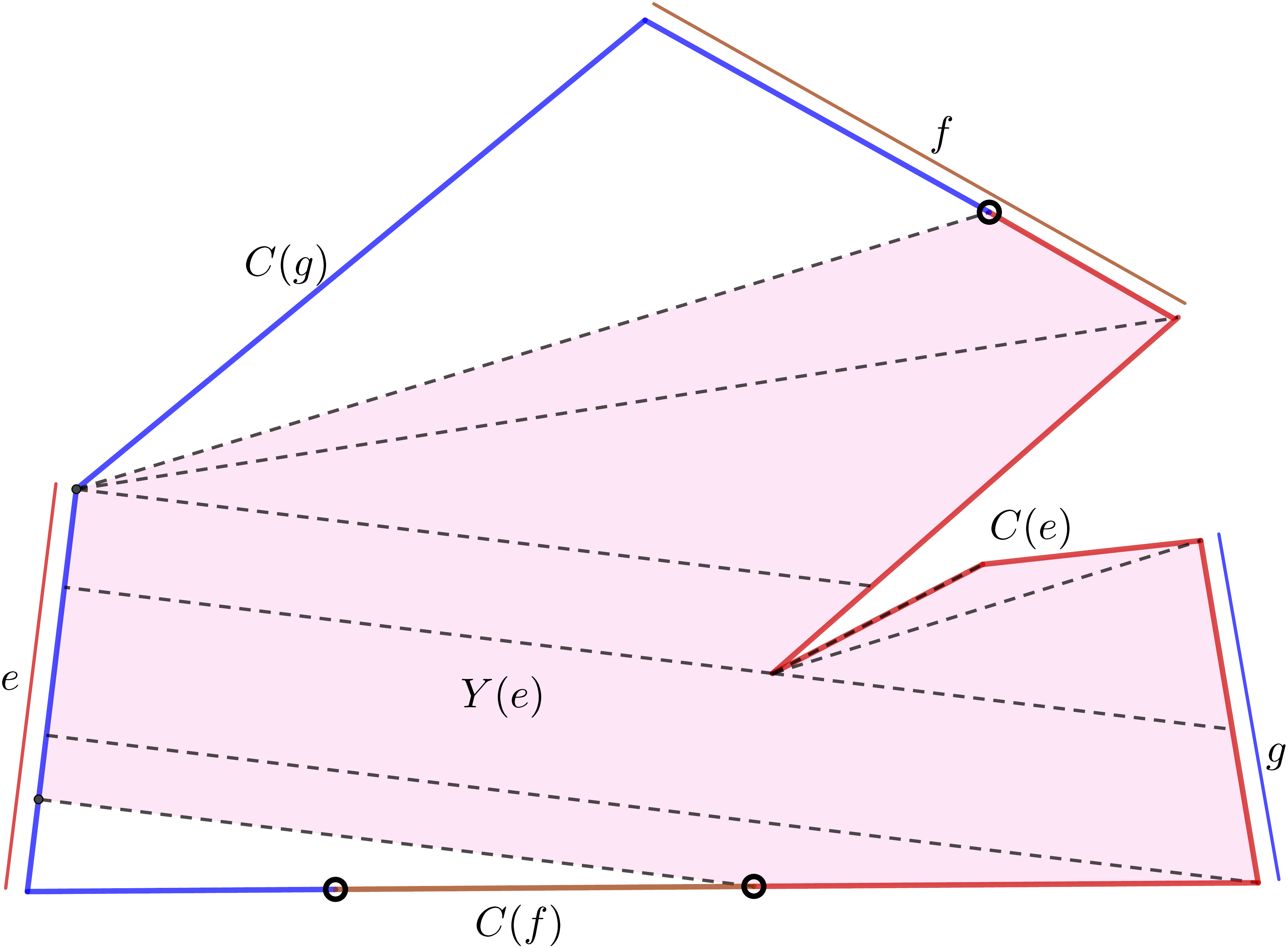}
\caption{
\reviewerchange{The farthest-edge Voronoi diagram restricted to the polygon boundary consists of chains $C(e),C(f),C(g)$ 
farthest from edges $e,f,g$, respectively.
To construct the coarse cover,
the funnel 
$Y(e)$ (shaded) is partitioned (by dashed segments) into a shortest path map from $e$ to $C(e)$.}
}
\label{fig:voronoi_split}
\label{fig:coarse-cover}
\label{fig:flower}
\end{figure}

Our final divide-and-conquer algorithm follows the vertex center algorithm of Ahn et al.~\cite{linear_time_geodesic}, generalized to farthest edges, and repairing flaws in their approach.
At each step of the algorithm we have a subpolygon $Q$ whose interior contains the center together with the coarse cover elements  
\newchanged{needed to compute the edge center}
and we shrink the subpolygon and eliminate a constant fraction of the 
coarse cover.
Each recursive step takes time linear 
in the size of the subproblem (the size of $Q$ plus the size of its coarse cover).
The subpolygons we work with \reviewerchange{are
\defn{$3$-anchor hulls}} defined as follows 
\reviewerchange{(see Figure~\ref{fig:shatter}).}
An \defn{anchor} is a point inside $P$, or a subchain of $\partial P$.  
A \defn{$3$-anchor hull} is the geodesic convex hull of at most three anchors.
These are weakly simple in general, but 
we only recurse on \defn{simple} $3$-anchor hulls.

The algorithm has two stages.
In Stage 1
no triangle of the coarse cover contains $Q$ (this is true initially when $Q=P$), so every triangle
has a chord crossing $Q$ and we
use $\epsilon$-net techniques on the set of such chords to reduce to a smaller cell $Q'$ that is crossed by a fraction of the chords, and hence by a %
fraction of the coarse cover triangles.
Once $Q$ is contained in a triangle of the coarse cover
we show (see
Lemma~\ref{lem:floating_cells_constant_size} in Appendix~\ref{appendix:Q-bounds})
that 
\reviewerchange{the size of $Q$, denoted $|Q|$, is at most 6.}
In fact, we will exit Stage 1 as soon as 
\reviewerchange{$|Q| \le 6$.}
It is then easy to reduce $Q$ to a triangle. 
After that, we switch to Stage 2, where the convexity of $Q$ allows us to use 
a Megiddo-style prune-and-search technique
\newchanged{(as Ahn et al.~do)}
to recursively reduce the size of the subproblem. 
Stage 2 is deferred to
Appendix~\ref{section:constantQ}.

\subsection{Stage 1: Algorithm for Large $Q$}
\label{section:largeQ}

Consider a subproblem corresponding to a 
\newchanged{simple $3$-anchor hull}
$Q$ with $|Q|>6$. %
We give an algorithm that either finds the edge center or reduces to a subproblem with $|Q| \le 6$, which is handled by Stage 2. 
In Stage 1, 
 no triangle of the coarse cover
contains $Q$ 
(this is proved in 
Lemma~\ref{lem:floating_cells_constant_size} in Appendix~\ref{appendix:Q-bounds}), 
so each one has a chord crossing $Q$---we denote this set of 
chords by ${\cal K}(Q)$.

To apply $\epsilon$-net techniques
we define a
\newchanged{\defn{$3$-anchor range space} as follows.
The ground set is 
a set $\cal K$ of chords of $P$,
and for each $3$-anchor hull $H$ of $P$ 
there is a range ${\cal K}(H)$ 
consisting of all chords of 
$\cal K$ that \defn{cross} $H$.
Here a chord \defn{crosses} a set  if both %
open half-polygons 
\reviewerchange{of the chord} contain points of the set.
}

The algorithm finds a constant size $\epsilon$-net of the \newchanged{$3$-anchor range space} on ${\cal K}(Q)$, 
which is 
a 
set $N \subseteq {\cal K}(Q)$ such that any \newchanged{$3$-anchor hull} 
not intersected by a chord of $N$ is intersected by only a constant fraction of the chords of ${\cal K}(Q)$---this is the important property that allows us to discard a fraction of the chords.
The set of chords $N$ forms an arrangement that partitions $Q$ into cells.
We use the chord oracle to determine which cell contains the center. We then 
add geodesic paths to 
subdivide this cell into 
a constant number of
$3$-anchor hulls
and use a \defn{geodesic oracle}
(see 
Appendix~\ref{appendix:geodesic-oracle})
to 
\newchanged{find which $3$-anchor hull contains the center, and to shrink it to a simple $3$-anchor hull 
$Q'$.} The algorithm recurses on $Q'$, whose coarse cover is a fraction of the size.

More details of the algorithm can be found in
Appendix~\ref{appendix:Algorithm-Stage1}.  
For now, we expand on the aspects of the algorithm that differ from the approach of Ahn et al.~\cite{linear_time_geodesic}.
Instead of 
$3$-anchor hulls,
their algorithm works with \emph{$4$-cells}, 
\newchanged{formed by taking the intersection of at most four half-polygons, where a half-polygon is the part of $P$ to one side of a chord}.
The number (three versus four) is not significant, but we 
\newchanged{bound our regions by geodesics instead of chords} 
in order 
to obtain the following two results. 

\subparagraph{1.} The
\newchanged{$3$-anchor}
range space has finite VC-dimension.  This implies that constant-sized $\epsilon$-nets exist. Furthermore, there is a ``subspace oracle'' 
that allows us to find an $\epsilon$-net $N$ in deterministic linear time~\cite[Chapter 47, Theorem 47.4.3]{toth2017handbook}. For 
further 
background see
Appendix~\ref{appendix:epsilon-net-overview}.

Ahn et al.~claim that their range space 
\newchanged{(of chords crossing $4$-cells)}
has finite VC-dimension but their proof is flawed.
Our proof shows that their range space does in fact have finite VC-dimension.
\newchanged{They do not mention subspace oracles, without which their algorithm runs in expected linear time rather than deterministic linear time as claimed.}
We expand on these aspects in Section~\ref{section:epsilon-net-results} below.

\subparagraph{2.} A cell of the arrangement of $N$ can be partitioned into 
constantly many
$3$-anchor hulls.

The method used by Ahn et al.~to 
subdivide a cell of $N$ into $4$-cells by adding a constant number of chords is
incomplete, 
\reviewerchange{see Figure~\ref{fig:decomposition-1}.}
We see how to repair their partition step
but we find 
$3$-anchor hulls
more natural.

\begin{figure}[ht!]
\centering
\includegraphics[width=.15\textwidth]{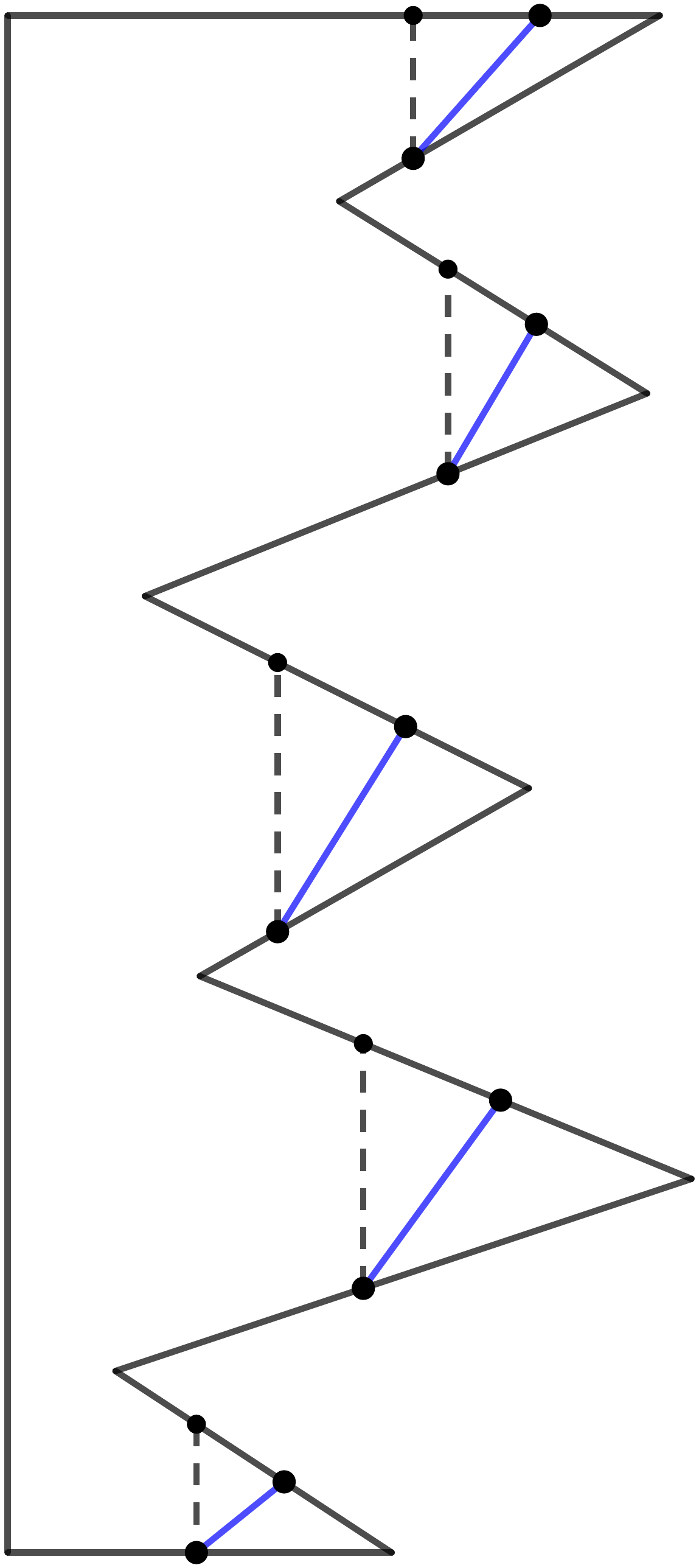}
\caption{
\reviewerchange{Ahn et al.~\cite{linear_time_geodesic} subdivide a cell of $N$ by adding 
vertical chords (dashed) at endpoints and intersection points of chords of $N$ (blue), which leaves a $5$-cell in this example.} 
}
\label{fig:decomposition-1}
\end{figure}

\subsection{$\epsilon$-Net Results for Stage 1}
\label{section:epsilon-net-results}

In this section we expand on the $\epsilon$-net results that are needed for Stage 1 of the algorithm as described above.  We also give details of the flaws in the approach of Ahn et al.~\cite{linear_time_geodesic}.
For an overview of $\epsilon$-nets as used for geometric divide-and-conquer, see
Appendix~\ref{appendix:epsilon-net-overview}.
To show that $\epsilon$-nets of
\reviewerchange{constant} size %
exist we need the following result.

\begin{lem}
\label{lem:constant_shattering_dimension}
The \newchanged{$3$-anchor}
range space has 
VC-dimension less than 259.
\end{lem}

Our proof of Lemma~\ref{lem:constant_shattering_dimension} works equally well for
\newchanged{$4$-anchor hulls}---the
bound becomes $372$. 
\newchanged{A $4$-cell is a special case of a $4$-anchor hull %
so}
our proof implies finite VC-dimension ($\le 372$) for the $4$-cell range space, which repairs the claim by Ahn et al.\footnote{In response to our enquiries, Eunjin Oh independently suggested a similar remedy.}.
We explain the flaw in 
their proof.
Let us refer to the set of chords intersecting a $4$-cell as a ``$4$-cell range''. 
Ahn et al.~prove that the $1$-cell range space has 
VC-dimension at most 65,535. 
They note that a $4$-cell is the intersection of four $1$-cells, and then claim in their Lemma 9.1 that
this implies finite VC-dimension for the $4$-cell range space.  As justification, they refer to 
Proposition 10.3.3 of Matousek's text~\cite{matousek2002}, which 
states 
that the VC-dimension is bounded for any family whose sets can be defined by a formula of Boolean connectives (union, intersection, set difference).
However, Matousek's proposition cannot be applied in this situation because, although a $4$-cell is the intersection of four $1$-cells,  
it is not true that 
a $4$-cell \emph{range} is the intersection of four $1$-cell \emph{ranges}.
In particular, a chord can intersect two $1$-cells, but not intersect the intersection of the two $1$-cells.  For example, a line of slope $-1$ can intersect the $+x$ half-plane and the $+y$ half-plane without intersecting the $+x,+y$ quadrant.

\begin{proof}[Proof of Lemma~\ref{lem:constant_shattering_dimension}]
We will prove that the shattering dimension is 6 and then apply the result that a range 
space with shattering dimension $d$ has VC-dimension bounded by $12d \ln{(6d)}$ (Lemma 5.14 from Har-Peled~\cite{har2011geometric}). 
For $d=6$, this is \reviewerchange{less than $259$}.

We must show that for
a set $\cal K$ of chords with 
$| {\cal K}| = m$, the number of distinct ranges is $O(m^6)$.
We prove that 
the range space for ${\cal K}$ is the same 
if we replace 
\newchanged{$3$-anchor hulls}
by   ``expanded $3$-anchor hulls'' that are defined in terms of $\cal K$, more precisely, 
in terms of the arrangement $A({\cal K})$
of the chords ${\cal K}$ plus the edges of $P$.
\newchanged{Define an \defn{expanded anchor} to be an internal face, edge, or vertex of $A({\cal K)}$, or a 
polygon chain
with endpoints in 
$V({\cal K})$, the set of endpoints of chords $\cal K$.
An \defn{expanded $3$-anchor hull} is the geodesic convex hull of at most three 
expanded anchors.}

\begin{figure}
  \centering
  \includegraphics[width=0.7\textwidth]{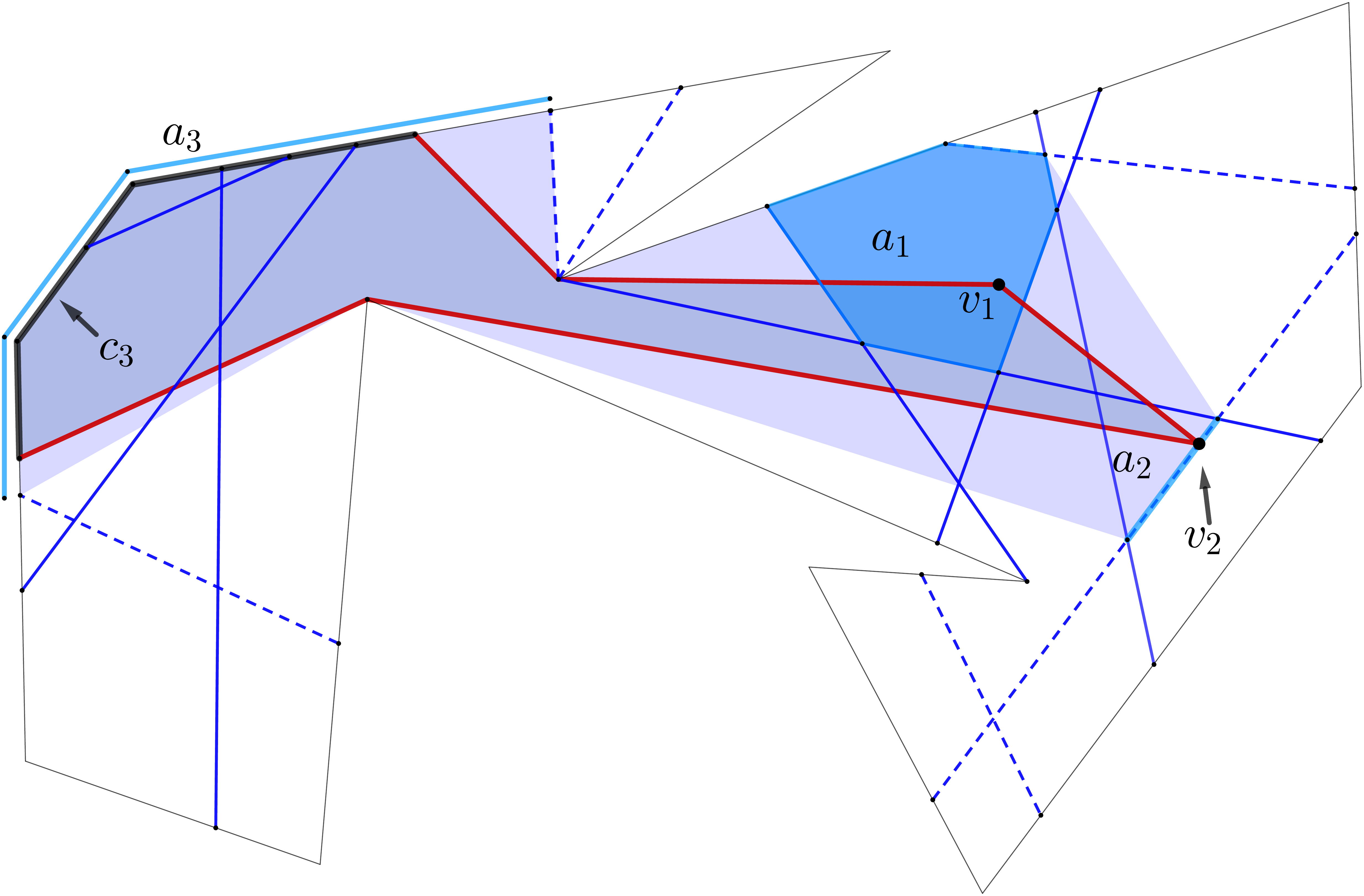}
  \caption{
  A simple 3-anchor hull $Q$ \reviewerchange{(outlined in red)} with anchors $v_1, v_2$ and the polygon chain $c_3$ \reviewerchange{(thick black)}. 
Solid blue chords cross $Q$, while dashed blue chords do not.  The expanded 3-anchor hull $\psi(Q)$ (lightly shaded)
has expanded anchors: $a_1$, the face in the chord arrangement containing $v_1$;
  $a_2$, the edge with $v_2$ in its interior; and
  $a_3$, the polygon chain extending $c_3$ to chord endpoints.
  The same  
  chords cross $Q$ and $\psi(Q)$.
}
\label{fig:shatter}
\end{figure}

\begin{lem}
\label{lem:anchor-ranges}
The set of ranges ${\cal R} = \{ {\cal K}(Q) \mid Q \text{ is 
a $3$-anchor hull}\,\}$ is the same as 
the set of ranges
$\overline{\cal R} = \{ {\cal K}({Q}) \mid {Q} \text{ is 
an expanded $3$-anchor hull}\,\}$.
\end{lem}

\begin{proof}
To prove ${\cal R} \subseteq \overline{\cal R}$, consider a 
\newchanged{$3$-anchor hull $Q$.
Replace any point anchor $p$}
by the 
smallest (by containment) internal vertex, edge, or face of $A({\cal K)}$ that contains $p$.
See Figure~\ref{fig:shatter}.
Replace any 
polygon chain anchor $C$  
by the smallest chain of 
$\partial P$ 
containing $C$ and with endpoints in $V({\cal K})$.
Let 
$\psi(Q)$ be the geodesic convex hull of these expanded anchors.  
Then $\psi(Q)$ is an expanded $3$-anchor hull that contains $Q$, and it is straight-forward to prove that 
${\cal K}(Q) = {\cal K}(\psi(Q))$
(see 
Claim~\ref{claim:same-crossing-chords} in Appendix~\ref{appendix:epsilon-nets}).

For the other direction, let $Q$ be an expanded $3$-anchor hull.
Replace an expanded anchor that is a face, edge, or vertex of $A({\cal K})$ by a point anchor in the interior of that face, edge, or vertex. An expanded anchor that is a polygon chain remains unchanged.
Let $\gamma(Q)$ be the
geodesic convex hull of the resulting anchors.
Observe that
\newchanged{$\gamma(Q)$ is a $3$-anchor hull and}
$\psi(\gamma(Q)) = Q$.
\reviewerchange{As above, this implies that}
${\cal K}(Q) = {\cal K}(\gamma(Q))$.
\end{proof}

\changed{
To complete the proof of Lemma~\ref{lem:constant_shattering_dimension} we claim that the number of expanded $3$-anchor hulls of ${\cal K}$ is $O(m^6)$.
An expanded anchor may be an
internal vertex, edge, 
\reviewerchange{or}
face
of $A(\cal K)$,  of which there are $O(m^2)$ possibilities.
Otherwise, an expanded anchor is 
a chain of $\partial P$ between vertices of $V({\cal K})$, 
also with
$O(m^2)$ possibilities.
Thus the number of expanded $3$-anchor hulls 
is $O((m^2)^3) = O(m^6)$.
}
\end{proof}

\paragraph*{Subspace Oracle}

\reviewerchange{To prove that the $3$-anchor range space has a subspace oracle, we present a deterministic algorithm that, given a subset ${\cal K}' \subseteq {\cal K}$ with $|{\cal K}'| = m$, computes the set
of ranges ${\cal R} = \{ {\cal K}'(Q) \mid Q \text{ is 
a $3$-anchor hull}\}$
in time $O(m^7)$.
The idea is 
to use Lemma~\ref{lem:anchor-ranges} and
to construct $A({\cal K}')$ minus the edges of $P$, 
and find, for each 
chord $K \in {\cal K}'$, which of the $O(m^2)$ expanded anchors in $A({\cal K}')$ \reviewerchange{intersects} each side of $K$, and then, for each of the $O(m^6)$ expanded $3$-anchor hulls, eliminate the chords that have all three expanded anchors to one side, leaving the chords that cross the hull. For further details, see Appendix~\ref{appendix:epsilon-nets}.
}

\bibliography{RS_refs}

\newpage

\appendix

\section{Extra Material for Section~\ref{section:prelims}, Preliminaries}

\subsection{\reviewerchange{Details on General Position Assumptions}}
\label{appendix:general-position}

We restate Assumptions~\ref{assumptions} in order to refer to the parts individually.

\begin{assumption}\label{assumption:no_three_points_collinear}
No three vertices of 
$P$ are collinear.
\end{assumption}

\begin{assumption}
\label{assumption:unique_farthest_neighbor}
After imposing the tie-breaking rule, no vertex is equidistant from two or more edges.
\end{assumption}

Figure~\ref{fig:twod_voronoi} shows the reason for Assumption~\ref{assumption:unique_farthest_neighbor}.

\begin{figure}[h!]
  \centering
  \includegraphics[width=.6\textwidth]{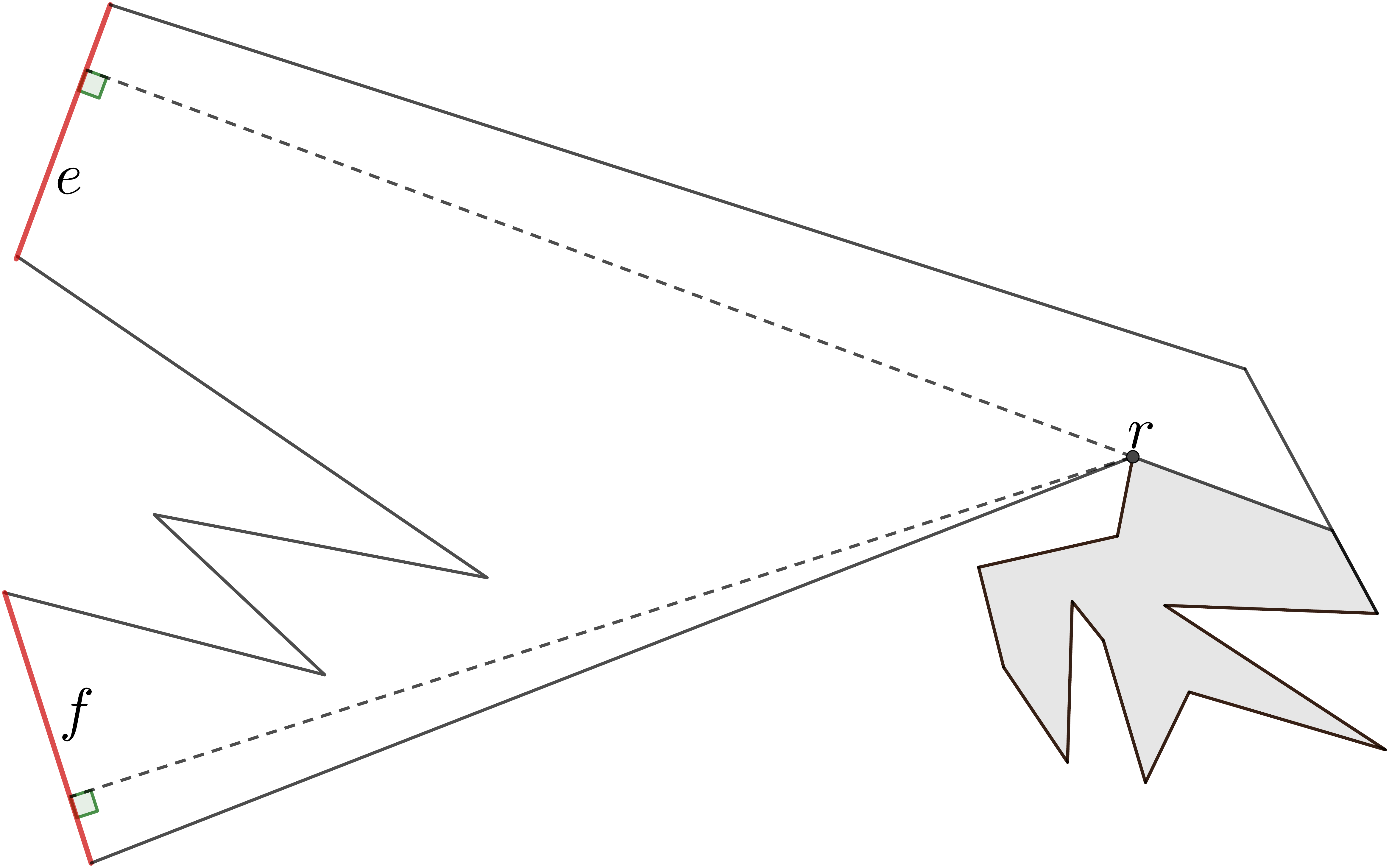}
  \caption{
  Vertex $r$ 
  \changed{is equidistant from edges $e$ and $f$ (colored red), and so is every point in the shaded region.
  Assumption~\ref{assumption:unique_farthest_neighbor} forbids this situation.
  }
  }
\label{fig:twod_voronoi}
\end{figure}

\begin{lem}
\label{lem:1D-Voronoi-edges}
Let $D$ be the  set of points in $P$ with more than one farthest edge (after imposing the tie-breaking rule).  Then $D$ does not contain a 2-dimensional ball, does not contain a vertex of $P$, and intersects $\partial P$ in isolated points. 
\end{lem}
\begin{proof}
It suffices to show that the conditions hold for the set of points equidistant from two edges $e$ and $f$.
Let $p$ be a point with $d(p,e) = d(p,f)$. 
The paths $\pi(p,e)$ and $\pi(p,f)$ do not share a vertex other than the terminal point, by Assumption~\ref{assumption:unique_farthest_neighbor}---in particular, $p$ cannot be a vertex.  

If $\pi(p,e)$ and $\pi(p,f)$ share 
a terminal vertex $u$, then the Tie-Breaking Rule would apply unless $p$ is on the bisector of the angle at $u$, which is 1-dimensional and intersects $\partial P$ in a single point because no three vertices are collinear by Assumption~\ref{assumption:no_three_points_collinear}.

Otherwise, the paths $\pi(p,e)$ and $\pi(p,f)$ diverge at $p$. 
Let $s_e$ be the first vertex on the path $\pi(p,e)$---or let $s_e=e$ in case there are no vertices.  Define $s_f$ similarly.  Then $p$ must be on the weighted bisector between $s_e$ and $s_f$, which is 1-dimensional, and and intersects $\partial P$ in isolated points.
\end{proof}

\begin{assumption}
\label{assumption:Voronoi-vertices}
No point on the polygon boundary has more than two farthest edges. No point in the interior of the polygon has more than a constant number of farthest edges.
\end{assumption}

We note that our
assumptions
can be effected by perturbing vertices, since, in the $2n$-dimensional space of allowed vertex perturbations, the configurations we must avoid are lower-dimensional.

\subsection{Details on Properties of Farthest Edges}
\label{appendix:farthest-properties}

In this section we give some  basic properties of shortest paths from points on $\partial P$ to their farthest edges in a polygon, with a focus on when and how such paths cross---more formally, we examine the ordering of the points and their farthest edges around the polygon boundary.

We first prove Lemma~\ref{lem:new_order}.
\begin{proof}[Proof of Lemma~\ref{lem:new_order}]

\begin{figure}
\centering
\includegraphics[trim={0 0cm 0 0cm}, width=.6\textwidth]{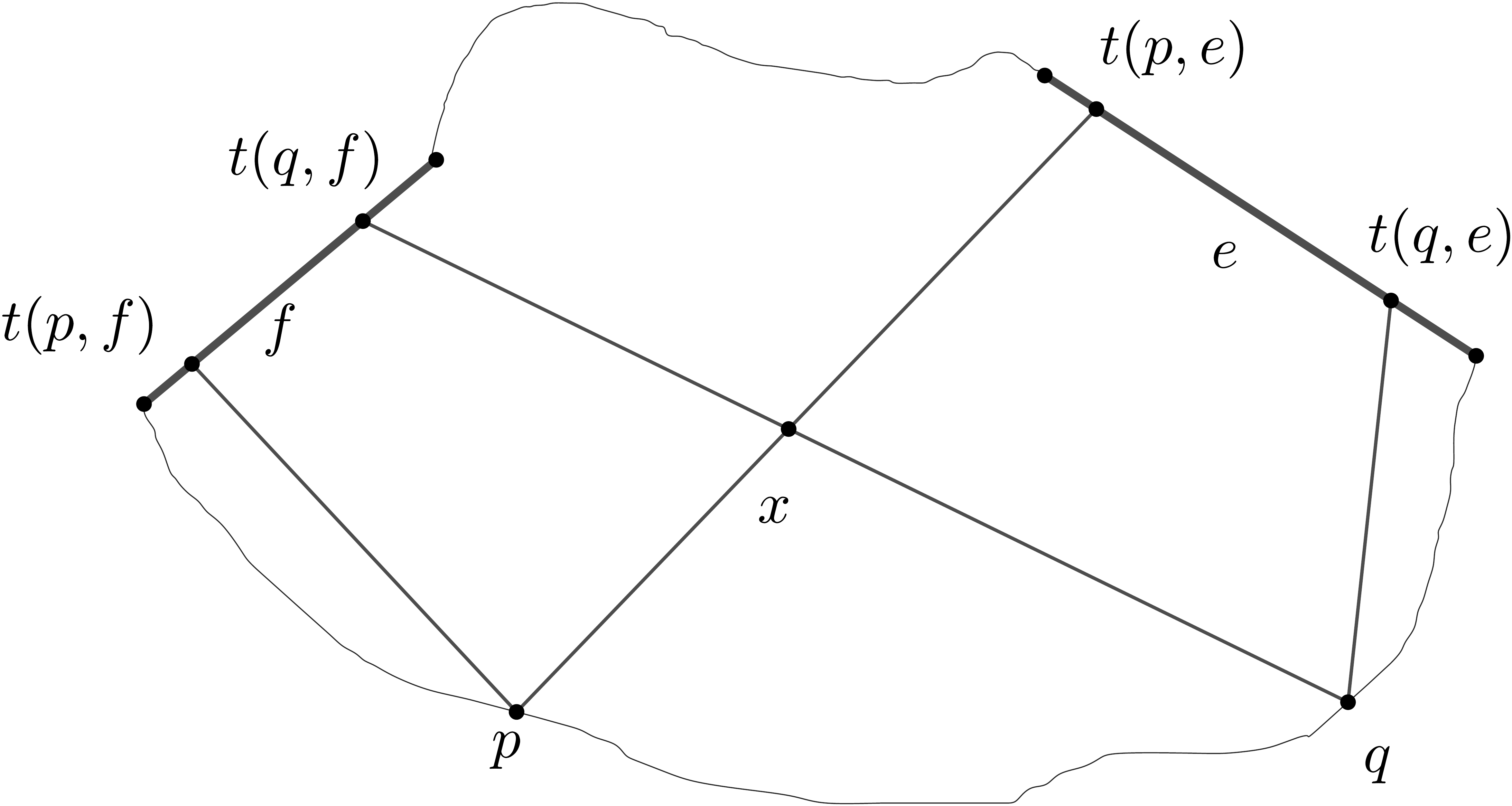}
\caption{Illustration for Lemma~\ref{lem:new_order}.
}
\label{fig:total_monotonicity_lemma}
\end{figure}
Suppose points
$p$, $q$ and edges $e$, $f$ occur in the order $p,q,e,f$ along the polygon boundary $\partial P$.
See Figure~\ref{fig:total_monotonicity_lemma}. 
We must prove that $d(p,e) + d(q,f) \geq d(p,f) + d(q,e)$.

Due to the ordering of $p,q,e,f$ on $\partial P$, the paths $\pi(p,e)$ and $\pi(q,f)$ must have a common point which we label $x$.   Then:
\begin{align*}
&d(p,e) + d(q,f) \\
  & = d(p,t(p,e)) + d(q,t(q,f)) \text{\quad [Distance to an edge is distance to the terminal]}\\
  & = d(p,x) + d(x,t(p,e)) + d(q,x) + d(x,t(q,f)) \\
  & = (d(p,x) + d(x,t(q,f))) + (d(q,x) + d(x,t(p,e))) \\
  & \geq d(p, t(q,f)) + d(q,t(p,e)) \text{\quad [Triangle Inequality]}\\
  & \geq d(p,t(p,f)) + d(q, t(q,e))  \text{\quad [Definition of a terminal]}\\
  & = d(p,f) + d(q,e) \text{\quad\quad\quad\quad\quad\ [Distance to an edge is distance to the terminal]}
\end{align*}
\end{proof}

We often use Lemma~\ref{lem:new_order} in the following form.

\begin{coro}
\label{cor:ordering-property-paths}
Under the same assumptions, if $d(p,f) > d(p,e)$, then $d(q,f) > d(q,e)$.
\end{coro}

We now return to farthest paths.
Let $p$ and $q$ be two  points on  the polygon boundary.
Let $F(p)$ be a farthest edge from  $p$ and let $F(q)$ be a farthest edge from $q$.
Note that these farthest edges need  not be unique since there are (isolated) points on the polygon boundary with two farthest edges.
Let $\pi_p = \pi(p,F(p))$ and let $\pi_q = \pi(q,F(q))$.
Let $t_p$ be the terminal point of the path $\pi_p$ and let $t_q$ be the terminal point of the path $\pi_q$.
If $F(p) \ne F(q)$, we 
 say that $\pi_p$ and $\pi_q$ \defn{cross} if the ordering around the polygon boundary (in either clockwise or counterclockwise order) is $p,q,F(p),F(q)$.  
Note that this allows the possibility that $t_p = t_q$ at a reflex vertex. 
If $F(p) = F(q)$, we say that $\pi_p$ and $\pi_q$ \defn{cross} if the ordering around the polygon boundary (in either clockwise or counterclockwise order) is $p,q,t_p,t_q$, with $t_p \ne t_q$.

\begin{lem}
\label{lem:ordering-properties}
With the above setup,
the paths $\pi_p$ and $\pi_q$ have the  following properties.

\leavevmode
\begin{enumerate}[label=\bf{(P\arabic*)}]

\setlength{\itemindent}{.15in}

\item \label{prop:path-to-same-edge}
If $F(p) = F(q)$, then 
$\pi_p$ and $\pi_q$ do not cross, 
i.e., 
the ordering of points around the boundary of $P$ is $p,q,t_q,t_p$, possibly with $t_q = t_p$ if the paths merge.
(See Figure~\ref{fig:correct_orderings}(a)).

\item \label{prop:anti-parallel}
If $F(p) \ne F(q)$
then the possible orderings are: $p,F(p), q,F(q)$ (see
Figure~\ref{fig:correct_orderings}(b));
or $p, q$, $F(p), F(q)$, i.e., the paths cross 
(see Figure~\ref{fig:correct_orderings}(c)).
Equivalently, 
the only other ordering, $p,q, F(q),F(p)$, 
cannot occur (see Figure~\ref{fig:correct_orderings}(d)).

\item 
\label{prop:ordering-property}
{\bf The Ordering Property.} 
As $p$ moves clockwise around $\partial  P$, so does $F(p)$.

\item 
\label{prop:no-shared-chord}

If the  paths $\pi_p$ and $\pi_q$ cross,
then they  
do  not share a directed 
 polygon chord. 
The paths may cross at a vertex or at internal points of chords.  They  may  share  a chord  in opposite directions.

\end{enumerate}
\end{lem}

The Ordering Property \ref{prop:ordering-property} was proved by Aronov et al.~\cite{aronov1993farthest} for the case of farthest vertices.
To prove Lemma~\ref{lem:ordering-properties}, we use
(as they did) the  ``triangle inequality'' Lemma~\ref{lem:new_order}.
Later on in our paper, we will appeal not only to the Ordering Property and the other parts of Lemma~\ref{lem:ordering-properties}, but also to the triangle inequality, since it  applies more generally to paths that do not go to farthest edges.
In fact, even for the basic problem of finding the farthest vertex from each vertex in a convex polygon, the triangle property is the key to \reviewerchange{linear-time algorithm}s.  The first such \reviewerchange{linear-time algorithm} was given by Aggarwal et al.~\cite{aggarwal1987geometric} using a technique called matrix searching in a totally monotone matrix.  They point out that assuming only the Ordering Property, there is a super-linear lower bound on the time to find all farthest vertices.  However, the triangle inequality implies a ``totally monotone'' matrix and allows a \reviewerchange{linear-time algorithm}.  The matrix searching technique is discussed further in 
Appendix~\ref{appendix:Hershberger-Suri}.

\begin{proof}[Proof of Lemma~\ref{lem:ordering-properties}]
\ 

\noindent
\ref{prop:path-to-same-edge}.
Suppose the ordering is $p,q,t_p,t_q$.  Then the paths must have a common point $x$. 
The shortest path from $x$ to the edge $F(p)=F(q)$ is unique, so the paths $\pi_p$ and $\pi_q$ are the same after $x$. 
Thus the ordering is $p,q,t_q,t_p$, possibly with $t_q=t_p$.

\smallskip\noindent
\ref{prop:anti-parallel}. 
We must show that 
the ordering $p,q, F(q),F(p)$ cannot occur.  Suppose  it does.  
First note that if $t_p = t_q$ then the tie-breaking rule would not allow the ordering $p,q,F(q),F(p)$.  Thus we may assume that $t_p \ne t_q$.

Since $F(p)$ is a farthest edge from $p$, $d(p,F(p))  \ge d(p,F(q))$.  Since $F(q)$ is a farthest edge from $q$, 
$d(q,F(q)) \ge d(q,F(p))$. 
By  Lemma~\ref{lem:new_order}, $d(p,F(q)) + d(q,F(p) \ge d(p,F(p)) + d(q,F(q))$.   
Therefore $d(p,F(p)) = d(p,F(q))$ and $d(q,F(q)) = d(q,F(p))$.

\begin{claim}
Both $p$ and $q$ have $F(p)$ and $F(q)$ as tied farthest edges.
\end{claim}
\begin{proof}
The distances are the same but we must be careful about the tie-breaking rule.  
If the tie-breaking rule applies to $\pi(p,F(p))$ and $\pi(p,F(q))$, then these two paths both terminate at a reflex vertex $u$ common to $F(p)$ and $F(q)$, but in this case $\pi(q,F(p))$ must also terminate at $u$ (since it cannot cross $\pi(p,F(p))$).
Then $\pi(q,F(q))$ also terminates at $u$, since we cannot have two equal-length paths from $q$ to different points on edge $F(q)$.  Thus the original paths $\pi_p$ and $\pi_q$ terminate at the same point, which we already ruled out.
\end{proof}

We claim that any  point $r$ that lies on $\partial P$ between $p$ and $q$ also has $F(p)$ and  $F(q)$ as farthest edges. 
Consider  any  edge $e$ that lies 
on the polygon chain from $r$ to $F(q)$ (the part containing $p$).
Note that $F(p)$ is one  such edge.  Applying Lemma~\ref{lem:new_order} to $q,r,e,F(q)$, gives $d(r,F(q)) + d(q,e) \ge d(r,e) + d(q,F(q))$.  Since $d(q,F(q)) \ge  d(q,e)$ this implies $d(r,F(q)) \ge d(r,e)$.  In particular, $d(r,F(q)) \ge d(r, F(p))$.  A symmetric argument shows that $d(r,F(p)) \ge d(r,e)$ for any  edge $e$ 
that lies on the polygon chain from $r$ to $F(p)$ (the part containing $q$).
In particular, $d(r,F(p)) \ge  d(r,F(q))$.  Since all edges $e$ are included in the  two ranges, this proves that $r$ has $F(p)$ and $F(q)$ tied for farthest edge.
(We can again show that the tie-breaking rule does not apply.)
By Lemma~\ref{lem:1D-Voronoi-edges},
only isolated points on $\partial  P$ can have 
$F(p)$ and $F(q)$ tied for farthest edge.
Therefore the ordering $p,q,F(p),F(q)$ cannot occur.

\smallskip\noindent
\ref{prop:ordering-property}. 
Consider a point $p$ with a farthest edge $F(p)$ and let $q$ be the first point after $p$ moving clockwise around $\partial P$ that has a farthest edge $F(q)$ 
that is not a farthest edge from $p$. 
Note that $q$ comes before $F(p)$.  Since the ordering $p,q, F(q), F(p)$ is prohibited, $F(q)$ must lie after $F(p)$ in clockwise order.

\smallskip\noindent
\ref{prop:no-shared-chord}.  Suppose $\pi_p$ and $\pi_q$ cross.
We first suppose that
$t_p = t_q$. Then the terminal point is a reflex vertex $u$ common to $F(p)$ and $F(q)$.  If the paths share a directed chord $(a,b)$, then the paths are identical after vertex $a$ and therefore identical on their last segment which is a chord from some vertex $v$ to $u$.  
The tie-breaking rule would not allow $p$ and $q$ to have farthest edges $F(p)$ and $F(q)$ unless $v$ lies on the bisector of $u$
which is excluded by Assumption~\ref{assumption:unique_farthest_neighbor}. 

Thus we may assume that
$t_p \ne t_q$ and the paths share the directed 
chord $(a,b)$.  
Consider the portion of $\pi_p$ from $b$ to  $F(p)$ and the  portion of $\pi_q$ from $b$ to $F(q)$.   Those are both shortest paths. 
Because of the \reviewerchange{general position} Assumption~\ref{assumption:unique_farthest_neighbor} (note that $b$ is a vertex), one of the paths must be longer, say the one to $F(q)$.  Now we claim that $d(p,F(q)) > d(p,F(p))$, contradiction to $F(p)$ being   a farthest edge from $p$.   
To show this, construct a path $\sigma$ from $p$ to  $F(q)$ by following $\pi_p$ from $p$ to $a$, then traversing chord $(a,b)$, then  following $\pi_q$ from $b$ to $F(q)$.  Then $\sigma$ is longer than $\pi_p$.  We are done if  we  can show  $\sigma$ is a shortest  path.  But the part  up  to $b$ is locally shortest and the  part after $a$ is locally shortest,  thus none of the  bends in the path  can be shortened, so $\sigma$ is  a geodesic  path  and thus a  shortest path.   
\end{proof}

\subsection{Details on Shortest Paths To/From Edges}
\label{appendix:shortest-paths}

In this section we give \reviewerchange{linear-time algorithm}s to find  shortest paths from a  given  point to all edges of the polygon, and to find shortest paths from a  given  edge to all  vertices of the polygon. In fact, in both cases, we will augment to a \defn{shortest path map} that divides the polygon into regions (triangles and trapezoids) in which the shortest paths are combinatorially the same.

\medskip
\noindent{\bf Shortest Paths from a Point to all Edges and Vertices.}
For a point $p$ in polygon $P$, define \defn{$T_p$} be the \defn{shortest path tree} that consists of shortest paths from $p$ to all the edges and vertices of the polygon.  
In some situations we 
will only care about the shortest paths to edges, but we will still use the notation $T_p$ and just clarify what we mean.

\begin{lem}\label{SPT_lemma}
There is a \reviewerchange{linear-time algorithm} to find, given a point $p$ in a polygon, the shortest path tree $T_p$
and its augmentation to a shortest path map.
\end{lem}

\begin{proof}
The idea is simple.  Construct the shortest path tree from $p$ to all vertices and augment to the shortest path map using the algorithm by Guibas et al.~\cite{SPT_linear}. 
Regions of the shortest path map are triangles. 
Check each triangle in $O(1)$ time to see if it contains the last segment of a shortest path from $p$ to an edge. 
For further details see~\cite[Section 4.1.2]{lubiw_et_al:LIPIcs.ESA.2021.65, lubiw2021visibility} , which solves the more general case of shortest paths to a set of chords in a polygon when no two chords nest.
Note that these algorithms assume $p$ is on the boundary of the polygon, but we can handle an interior point $p$ by first cutting the polygon at a chord through $p$ in linear time and then finding shortest paths on each side of the chord.
\end{proof}

\medskip
\noindent{\bf Shortest Paths from an Edge to all Vertices.}
For an edge $e= ab$ of polygon $P$, define \defn{$T(e)$} to be the forest of shortest paths from $e$ to all vertices of the polygon.  %
We will use the shortest path trees $T_a$ and $T_b$ at the endpoints of $e$ to construct $T(e)$.

A vertex $v$ is \defn{visible} from $e$ if there is a line segment $xv$ inside $P$ for some point $x \in e$, and $v$ is \defn{orthogonally visible} from $e$ if $xv$ can be orthogonal to $e$.
Pollack et al.~\cite{pollack_sharir} note that vertex
$v$ is visible from $e$ iff $v$ has different parents $p_a(v)$ and $p_b(v)$ in $T_a$ and $T_b$.
Furthermore, if $p_a(v) \ne p_b(v)$ then $v$ is visible from the interval $[x_a(v),x_b(v)]$ in $e$ where $x_a(v)$ and $x_b(v)$ are the intersections of $e$ with the lines from $v$ to  $p_a(v)$ and $p_b(v)$ respectively.

\begin{figure}
\begin{subfigure}{\textwidth}
  \centering
  \includegraphics[width=\textwidth]{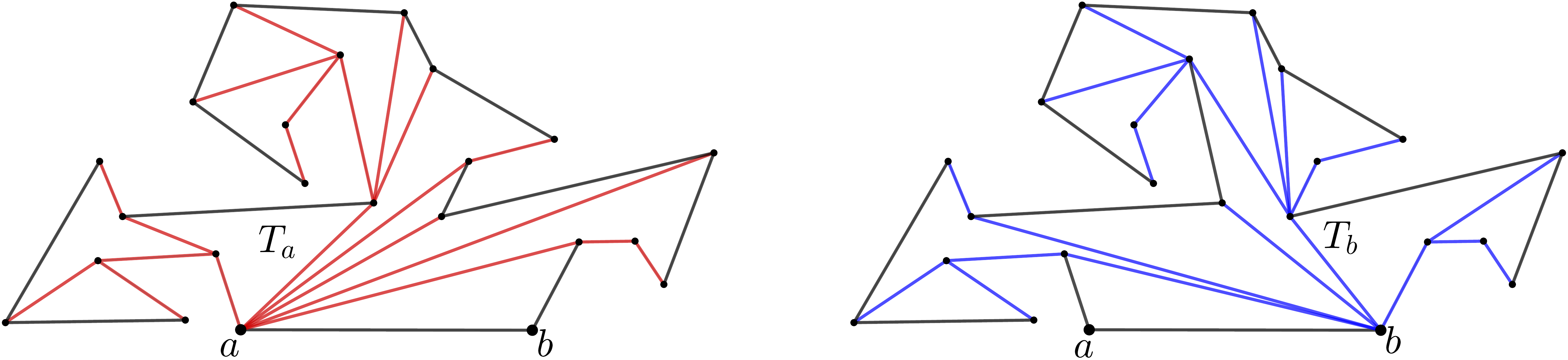}
%  \caption{Shortest Path Trees from the endpoints of edge $e=ab$}
  \label{fig:sfig1}
\end{subfigure}%
\newline
\begin{subfigure}{\textwidth}
  \centering
  \includegraphics[width=.8\textwidth]{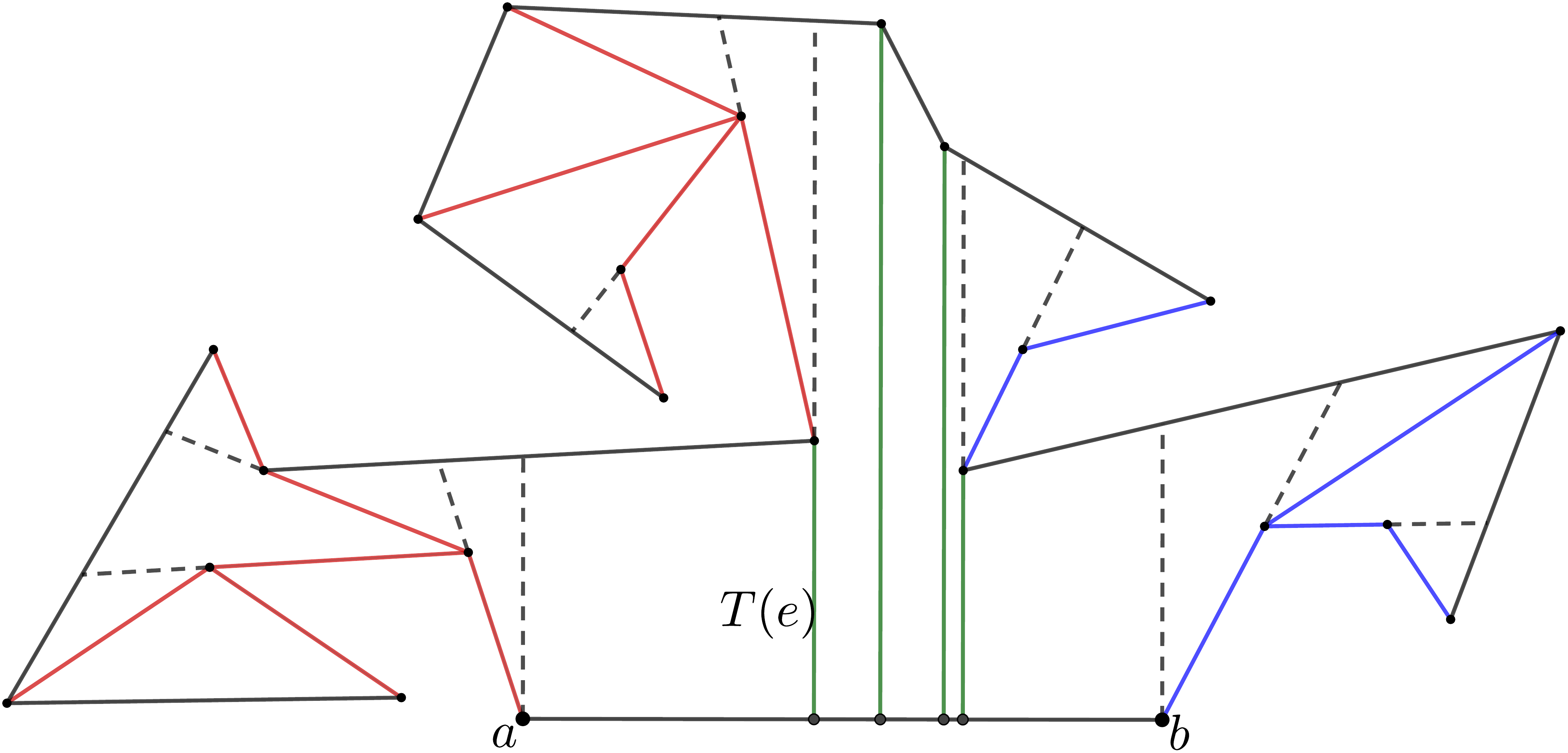}
%  \caption{The resulting shortest path map for the edge $e$}
\label{fig:sfig2}
\end{subfigure}
\caption{The shortest path trees from the endpoints of an edge can be used to construct its shortest path forest and map.
The color of 
an edge
indicates the shortest path tree it originates from. 
Green edges 
indicate orthogonal visibility from $e$.}
\label{fig:setvtree}
\end{figure}

\begin{lem}\label{ShortestPathForest_lemma}
There is a \reviewerchange{linear-time algorithm} to find, given an edge $e$ of a polygon, the forest $T(e)$ that consists of shortest paths from $e$ to all the vertices of the polygon, 
and to augment this to the shortest path map.
\end{lem}
\begin{proof}
See Figure~\ref{fig:setvtree}.
To construct $T(e)$,  we define the parent of each vertex of $P$. If $p_a(v) = p_b(v)$ then $v$ has the same parent in $T(e)$.  Otherwise, $v$ is visible from $e$.  If the angles $\angle v x_a(v) b$ and $\angle v x_b(v) a$ are both $\le \pi /2$ then $v$ is orthogonally visible from $e$, and we define the parent of $v$ to be the foot of the perpendicular from $v$ to $e$.   
And if one of the angles is obtuse, then the parent of $v$ in $T(e)$ is whichever of $p_a(v)$ or $p_b(v)$ that leads to the obtuse angle.

We now have the shortest path forest $T(e)$.
To augment to the shortest path map, we first 
construct the vertex shortest path maps for subtrees rooted at $a$, $b$, and all the orthogonally visible vertices.
This %
takes linear time.
Finally, we can extend the perpendiculars
from orthogonally visible vertices (plus the endpoints of $e$) until they intersect $\partial P$.
\changed{This splits
the polygon into} trapezoids and triangles giving the required shortest path map.
\changed{The runtime is linear.}
\end{proof}

Note that we can easily modify the algorithm in Lemma~\ref{ShortestPathForest_lemma} to construct the shortest path forest from any %
chord of a given simple polygon in linear time.

\subsection{Details on Separators and Funnels}
\label{appendix:separators}

Any geodesic path between two vertices of $P$ separates the boundary of $P$ into two parts, and when we focus on which vertices/edges are in opposite parts,
we call the geodesic path 
a ``separator''.
Separators are a 
main tool for finding all farthest vertices in a polygon. 
They were first introduced by Suri~\cite{suri1989} (although he called them ``connectors'' rather than  ``separators'') in his $O(n \log n)$ time algorithm to find farthest
vertices of all vertices, and then they were used by Hershberger and Suri~\cite{hershberger1997matrix} who improved the runtime to $O(n)$.  
Vertex separators (called ``separating paths'') were also used by Ahn et al.~\cite{linear_time_geodesic}, both when they appealed to Hershberger-Suri, and in more direct ways.  
We need edge separators in similar ways. 

The basic properties that Suri~\cite{suri1989} proved for 
separators for farthest vertices are as follows:
\begin{enumerate}
\item 
If two vertices $x$ and $y$ are separated by a geodesic path $\pi(a,b)$, then the shortest path from $x$ to $y$ is contained, except for one edge, in the shortest path trees of $a$ and $b$~\cite[Lemma 4]{suri1989}.  Thus, after constructing the shortest path trees from $a$ and $b$, it is easy to find shortest paths for any pair $x,y$ that is separated by $\pi(a,b)$.
\item A constant number of separators suffice to separate every vertex from its farthest vertex~\cite[Section 4]{suri1989}.

\end{enumerate}

In this section we develop the analogous theory of separators for farthest edges.

\begin{definition}
A \defn{farthest edge separator} is a directed geodesic path $\gamma = \pi(a,b)$ from some vertex $a$ to some vertex $b$ of $P$ such that for every point $p \in \delta P$ to the %
right of $\gamma$, all of $p$'s farthest edges lie to the %
left of $\gamma$.
\end{definition}

Note that we define separators via the strong property that \emph{all} points to one side have their farthest edge on the other side.  Although this property is not part of Suri's original definition, his construction produces vertex separators with the property.

In this section we will prove that Suri's two properties hold for our farthest edge separators.
We first note an even more basic property that is the main reason for using separators: 

\begin{claim}
\label{claim:across-separator}
If $\gamma = \pi(a,b)$ is a farthest edge separator and points $p$ and $q$ lie to the right
of $\gamma$ and their farthest edges $F(p)$ and $F(q)$ are distinct, then the paths to their farthest edges cross.
\end{claim}
\begin{proof}
The ordering $p,F(p),q,F(q)$ is excluded by the separator.  The ordering $p,q,F(q),F(p)$ cannot occur by Property~\ref{prop:anti-parallel}.  Thus the ordering must be $p,q,F(p),F(q)$.  See also Figure~\ref{fig:correct_orderings}.
\end{proof}

\subsubsection{Funnels and Shortest Paths Across a Separator}
\label{sec:funnels-paths}

We first address Suri's property (1) by examining how a shortest path crosses a geodesic $\gamma(a,b)$.
In this subsection
the geodesic need not be a farthest edge separator, and the shortest path need not go to a farthest edge.
Hershberger and Suri~\cite{hershberger1997matrix} expanded on Suri's result and showed how a shortest vertex-to-vertex path that crosses $\gamma$ is related to the \emph{funnels} of the vertices.  We follow their analysis.

Suppose that vertex $v$ lies to the 
right of $\gamma$ and edge $e$ lies to the 
left of $\gamma$.  Then $\pi(v,e)$ crosses $\gamma$, either at a single point, or by sharing chords with $\gamma$.  
See Figure~\ref{fig:funnels-and-tangents}.
We show that $\pi(v,e)$ lies in the \emph{funnels} of $v$ and $e$ which are defined in terms of the shortest path trees $T_a$ and $T_b$.
Note that $\gamma$ is in both $T_a$ and $T_b$ since it is the shortest path from $a$ to $b$.

The \defn{funnel of $v$}, denoted \defn{$Y(v)$}, 
is bounded by $\pi(a,v)$,  $\pi(b,v)$ and $\gamma$,
where $\pi(a,v)$ and  $\pi(b,v)$ are called the \defn{walls} of the funnel.
The vertex where $\pi(a,v)$ diverges from $\gamma$ is \defn{$\gamma_a(v)$}, defined to be the lowest common ancestor of $v$ and $b$ in the tree $T_a$.  Similarly, the vertex where $\pi(b,v)$ diverges from $\gamma$ is \defn{$\gamma_b(v)$}, the lowest common ancestor of $v$ and $a$ in the tree $T_b$. 
The vertex where $\pi(v,a)$ diverges from $\pi(v,b)$ is called the \defn{apex} of the funnel. Observe that the path between the apex and $\gamma_a(v)$ [or $\gamma_b(v)$] is reflex.

Similarly, the \defn{funnel of $e$}, \defn{$Y(e)$} is bounded by $\pi(a,e)$, $\pi(b,e)$, $\gamma$, together with the piece of $e$ between the terminals $t(a,e)$ and $t(b,e)$ if those terminals are distinct. The lowest common ancestors \defn{$\gamma_a(e)$} and \defn{$\gamma_b(e)$} and the \defn{apex} can be defined analogously, where we allow the apex to be the piece of edge $e$ between $t(a,e)$ and $t(b,e)$ when those terminals are distinct. 
Funnels have been used in many shortest path algorithms,
and there are variations on how they are defined (as a subpolygon or a set of edges; including the edges common to two paths or not, etc.).

\begin{figure}
    \centering
    \includegraphics[width=\textwidth]{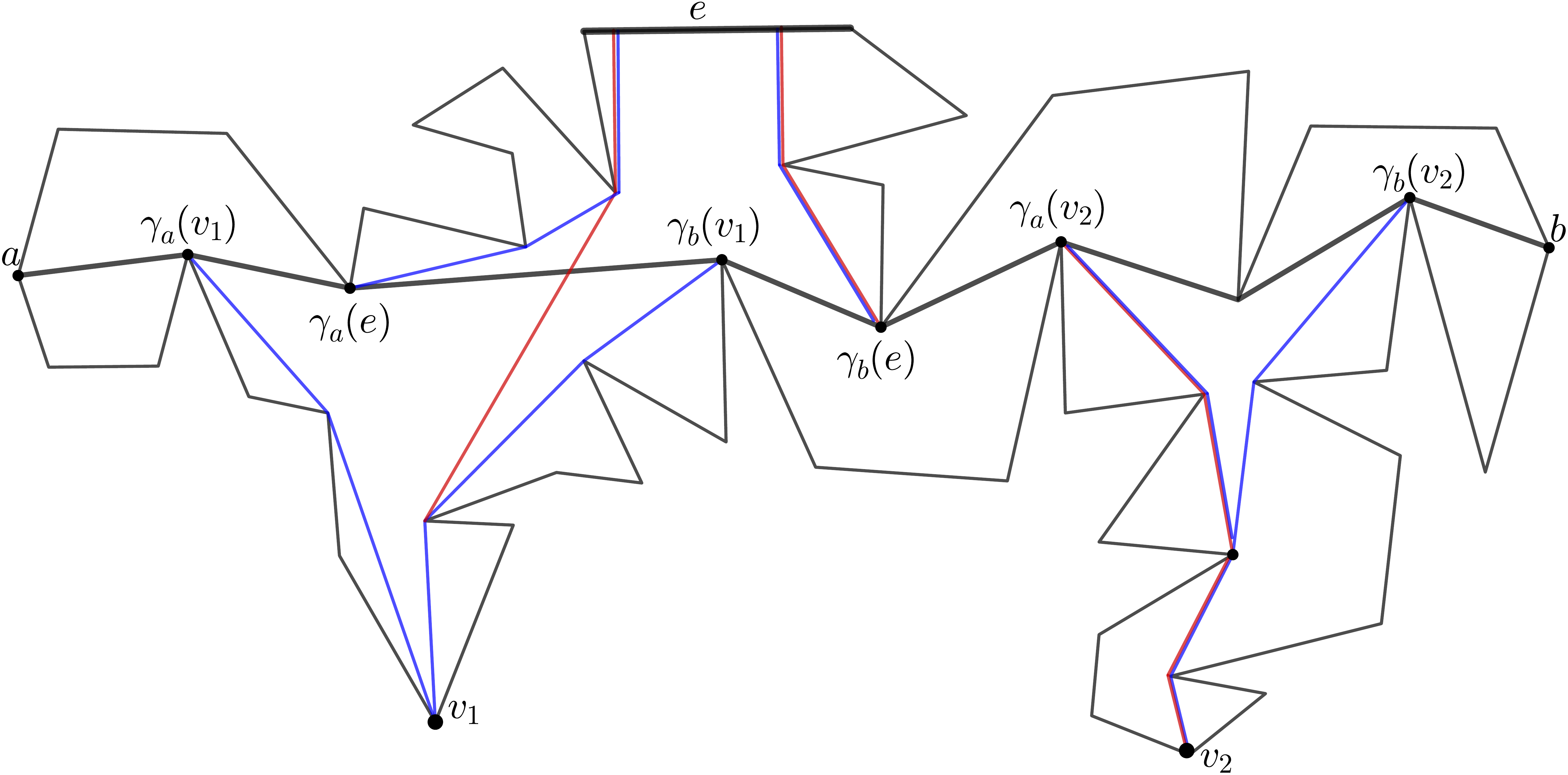}
    \caption{A geodesic $\gamma = \pi(a,b)$, the funnels $Y(v_1), Y(v_2)$, and $Y(e)$ (in blue) and the paths $\pi(v_1, e)$ and $\pi(v_2, e)$ (in red).}
    \label{fig:funnels-and-tangents}
\end{figure}

The pair of funnels $Y(v),Y(e)$ is \defn{closed} if 
the paths  $\pi(\gamma_a(v),\gamma_b(v))$ and  $\pi(\gamma_a(e),\gamma_b(e))$ are internally disjoint, see $Y(v_2)$ and $Y(e)$ in Figure~\ref{fig:funnels-and-tangents}.  Otherwise the pair of funnels  is \defn{open}, see $Y(v_1)$ and $Y(e)$ in the figure.
Hershberger and Suri 
dealt with the case where 
$e$ is replaced by a vertex $u$.
They showed that if the pair $Y(v), Y(u)$ is closed, then the edges of $\pi(v,u)$ are edges of the funnels.  In particular, suppose $\gamma_a(u)$ and $\gamma_b(u)$ are closer to $a$ than  $\gamma_a(v)$ and $\gamma_b(v)$ (the other ordering is analogous).  Then $\pi(v,u)$ consists of the paths 
$\pi(v, \gamma_a(v))$, $\pi(\gamma_a(v),\gamma_b(u))$, 
$\pi(\gamma_b(u), u)$.
On the other hand, if the pair $Y(v),Y(u)$ is open, then $\pi(v,u)$ consists of part of a wall of $Y(v)$ and part of a wall of $Y(u)$ joined by  
a \defn{tangent} edge \defn{$\ell(v,u)$} that crosses $\gamma$.

To deal with an edge funnel $Y(e)$, we abuse the notation and say that a segment that meets $e$ at right angles is tangent to $Y(e)$ (this makes   sense if we imagine that the edges that meet $e$ at right angles extend off to infinity). 

The above results are used to prove the following two lemmas and will also be used in
Appendix~\ref{appendix:Hershberger-Suri}
when we show how to extend Hershberger and Suri's algorithm for finding farthest vertices of all vertices to the case of farthest edges.

\begin{lem}
\label{lem:separator-paths}
\changed{Consider a geodesic path $\gamma = \pi(a,b)$ with vertex $v$ to the right and edge $e$ to the left.  If the pair of funnels $Y(v),Y(e)$ is closed then the edges of $\pi(v,e)$ are contained in the shortest path trees $T_a$ and $T_b$.
If the pair of funnels $Y(v),Y(e)$ is open then the edges of $\pi(v,e)$ are contained in the shortest path trees $T_a$ and $T_b$, except for one edge $\ell(v,e)$ that crosses $\gamma$ and is tangent to $Y(v)$ and $Y(e)$.
}

\end{lem}
\begin{proof}
Consider the terminal point $t(v,e)$ of the path $\pi(v,e)$.  If $t(v,e)$ is a vertex of $P$, then the previous results apply. 
Otherwise, 
\changed{let $s$ be the last segment of the path $\pi(v,e)$.  Segment $s$ meets $e$ at a right angle at point $t(v,e)$. Let $u$ be the vertex at the start of $s$. 
If $u$ is to the left of $\gamma$ then $s$ is an edge of $Y(e)$ and the result follows from the previous result for $v$ and $u$.
Otherwise, $u$ is to the right of $\gamma$, the pair of funnels is open, and $s$ is the tangent edge that crosses $\gamma$.}
\end{proof}

\begin{lem}\label{lemma:walls_can_be_constructed_fast}
Let $\gamma = \pi(a,b)$ be a geodesic path.  
After \reviewerchange{linear-time} preprocessing (to compute the trees 
$T_a$ and $T_b$
and preprocess them for answering lowest common ancestor queries in constant time), the shortest path $\pi(v,e)$ from any vertex $v$ to the 
right %
of $\gamma$ to any edge $e$ to the 
left %
of $\gamma$
can be computed in 
time proportional to the number of vertices in $\pi(v,e)$.
\end{lem}
\begin{proof}
Compute the least common ancestors $\gamma_a(v)$, $\gamma_b(v)$, $\gamma_a(e)$, and $\gamma_b(e)$ in constant time. \reviewerchange{(This can be done after linear-time preprocessing using the algorithm of Harel and Tarjan~\cite{harel1984fast})}.
Test if the pair of funnels $Y(v),Y(e)$ is closed or open in constant time using least common ancestor queries. 
If the pair is closed, the path $\pi(v,e)$ consists of subpaths that can be found in time linear in the number of vertices.

Otherwise, we must find the tangent edge $\ell(v,e)$ of the two funnels.
The edge $\ell(v,e)$ 
may meet $e$ at right angles, which is a special case we deal with later.
Note that it suffices to search between the apexes of the funnels---to ease notation, we will just suppose that that those apexes are $v$ and $e$ themselves.
Then $\ell(v,e)$ is tangent to two reflex curves where one is a wall of the funnel $Y(v)$ and one is a wall of the funnel $Y(e)$.  
There are four possible choices for the two reflex curves, and for each choice, we are essentially finding the common tangent of two disjoint convex polygons, a very well-solved problem (see~\cite{abrahamsen2018common} for some history).  A simple search that walks from the two apexes along the chosen paths towards $\gamma$ will find the tangent in time proportional to the number of vertices traversed---and those vertices are part of the output path.  Doing this in parallel over the four choices, we can find $\ell(v,e)$ and $\pi(v,e)$ in time linear in the number of vertices of $\pi(v,e)$.
(This is the same argument as given by Ahn et al.~\cite[Lemma 3.5]{linear_time_geodesic}.)
Finally, to address the possibility that
$\ell(v,e)$ meets $e$ at right angles, we can perform a similar search between $e$ and each of the walls of $v$'s funnel.
\end{proof}

\subsubsection{Constant Number of Separators}
\label{sec:constant-separators}

We now turn to Suri's property (2)---finding a constant number of separators.

\begin{lem}
\label{lem:constant_splits}
\label{lem:constant-separator-set}
There is a set of at  most five farthest edge separators such that 
every  point $p  \in   \partial P$ 
(and consequently, every edge of $P$) lies to the 
\reviewerchange{right}
of at least one  of the separators. 
Furthermore, such  a set of separators can be found in linear time.
\end{lem}

This lemma 
is extremely important because it reduces farthest edge problems to a constant number of ``bipartite'' cases where the source vertices are separated from the target edges.
Lemma~\ref{lem:constant-separator-set} will be used in 
Appendix~\ref{appendix:Hershberger-Suri}
to find farthest edges from all vertices.  It will also be used in
Appendix~\ref{appendix:hourglasses}
to find 
hourglasses in $P$
and in 
Appendix~\ref{appendix:coarse-cover}
to construct the coarse cover of $P$.

\begin{proof}[Proof of Lemma~\ref{lem:constant_splits}]
We first note the 
consequence that for every polygon edge $(x,y)$, one of the five separators has both $x$ and $y$ to its 
\reviewerchange{right.}
This is because separator
endpoints are
vertices so a farthest edge separator for the midpoint of edge $(x,y)$ must have $x$ and $y$ to its 
\reviewerchange{right}.
Thus it suffices to prove that there are five farthest edge separators such that every point $p \in \partial P$ lies to the 
\reviewerchange{right}
of at least one separator.

The plan in Suri's proof for the case of farthest vertices, was to follow a chain $v_1, v_2, v_3,v_4$ where $v_{i+1}$ is the farthest vertex from $v_i$, and argue that $\pi(v_3,v_4)$ crosses $\pi(v_1,v_2)$, and that this provides three separators, namely the %
three paths.  Our plan is similar, but a bit trickier because our paths go from a vertex/point to a farthest edge, so we must then choose a point in the edge to continue the chain.

Take an arbitrary vertex $u$ and find its farthest edge $F(u)$. Note that $F(u)$ is unique by Lemma~\ref{lem:1D-Voronoi-edges}.
This can be achieved in linear time by constructing the shortest path tree
$T_u$
(Lemma~\ref{SPT_lemma}) and finding a leaf furthest from $u$ in this tree.
Suppose $F(u)$ is the  edge $e$ with endpoints $e^-,  e^+$ in  clockwise order.
Find the farthest edges $F(e^-)$ and  $ F(e^+)$ in linear time.

\textbf{Case 1.} First, we suppose that the geodesics $\pi(e^-, F(e^-))$ and $\pi(e^+,F(e^+))$ both cross $\pi(u,e)$. See Figure~\ref{fig:possible_case_2a}.  We claim that 
the geodesics 
$\gamma_1 = \pi(u,e^+)$,  and $\gamma_2 = \pi(e^-,u)$
are farthest edge separators.
To  prove this,  consider a point $p$ to  the 
right
of $\gamma_1$, i.e., a point in the clockwise chain $C$ from  $e^+$ to $u$, and suppose that $p$ has a farthest edge  $F(p)$ on the same chain.  Note that $F(p) \ne e$,  since $e$ is not part of the  chain.   If $F(p)$ occurs before $p$ along the chain $C$, then $p,u,F(u),F(p)$ occur in that clockwise order and violate Property~\ref{prop:anti-parallel}. 
Otherwise $F(p)$  occurs after $p$ along the chain  $C$ in which case $e^+,p,F(p),F(e^+)$  occur in that clockwise order and violate Property~\ref{prop:anti-parallel}.
A symmetric argument shows that 
$\gamma_2$ is a farthest edge separator.

The two geodesics $\gamma_1$ and $\gamma_2$ separate all  points  of $\delta P$ from their farthest edges except the  points  of edge $e$.  We separate those points by  adding  one more  geodesic $\pi(e^+,e^-)$.  Note that this kind of degenerate %
separator is allowed by the definition, and 
is a farthest edge separator 
since every point in  the edge $e$ has its farthest edge outside $e$.

This gives a set of three farthest edge separators.  Note that they can be found in linear time.

\textbf{Case 2.} 
Otherwise at least 
one of the geodesics $\pi(e^-, F(e^-))$ and $\pi(e^+,F(e^+))$ does not cross $\pi(u,e)$.  
We will consider the case when 
the geodesic $\pi(e^-,F(e^-))$ does not cross $\pi(u,e)$---the other case is symmetric.
Suppose $F(e^-)$ is the edge $f = (f^-,f^+)$ in clockwise order. 
Find the shortest path from $u$ to edge $f$, and let
point $p:=t(u,f)$ %
be the terminal of that path.  Find the farthest edge $g := F(p)$ and suppose $g = (g^-, g^+)$ in clockwise order.
We claim that $g$
cannot lie in the 
clockwise chain from $f^+$ to $e^-$. 
Suppose it does.
Then $g \ne e$, which implies that $p \ne u$ (since $u$ has the unique farthest edge $e$). But then $u,p, F(p), F(u)$ violate Property~\ref{prop:anti-parallel}.
Therefore, the edge $g$ lies either: (a) in the clockwise chain
from $e^-$ to $u$, in which case we find separators; or 
(b) in the clockwise chain from $u$ to $f^-$, which we prove is impossible. We consider the two cases (a) and (b).

\textbf{Case 2a.} The edge $g = F(p)$ lies in the clockwise chain from $e^-$ to $u$. 
The situation is depicted in Figure~\ref{fig:possible_case_2a}.
We claim that the geodesics $\gamma_1 =  \pi(u,e^+)$, $\gamma_2 =  \pi(f^-,g^+)$ and $\gamma_3 = (e^-,f^+)$ are 
farthest edge separators. Note that $\gamma_2$ is redundant if $f^- = u$, and  $\gamma_1$ is redundant if $g=e$. 
To prove that $\gamma_1$ is a farthest edge separator, 
note that because of the ``anti-parallel'' pair $\pi(u,e)$ and $\pi(e^-,f)$, no point $p \in \partial P$ to the 
right %
of $\gamma_1$ has a farthest edge to the 
right %
of $\gamma_1$ (otherwise the path from $p$ to  such  a farthest edge 
must go in the same direction as one of 
$\pi(u,e)$ and $\pi(e^-,f)$, thus violating 
Property~\ref{prop:anti-parallel}).
Similarly, $\gamma_2$ is a farthest edge separator because of the anti-parallel pair $\pi(p,g)$ and $\pi(e^-,f)$, and $\gamma_3$  is a farthest edge separator because of the same anti-parallel pair. 

The three geodesics $\gamma_1, \gamma_2$ and $\gamma_3$ separate all points of $\partial P$
from their farthest edges except the points of edges $e$ and $f$.   We can separate those points by adding the geodesics $\pi(e^+,e^-)$ and $\pi(f^+,f^-)$.
This gives a set of five farthest edge separators.  Note that they can be found in linear time.

\begin{figure}[tb]
\centering
\includegraphics[width=0.9\textwidth]{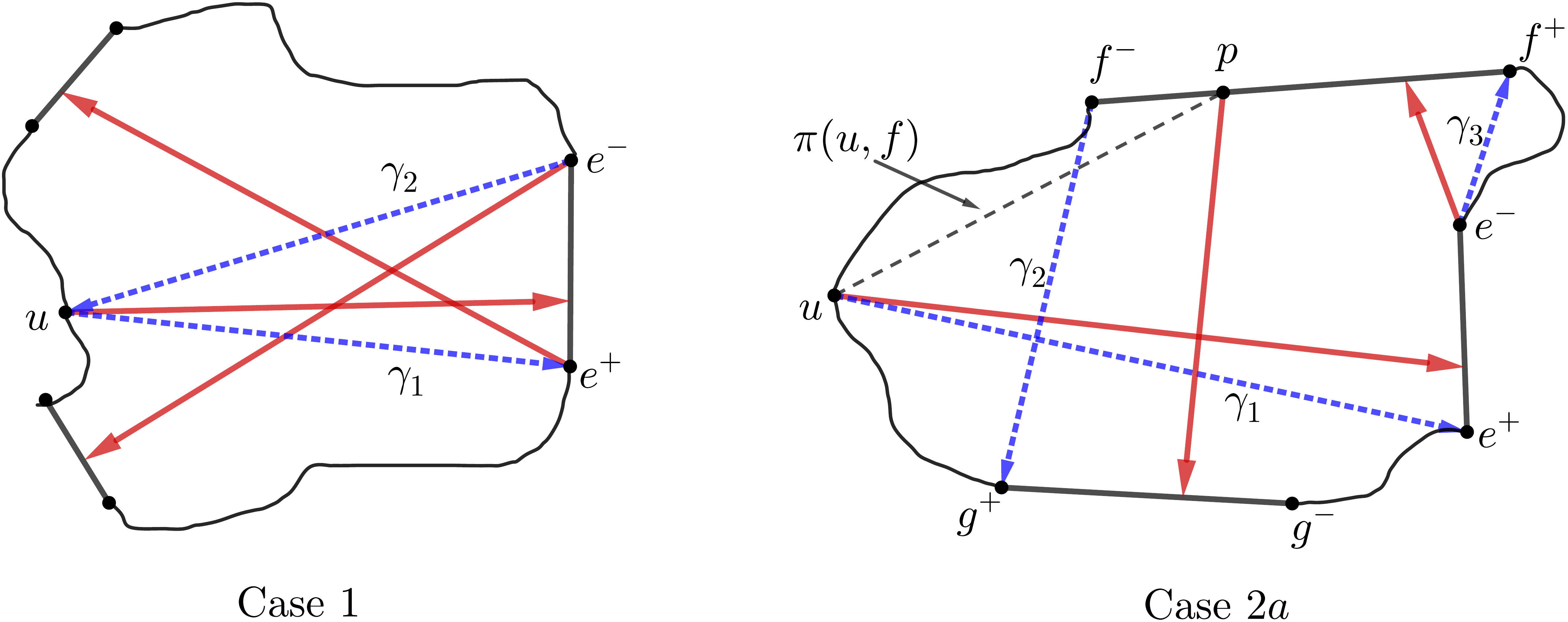}
\caption{
Case 1 and Case 2a from the proof of Lemma~\ref{lem:constant_splits}, showing (schematically) the paths from points to farthest edges (in red) and the separators (in blue).
}
\label{fig:possible_case_2a}
\end{figure}

\textbf{Case 2b.} The edge $g = F(p)$ lies in the clockwise chain from $u$ to $f^-$. 
See Figure~\ref{fig:impossible_case_2b}.
We will prove that this case cannot occur. 
To show this, 
consider the geodesic paths 
$\sigma_1 := \pi(u,f) = \pi(u,p)$, 
$\sigma_2 :=  \pi(e^-,g)$ and $\sigma_3 := \pi(p,e)$.
Note that because $e$ is the unique farthest edge from $u$, $d(u,e) > |\sigma_1|$.
Similarly, $d(e^-, f) > |\sigma_2|$, and $d(p,g) \ge |\sigma_3|$ ($p$ need not have a unique farthest edge).
Adding these together, %
we obtain
\begin{align}
d(u,e) +d(e^-,f) + d(p,g) > |\sigma_1| + |\sigma_2| + |\sigma_3|.
\label{eqn:path-lengths}
\end{align}

\begin{figure}
\centering
\includegraphics[width=0.55\textwidth]{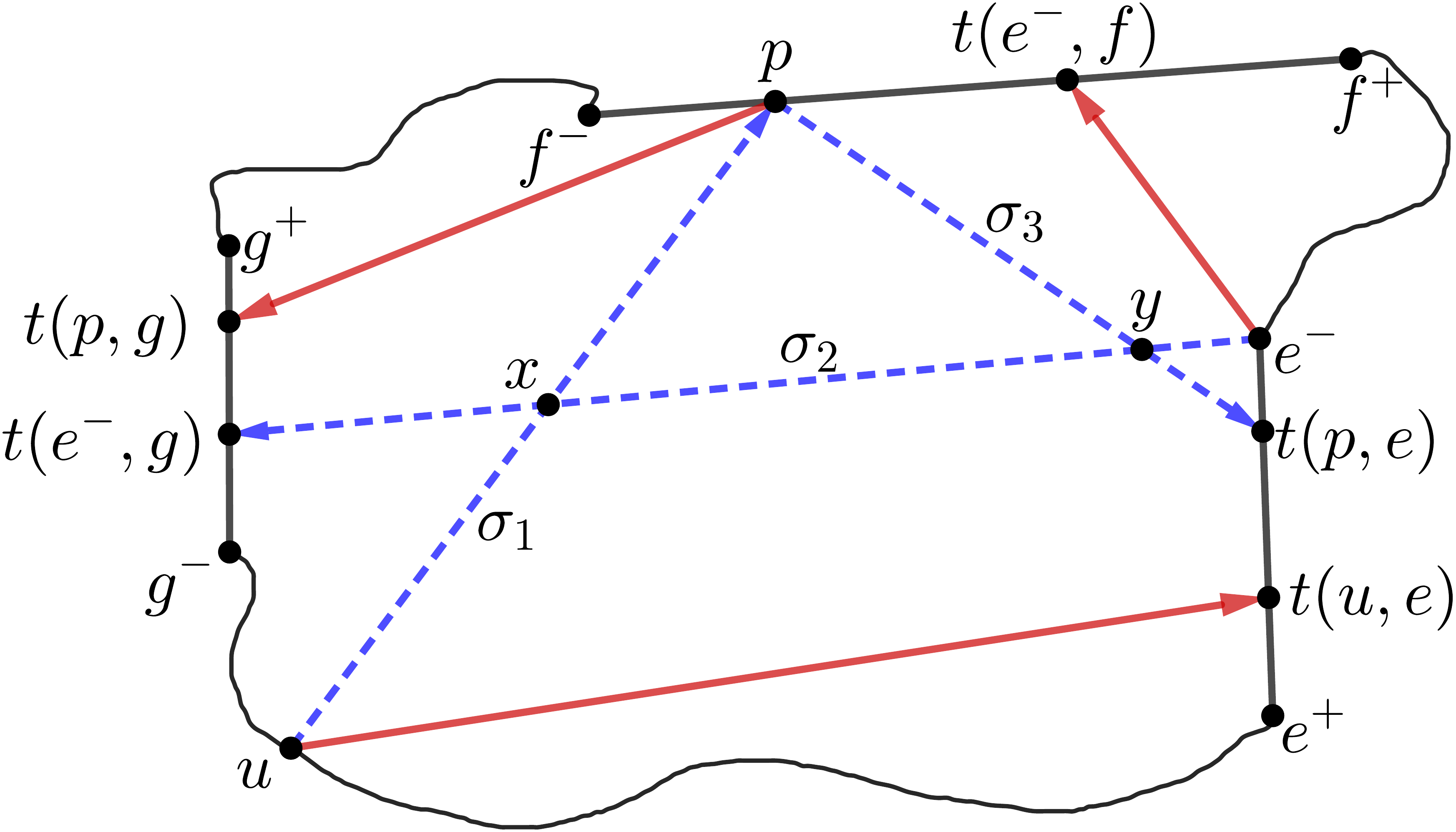}
\caption{
Case 2b from the proof of  Lemma~\ref{lem:constant_splits}, showing (schematically) the paths from points to farthest edges (in red) and the shortest paths $\sigma_i$ (in blue).
}
\label{fig:impossible_case_2b}
\end{figure}

Recall that for vertex $v$ and edge $h$,  $t(v,h)$ is the terminal point of the path $\pi(v,h)$.
Observe that in clockwise order $p \le t(e^-,f)$ on edge $f$, and $t(e^-,g) \le t(p,g)$ on edge $g$.
Let $x$ be the intersection point of $\sigma_1$ and $\sigma_2$ (possibly at one of their endpoints).  Let $y$ be the intersection point of $\sigma_2$ and $\sigma_3$ (possibly at one of their endpoints). Observe that along $\sigma_2$, $y$ precedes (or is equal to $x$).
See Figure~\ref{fig:impossible_case_2b}.
We get the following inequality from the definition of a terminal and the triangle inequality:
\begin{align*}
  d(u,e) & = d(u,t(u,e)) \leq d(u,t(p,e)) \\
   &\leq d(u,x) + d(x,y) + d(y,t(p,e)) 
\end{align*}
\noindent Reasoning as above, we also get the following two inequalities:
\begin{align*}
& d(e^-,f) \leq d(e^-,y) + d(y,p)\\
& d(p,g) \leq d(p,x) + d(x,t(e^-,g))
\end{align*}
\noindent Adding the three inequalities 
and noting that we have used each subpath of each $\sigma_i, i=1,2,3$ exactly once, we obtain:
$$d(u,e) +d(e^-,f) + d(p,g) \leq |\sigma_1| + |\sigma_2| + |\sigma_3|.$$
which contradicts Equation~\ref{eqn:path-lengths}.
\end{proof}

\subsection{Details for Section~\ref{section:chord_oracle}, Chord Oracles and Coarse Covers}
\label{appendix:oracles}

The algorithms to find the relative center on a chord and to find the center of a polygon depend on a crucial convexity property.
Define the \defn{geodesic radius function}, $r(x)$, for $x \in P$
to be the maximum geodesic distance from $x$ to a site (a vertex or edge). 
Thus the center is the point $x$ that minimizes $r(x)$.
A function is \defn{geodesically convex} on $P$ if the function is convex on every geodesic path in $P$.
The following result was proved for vertex sites by Pollack et al.~\cite{pollack_sharir} and for edge sites by Lubiw and Naredla~\cite{lubiw2021visibility}.

\begin{lem}
\label{lem:geodesically-convex}
The geodesic radius function $r(x)$ is geodesically convex.
\end{lem}

For our extensions of the chord oracle in the following subsection, we need some details of the
$O(n)$ time chord oracle algorithms of Pollack et al.~\cite[Section 3]{pollack_sharir} for the vertex center and of Lubiw and Naredla~\cite[Section 4.1]{lubiw2021visibility} for the edge center. 

\medskip
\noindent
\ \ \textbf{Chord Oracle}

\smallskip
\noindent
\ \ \ \ {\bf Input:} a chord $K$ of polygon $P$ on $n$ vertices.\\
\ \ \ \ {\bf Output:} whether the center of $P$ lies left/right/on $K$.
\vskip -1in
\setlength{\leftmargini}{.5in}
\begin{enumerate}
\squeezelist
\item Find a coarse cover of $K$.
\item Find
the point on $K$ that minimizes the upper envelope of the coarse cover functions---this is the relative center $c_K$. 
\item Examine the 
maximum values of the coarse cover functions
at $c_K$ to determine
whether the center of $P$ lies 
left/right/on $K$.
\end{enumerate}
\setlength{\leftmargini}{\parindent}

The details of these steps (none of which is trivial) can be found in~\cite{pollack_sharir,lubiw2021visibility}.
Step 1 runs in time $O(n)$ and produces a coarse cover $\cal T$ of size $O(n)$.

Step 2 runs in time $O(|{\cal T}|)$ using divide-and-conquer to reduce the search space to a subinterval of $K$ while eliminating elements of the coarse cover.  It uses a basic test of whether the relative center lies to the left or right of a point on $K$.  The correctness of this test depends on 
convexity of the upper envelope function on $K$ (Lemma~\ref{lem:geodesically-convex}), and on the fact that 
the coarse cover captures the first segments of paths to the farthest sites.  
If the first segments pull the test point in opposite directions on $K$, then the point is locally optimum and therefore is the relative center; and otherwise, we know which direction the test point should move.

Step 3 
similarly relies on Lemma~\ref{lem:geodesically-convex} and uses the first segments of paths from $c_K$ to its farthest sites.  From those segments, we can detect if $c_K$ is locally---and hence globally---optimal, and otherwise decide which side of $K$ to move to.
Step 3 takes time $O(|{\cal T}|)$.

\subsubsection{Extensions of the Chord Oracle}
\label{appendix:geodesic-oracle}

In this section we give two extensions of the chord oracle that we use in 
Phase II.  
The divide and conquer algorithm in Phase II recurses on subpolygons that are \newchanged{simple $3$-anchor hulls.} 
The first extension of the chord oracle is
a version that works on a chord of a subpolygon (a \newchanged{simple $3$-anchor hull}) and uses the coarse cover of the subpolygon to get a coarse cover of the chord.
The runtime will be linear in the size of the subpolygon's coarse cover, which decreases during the divide-and-conquer algorithm.
The second extension is 
a generalization of the chord oracle to a \emph{geodesic oracle}.
\newchanged{A geodesic path divides a polygon into regions, and the geodesic oracle tells us which region contains the center.}
We need this because 
our subpolygons are bounded by geodesics.

Let $Q$ be a 
\newchanged{simple $3$-anchor hull}
in $P$.
By Observation~\ref{obs:3-geo-cell}, $Q$ is geodesically convex in $P$
so the intersection of $Q$ with a chord [geodesic] of $P$ is a chord [geodesic] of $Q$.

We say that a subset ${\cal T}$ of the coarse cover of $P$ is a \defn{coarse cover of $Q$ in $P$} if condition 3 of the definition of a chord cover (Definition~\ref{defn:coarse-cover}) holds for all points in 
the interior of
$Q$, i.e., for any point $x$ in the interior of $Q$
and any edge $e$ of $P$ that is farthest from $x$, there is a triple $(R,f,e)$ in the coarse cover with $x \in R$.
\newchanged{(In particular, we get a coarse cover of $Q$ by  taking all the coarse cover elements whose triangles intersect the interior of $Q$.)}

Recall that each triangle $T$ of the coarse cover of $P$ is bounded by a segment of an edge of $P$ and two chords, and we store the endpoints of the chords on $\partial P$.  
As we recurse on subpolygons $Q$ we will maintain the endpoints of these chords on $\partial Q$.

\begin{lem}
\label{lem:generalized-chord-oracle} 
For the geodesic \reviewerchange{edge center} problem, 
there is an algorithm that
takes as input 
a 
\newchanged{simple $3$-anchor hull}
$Q$ known to contain the center of $P$ in its interior,
a coarse cover $\cal T$ of $Q$ in $P$,
and a chord $K$ of $Q$, 
and decides whether the center lies left/right/on $K$. The runtime is 
$O(|{\cal T}|)$. 
\end{lem}

\begin{proof}
We first construct a coarse cover ${\cal T}_K$ of $K$. 
For each triple $(T,f,e)$
in ${\cal T}$, let $T_K$ be the intersection of triangle $T$ with $K$.
Note that $T_K$ is a subsegment of $K$ and can be found in constant time from the boundary chords of $T$.
Add the triple $(T_K,f,e)$ to ${\cal T}_K$.  The resulting set ${\cal T}_K$ is a coarse cover of $K$ of size at most $O(|{\cal T}|)$.

Next, we follow Steps 2 and 3 of the 
chord oracle algorithm---Step 2 finds the relative center on $K$, and Step 3 decides whether the center lies to the left or right (or on) $K$.
As noted above,
each step takes time $O(|{\cal T}_K|)$.  
\end{proof}

Next we generalize Lemma~\ref{lem:generalized-chord-oracle} to a geodesic path.

\begin{lem}
\label{lem:geodesic-oracle} For the geodesic \reviewerchange{edge center} problem, 
there is an algorithm that
takes as input 
a 
\newchanged{simple $3$-anchor hull} %
$Q$ known to contain the center of $P$
in its interior, a coarse cover $\cal T$ of $Q$ in $P$,
and a geodesic $\gamma = \pi(a,b)$ in $Q$ with $a,b \in \partial Q$, 
\newchanged{and finds which subregion formed by $\gamma$ contains the center, or if the center lies on the geodesic.}
The runtime is 
$O(|Q| + |{\cal T}|)$.
\end{lem}

The idea is similar to that of  Lemma~\ref{lem:generalized-chord-oracle}.  We must first describe how to intersect the triangles of the coarse cover of $Q$ with the geodesic $\gamma$.
Since coarse cover triangles are bounded by chords, we can 
use the following result.

\begin{lem}
\label{lem:intersection-with-geodesic}
There is an algorithm that takes as input a \newchanged{simple} subpolygon $Q$, a geodesic $\gamma = \pi(a,b)$ in $Q$ with $a,b \in \partial Q$, and a set of chords $\cal K$ of $Q$, and finds the intersections of the chords of $\cal K$ with $\gamma$.  Each chord of $\cal K$ is given by its endpoints together with the identity of the edge of $Q$ containing the endpoint.
The runtime is $O(|Q| +  |{\cal K}|)$.
\end{lem}
\begin{proof}
Suppose the 
geodesic $\gamma$ has $g$ segments.  Then it divides the boundary of $Q$ into $g+1$ subchains, and we can traverse $\partial Q$ once to identify, for each edge of $Q$, which subchain contains it.

Direct $\gamma$ from $a$ to $b$.  This also directs the subchains of $Q$.  Identify each segment $s = uv$ of $\gamma$ with the subchain $c_s$ that ends at $v$. 
The subchain $c_s$ is unique except for the last segment incident to 
vertex $b$ where two subchains end---use one of the two subchains and ignore the other one. 

Observe that if chord $K \in {\cal K}$ crosses segment $s$, then $K$ has an endpoint in $c_s$.
Thus 
we can iterate through the chords $K$, finding which subchain contains each endpoint, and testing whether the associated segment of $\gamma$ intersects $K$.

The total time is $O(|Q| + |{\cal K}|)$.
\end{proof}

With Lemma~\ref{lem:intersection-with-geodesic} in hand, we can prove Lemma~\ref{lem:geodesic-oracle}.

\begin{proof}[Proof of Lemma~\ref{lem:geodesic-oracle}]
For each triangle of the coarse cover ${\cal T}$, the endpoints of its defining chords on $\partial Q$ are known.
Denoting the set of these defining chords by ${\cal K}$, we apply Lemma~\ref{lem:intersection-with-geodesic} to determine the intersections of the chords in ${\cal K}$ with $\gamma$. 
This takes $O(|Q| + |{\cal K}|)$ time, or equivalently, $O(|{Q}| + |{\cal T}| )$ time.

From the chord intersections, we can determine the intersections of the triangles of the coarse cover $\cal T$ with the segments of $\gamma$.
Each segment $s$ of $\gamma$ is a chord of $Q$.
Let ${\cal T}_s$ be the coarse cover elements whose triangles intersect the interior of $s$.
Each intersection is an interval of $s$ and the set of these intersections gives a coarse cover of $s$ of size $O(|{\cal T}_s|)$.
The chord oracle of Lemma~\ref{lem:generalized-chord-oracle} then determines whether the edge center lies left/right/on the segment $s$
in time 
$O(|{\cal T}_s|)$. 
Running the algorithm for \textit{all} the segments of $\gamma$ will take $O(|{\cal T}|)$ time in total
because each triangle of $\cal T$ intersects the interior of at most one segment of $\gamma$.
If the center lies on one of the segments, then it lies on $\gamma$.
Otherwise, 
since the segments  partition $Q$ into disjoint regions, knowing which side of each segment  contains the center tells us the region that contains the center.

The algorithm takes $O(|{Q}| + |{\cal T}|)$ time.
\end{proof}

\section{Extra Material for  Section~\ref{section:Phase-I}, Phase I
}

\subsection{Finding the Farthest Edge from each Vertex}
\label{appendix:Hershberger-Suri}

Phase I is to find the farthest edge Voronoi diagram restricted to the polygon boundary.  In this section we give the first step of Phase I:

\begin{thm}
There is a \reviewerchange{linear-time algorithm} to find the farthest edge from each vertex of a simple polygon.
\end{thm}

Hershberger and Suri~\cite{hershberger1997matrix} gave a \reviewerchange{linear-time algorithm} to find the farthest \emph{vertex} from each vertex in linear time.
We show that their algorithm extends to finding the farthest \emph{edge} from each vertex in linear time.
Hershberger and Suri build upon an algorithm called SMAWK due to Aggarwal et al.~\cite{aggarwal1987geometric} that finds row maxima in a totally monotone matrix in linear time. 
The SMAWK algorithm immediately solves the problem of finding the farthest vertex from each vertex in a convex polygon in linear time, but Hershberger and Suri need substantial new ideas to extend to general simple polgons. In order to extend Hershberger and Suri's algorithm to find the farthest edge from each vertex, we must examine their algorithm in more detail. 

We structure this section as follows:
\begin{enumerate}
\squeezelist
\item Use separators to reduce the farthest vertex/edge problem to a problem of finding all row maxima in a totally monotone matrix. 
The matrix is given implicitly---each entry in the matrix represents the distance from one vertex to a vertex/edge, and this distance is computed only when needed.
\item An overview of the SMAWK algorithm to find row maxima in a totally monotone matrix.
Together with item 1, this solves the problem of finding the farthest vertex from each vertex in a \emph{convex} polygon, because then each matrix entry (the distance between two vertices) can be computed in constant time.
\item An overview of the Hershberger-Suri algorithm that solves the problem of finding the farthest vertex from each vertex in a general simple polygon.  To do this, they show how to compute each matrix entry needed in the SMAWK algorithm in constant amortized time. 
\item The modifications required for finding the farthest \emph{edge} from each vertex.
\end{enumerate}

\paragraph*{Reducing farthest vertices/edges to row maxima in a matrix.}  We use
the notion of separators from Appendix~\ref{appendix:separators}.
Suri~\cite{suri1989} proved that there are a constant number of separators that separate every vertex from its farthest vertex.  In Lemma~\ref{lem:constant-separator-set} we extended this result to the farthest \emph{edge} from each vertex.
Thus, in either case, to find the farthest vertex/edge from each vertex it suffices to solve the following problem:  given a separator $\pi(a,b)$, find, for each vertex to the right of the separator, the farthest vertex/edge that lies  to the left of the separator.  

Consider a \defn{distance matrix} $M$ with rows indexed by the vertices to the right of the separator in counterclockwise
order and columns indexed by either the vertices or the edges to the left of the separator in counterclockwise order, and with $M(v,s)$ defined to be the geodesic distance from vertex $v$ to vertex/edge $s$.  Then we seek the maximum in each row of the matrix.  

A matrix $M$ is \defn{totally monotone} if for any $2 \times 2$ submatrix 
$\begin{bmatrix}
a & b \\
c & d
\end{bmatrix}$, 
if $b>a$ then $d > c$.

Hershberger and Suri prove that the distance matrix for farthest vertices is totally monotone.  We prove the analogous result for farthest edges.

\begin{claim}
The distance matrix $M$ for farthest edges as described above is totally monotone.
\end{claim}
\begin{proof}
Consider a $2 \times 2$ submatrix with rows indexed by vertices $u,v$ and columns indexed by edges $e,f$:
\[
\begin{blockarray}{ccc}
    & e & f \\
\begin{block}{c[cc]}
  u & a & b \\
  v & c & d \\
\end{block}
\end{blockarray}
 \]
Because the row order and column order are counterclockwise, and because $u,v$ are to the right of the separator and $e,f$ are to the left of the separator, $u,v,e,f$ occur in counterclockwise order around the polygon.
By Corollary~\ref{cor:ordering-property-paths}, if $d(u,f) > d(u,e)$ then $d(v,f) > d(v,e)$, i.e., if $b>a$ then $d>c$.
\end{proof}

Thus the problem of finding the farthest vertex/edge from each vertex is reduced in linear time to the problem of finding row maxima in an totally monotone distance matrix (where we must take into account the time required to access matrix entries).

\paragraph*{The SMAWK algorithm to find row maxima in a totally monotone matrix.}
Let $M$ be an $n \times m$ totally monotone matrix.
Break ties for the maximum value in a row by choosing the leftmost maximum. 
The positions of these row maxima progress rightward and downward---more precisely, if the maximum in row $i$ occurs in column $k$, then the maximum in row $j>i$ occurs in column $l \ge k$. 
The SMAWK algorithm~\cite{aggarwal1987geometric} finds the (leftmost) maximum in each row as follows:
\begin{enumerate}
\squeezelist
\item Delete columns (without eliminating any row maxima) to reduce to an $n \times m'$ matrix $M'$, where $m' \le n$.
This is accomplished by a routine called REDUCE that accesses $2m-n$ matrix entries. 

\item Let $M''$ consist of the even numbered rows of $M'$.  Recursively find the  row maxima in $M''$.  This gives us the row maxima for all even-numbered rows of $M'$.

\item Fill in the row maxima for the odd numbered rows of $M'$.  Observe that the column of the maximum in row number $2i+1$ occurs between the columns of the maxima in row numbers $2i$ and $2i+2$, which means that this step accesses $n + m'$ matrix entries, where the next access is below or to the right of the current one.
\end{enumerate}

Aggarwal et al.~\cite{aggarwal1987geometric} prove that the SMAWK algorithm runs in time $O(n+m)$ assuming that matrix entries can be accessed and compared in constant time.
The number of recursive calls (``phases'') is $O(\log n)$. An important property is that in step 1 and step 3 
each successive matrix entry access is to the right, or up, or down from the current one---in particular there are no left moves.
This is stated as Equation (2.3) by Hershberger and Suri~\cite{hershberger1997matrix}.

\paragraph*{The Hershberger-Suri Algorithm and Its Extension to Farthest Edges.}
As noted above, the SMAWK algorithm gives a \reviewerchange{linear-time algorithm} to find 
the farthest vertex from each vertex in a convex polygon, because in that case each entry in the distance matrix can be computed in constant time.
However, for a general simple polygon, each matrix access involves finding the distance between two vertices $v$ and $u$ on opposite sides of the separator.  
Hershberger and Suri show that this can be done in $O(1)$ amortized time per matrix access.  Their algorithm relies on the order of matrix accesses in the SMAWK algorithm as mentioned above, and on the properties of shortest paths that cross the separator $\gamma(a,b)$, as discussed in Appendix~\ref{sec:funnels-paths}.   We use the terminology and notation from Appendix~\ref{sec:funnels-paths}.
The shortest path $\pi(v,u)$ consists of edges of the funnels $Y(v)$ and $Y(u)$, with one additional tangent edge $\ell(u,v)$ in case the pair of funnels is open. 
The shortest path trees $T_a$ and $T_b$ can be preprocessed in linear time to allow constant time queries for least common ancestors, and for lengths of paths to $a$ or $b$.  Then the length of $\pi(u,v)$ can be found in constant time if the pair of funnels $Y(v),Y(u)$ is closed.
The same applies to our case of the shortest path from  vertex $v$ to edge $e$---for example, in Figure~\ref{fig:funnels-and-tangents}, the funnels $Y(v_2)$ and $Y(e)$ are closed and $d(v_2,e) = (d(a,v_2) - d(a,\gamma_a(v_2)) + (d(b,e) - d(b,\gamma_a(v_2))$.

When a pair of funnels is open the only hard part is finding their tangent edge. 
Given the tangent edge, the length of the path can be found in constant time.
For example, in Figure~\ref{fig:funnels-and-tangents}, the funnels $Y(v_1)$ and $Y(e)$ are open with tangent edge $\ell(v_1,e) = (x,y)$ of length $d_2(x,y)$ so  $d(v_1,e) = (d(b,v_1) - d(b,x)) + (d(a,e) - d(a,y)) + d_2(x,y)$.
Hershberger and Suri give a data structure to find the tangent between a pair of open funnels 
in constant amortized time
by
storing and  maintaining the 
walls of the funnels $Y(v)$ and $Y(u)$ during each phase of the algorithm as $v$ moves counterclockwise and $u$ moves in either direction.
In fact, it suffices to maintain the parts of the walls from the apex of the funnel to $\gamma$.

Binary search along the walls can be used to find the tangent edge
but this is too inefficient for a linear-time algorithm.
Therefore, a more complex data structure that modifies the shortest path trees ($T_a$ and $T_b$) at each phase is used.
Paths in the trees are broken into subpaths, and each subpath is represented by a \emph{supernode} that supports fast searching. 
Supernodes are stored as binary trees with the original polygon vertices at their leaves, and internal nodes representing the edge joining the subtrees below.
Any path of the shortest path tree in the $k$-th phase is a list of supernodes connected by superedges, such that every supernode has at most $2^k$ vertices of the original polygon.
Finally, Hershberger and Suri provide a method for obtaining the supernode representation for the trees before the $k$-th phase in time $O(k n / 2^k)$.
This takes $O(n)$ time for all the $O(\log n)$ phases and also ensures that tangents between open funnels in the $k$-th phase can be determined efficiently.
The maintenance of the supernode representation between phases is quite involved and we do not describe more details here.

The data structure permits them to find the tangent edge $\ell(v,u)$ and to update the funnels, in $O(1)$ amortized time per operation. 
Their algorithm and its  analysis depend on a lemma about the difference between two funnels. 
For two sites (vertices/edges) $s_i$ and $s_j$ on the same side of the separator $\gamma(a,b)$, the \defn{funnel-difference} is the set of edges in $Y(s_i)$ that do not occur in $Y(s_j)$. 
We observe that their result about funnel differences~\cite[Lemma 3.3]{hershberger1997matrix} extends to our situation and is crucial for the amortized analysis.

\begin{lem}
\label{lemma:funnel_difference}
The funnel difference of $s_i,s_j$ forms a path that includes the apex of $Y(s_i)$, and 
is edge-disjoint from $Y(s_k)$, for any vertex/edge $s_k$ 
that appears in the  order $s_i,s_j,s_k$ on 
the same side of the separator.
\end{lem}

In summary, the Hershberger-Suri algorithm extends in a straightforward manner to farthest edges. 
The only modifications needed are the extension to edge funnels (instead of funnels based on vertices) and the different number of separators.

\subsection{
Hourglasses}
\label{appendix:hourglasses}

\changed{In this section we show that to find the Voronoi diagram on a transition edge $ab$ it suffices to look at the \emph{hourglass} of $ab$, and we show that all the hourglasses can be found in linear time, and the sum of their sizes is linear.}

Hourglasses were first used in algorithms for shortest paths~\cite{SPT_linear,GUIBAS1989126,chazelle1989visibility}, and then used in algorithms to find the farthest vertex geodesic Voronoi diagram (Aronov et al.~\cite{aronov1993farthest}) and in algorithms to find the geodesic [vertex] center (Ahn et al.~\cite{linear_time_geodesic}).

\begin{figure}
\centering
\includegraphics[width=0.6\textwidth]{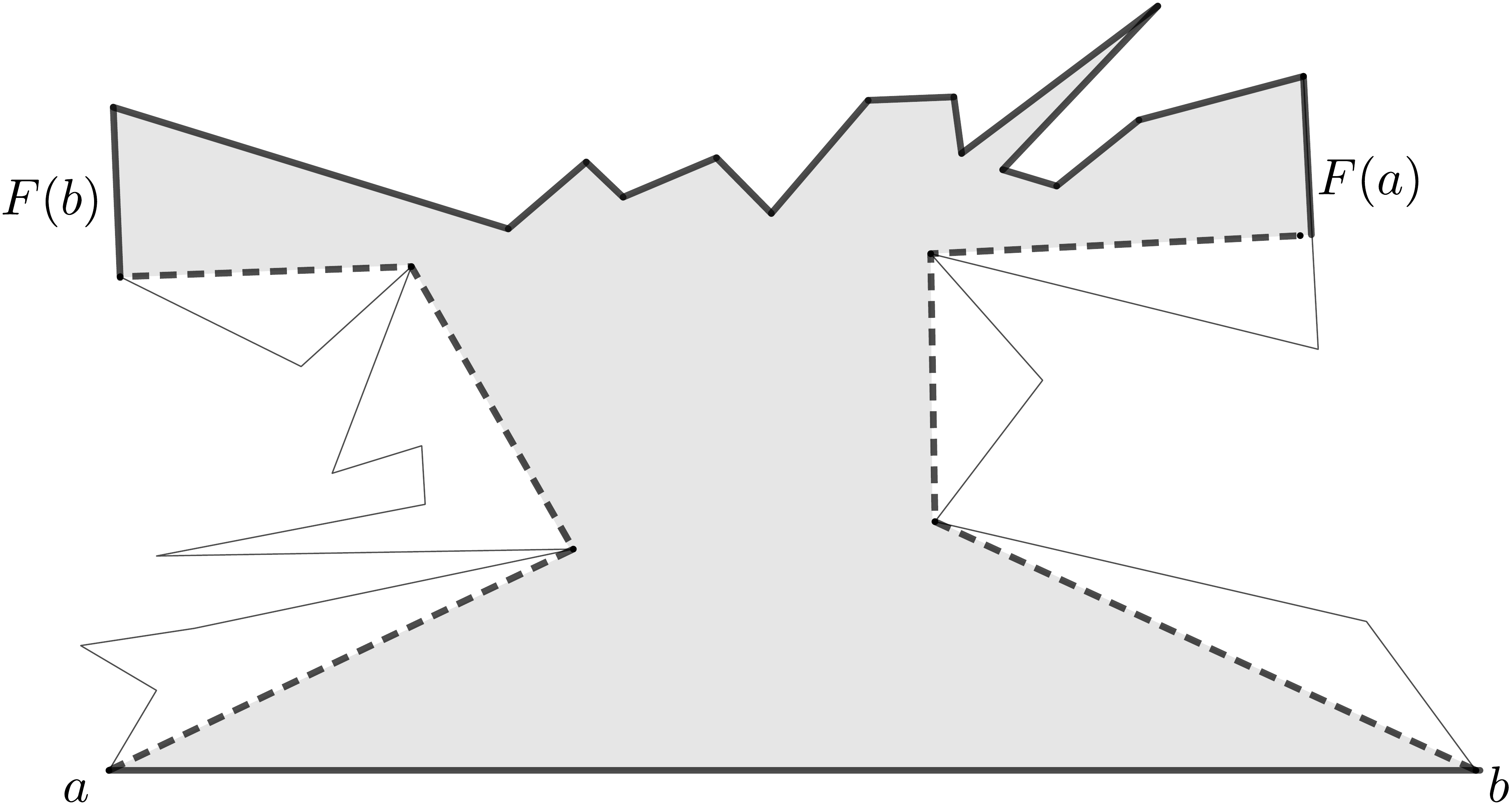}
\caption{
The hourglass for a transition edge $ab$.
The walls $\pi(a,F(b))$ and $\pi(b,F(a))$ are represented by dashed polylines.
}
\label{fig:hourglass}
\end{figure}

Let $ab$
be 
a transition edge directed counterclockwise.
Note that
$\pi(a,F(a))$ and $\pi(b,F(b))$ cross each other by
the Ordering Property~\ref{prop:anti-parallel} of Lemma~\ref{lem:ordering-properties}. 
The \defn{hourglass} $H(a,b)=H(e)$ is
the subpolygon of $P$ bounded by $ab$, $\pi(a,F(b))$, $\pi(b,F(a))$ and the clockwise portion of $\partial P$ between the terminals $t(a,F(b))$ and $t(b,F(a))$. See Figure~\ref{fig:hourglass}. 
Recall that: every vertex of $P$ has a unique farthest edge by Assumption~\ref{assumption:unique_farthest_neighbor}; and    $t(a,F(b))$ is the terminal point of the geodesic path from vertex $a$ to the edge $F(b)$.
The geodesics $\pi(a,F(b))$ and $\pi(b,F(a))$  are called the \defn{walls} of the hourglass $H(a,b)$, and the part of $\partial P$ clockwise from $t(b,F(a))$ to $t(a,F(b))$ is called the \defn{chain} of the hourglass.
The \defn{size} of an hourglass is its number of vertices.

The following lemma justifies restricting attention to the hourglass of a transition edge $ab$ in order to find the farthest edge Voronoi diagram restricted to $ab$.
It is a consequence of the Ordering Property from Lemma~\ref{lem:ordering-properties}.

\begin{lem}\label{lemma:hourglasses_have_necessary_information}
Let $p$ be a point on the transition edge $ab$ and  
let $e$ be a farthest edge from $p$ in $P$.  Then $e$ lies in the chain of the hourglass $H(a,b)$.
\end{lem}
\begin{proof}
Suppose $e$ lies in the clockwise chain of $\partial P$ from $F(a)$ to $b$. Then the clockwise ordering around $\partial P$ is $p,a,F(a),e$ in contradiction to the Ordering Property (Property~\ref{prop:anti-parallel} of Lemma~\ref{lem:ordering-properties}).  Similarly, if $e$ lies in the clockwise chain $\partial P$ from $a$ to $F(b)$, the clockwise ordering  $b,p,e, F(b)$ contradicts the Ordering Property.   
\end{proof}

Let $\cal H$ be the set of  hourglasses of all the transition edges of $P$. 
In the remainder of this subsection we show how to find 
$\cal H$ in linear time.

\begin{lem}\label{Total_Size_of_All_Hourglasses_Lemma}
All the hourglasses of $\mathcal{H}$ can be constructed in $O(n)$ time.
In particular, the sum of 
their sizes
is $O(n)$.
\end{lem}

\begin{proof}
Recall that by Lemma~\ref{lem:constant_splits} there are five farthest edge separators such that for every edge $ab$, one of the separators has $a$ and $b$ to its right and $F(a)$ and $F(b)$ to its left.
Let $\gamma = \pi(p,q)$ be a farthest edge separator and let 
${\cal H}_\gamma$ be the set of hourglasses of transition edges that lie to the right of $\gamma$.
It suffices to 
prove the lemma for one set ${\cal H}_\gamma$.
Each hourglass in $\mathcal{H}_{\gamma}$ consists of a transition edge $ab$ to the right of $\gamma$, two walls, and a chain to the left of $\gamma$.
In 
Appendix~\ref{appendix:Hershberger-Suri}
we found the farthest edge from each vertex in linear time, 
so we 
know $F(a)$ and $F(b)$.
Because the hourglass chains are internally disjoint, we just need to show that we can find all the walls of the hourglasses in $\mathcal{H}_{\gamma}$ in linear time.  

Each wall is a shortest path from a vertex to the right of the separator $\gamma=\pi(p,q)$ to an edge to the left of $\gamma$, so by Lemma~\ref{lem:separator-paths} each wall consists of edges of  the shortest path trees $T_p$ and $T_q$, except for at most one edge crossing $\gamma$. 
The set of crossing edges has size $O(n)$ because there are $O(n)$ hourglasses. 
By Lemma~\ref{SPT_lemma} the shortest path trees $T_p$ and $T_q$ can be found in 
time $O(n)$ and have size $O(n)$.
By Lemma~\ref{lemma:walls_can_be_constructed_fast} we can find each wall in time proportional to the size of the wall. 
Thus we can find all the walls in time $O(n)$ so long as we show that each polygon chord is in $O(1)$ walls.
(Note that walls of hourglasses are not paths to farthest edges, so we cannot simply apply Property~\ref{prop:no-shared-chord} that crossing paths to farthest edge do not share directed chords.)

\begin{claim}
Any chord of the polygon is in $O(1)$ walls of hourglasses of $\mathcal{H}_{\gamma}$. 
\label{claim:hourglass-wall}
\end{claim}

\begin{proof}
Let $t_1, \ldots, t_k$ be the transition edges to the right of $\gamma$ in clockwise order. 
If two transition edges are close together in this ordering, then their walls may have common chords, but we will show that if $t_i$ and $t_j$ are separated by at least three transition edges, i.e., 
$j-i \ge 4$, then the walls of the hourglasses $H(t_i)$ and $H(t_j)$ have no common chords.  Note that this proves the Claim.

So, consider $t_i$ and $t_j$ with $j-i \ge 4$, and suppose for a contradiction that a wall of $H(t_i)$ and a wall of $H(t_j)$ share a common chord $f$.  Take the intermediate transition edge $t_k = (u,v)$, where $k = i+2$.  Then the farthest edges of the endpoints of $t_i, t_j$ and $t_k$ are all distinct (this is the reason for choosing $i,j,k$ as we did), and all lie to the right of $\gamma$. 
We will show that the paths $\pi(u,F(u))$ and $\pi(v,F(v))$ also share the chord $f$.  This means that we have crossing paths to distinct farthest edges and the paths share a chord, which contradicts Property~\ref{prop:no-shared-chord} from Lemma~\ref{lem:ordering-properties}.

\begin{figure}[ht]
\centering
\includegraphics[width=.75\textwidth]{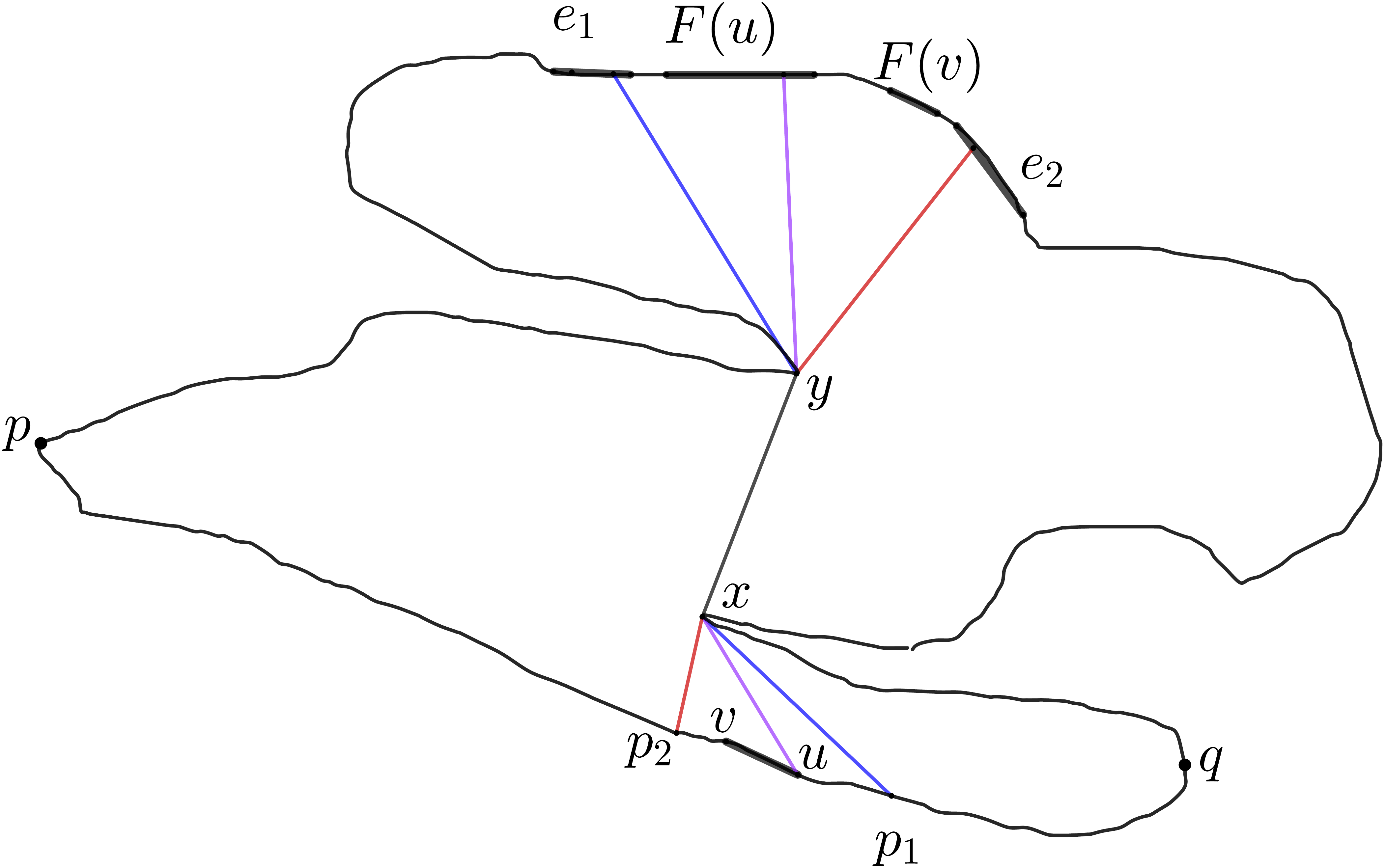}
\caption{Illustration for Claim~\ref{claim:hourglass-wall}. The walls $\pi(p_1,e_1)$ and $\pi(p_2,e_2)$ share the chord $f = xy$, which forces the path $\pi(u,F(u))$ to also use $f$.
}
\label{fig:shared-walls}
\end{figure}

It remains to show that $\pi(u,F(u))$ uses the chord $f$. (The case of $\pi(v,F(v))$ is exactly the same.)  The idea is that this path is ``squashed'' between the two walls that use $f$.
Suppose that the wall $\pi(p_1, e_1)$ of $H(t_i)$ and the wall $\pi(p_2,e_2)$ of $H(t_j)$ both use chord $f$. 
See Figure~\ref{fig:shared-walls}.
Here $p_1$ and $p_2$ are distinct vertices on the right %
of the separator $\gamma$ and $e_1$ and $e_2$ are distinct edges on the 
left %
of $\gamma$.  Vertex $u$ lies between $p_1$ and $p_2$ in clockwise order, and $F(u)$ lies between $e_1$ and $e_2$ in clockwise order, and all are distinct.  
Let $f = xy$ where $x$ and $y$ are vertices of the polygon.
Because shortest paths to the same destination do not cross, 
the shortest path from $u$ to $y$ goes through $x$.  Similarly, the shortest path from $x$ to the edge $F(u)$ goes through $y$.  The union of these two shortest paths is a geodesic (locally shortest) path from $u$ to $F(u)$ and must therefore be the shortest path from $u$ to $F(u)$.  Thus $\pi(u,F(u))$ uses the edge $f$. 
\end{proof}

This completes the proof of Lemma~\ref{Total_Size_of_All_Hourglasses_Lemma}.
\end{proof}

\subsection{Details for Section~\ref{sec:one-edge-Vor}, Farthest Edge Voronoi Diagram on One Edge}
\label{appendix:ab-course-cover}
\label{appendix:one-edge-Vor}

We need some more results to 
prove Lemma~\ref{lem:consecutive-coarse-cover}. 
First we need more details of the construction of the coarse cover~\cite{lubiw2021visibility}.
Define \defn{$p_a(u)$} and \defn{$p_b(u)$} to be the parents of node $u$ in $T_a$ and $T_b$, respectively.
As noted by 
Pollack et al.~\cite{pollack_sharir}, a vertex $u$ is visible from some point on $ab$
if and only if $p_a(u) \ne p_b(u)$. 
If $u$ is visible from some point on $ab$, then extending the edge from $u$ through $p_a(u)$ reaches a point
\defn{$x_a(u)$}
on $ab$ from which $u$ is visible.
Similarly, extending the edge from $u$ through $p_b(u)$ reaches a point
\defn{$x_b(u)$}
on $ab$ from which $u$ is visible.
According to the definition \reviewerchange{in\cite{lubiw2021visibility}}, if
edge $uv$ of $T_a$
has an associated 
$a$-side coarse cover element $(I,f,e)$, then 
$I=[x_a(u), x_a(v)]$.
Similarly for $b$-side elements.
\reviewerchange{If
edge $uv$ of $T_a \cap T_b$
has an associated 
central triangle coarse cover element $(I,f,e)$, then 
$I=[x_a(u), x_b(u)]$.
And if polygon edge $e$ has an associated central trapezoid coarse cover element $(I,f,e)$, then $I$ consists of the points of $ab$ whose shortest paths to $e$ arrive perpendicularly, and with the added $0$-length edges this is
$I = [x_a(t(a,e)),x_b(t(b,e))]$.
}

\begin{obs}
\label{obs:consecutive-coarse-cover}
If $uv$ and $vw$ are edges of $T_a$ that have associated coarse cover elements $(I_1, f,e)$ and $(I_2,f',e')$, then
the right endpoint of $I_1$ is $x_a(v)$ and the left endpoint of $I_2$ is $x_a(v)$, i.e., 
\reviewerchange{$I_1$ and $I_2$ appear in that order along $ab$ and intersect in a single point.  
This observation is also true for an edge $uv$ and a central trapezoid at $v$ if $v$ happens to be a leaf.}
A similar property holds for $T_b$.
\end{obs}

\begin{lem}
\label{lem:coarse-cover-path}
Suppose edge $uv$ of $T_a$ has an associated coarse cover element
\changed{$(I,f,e)$ for polygon edge $e$.}
Then:

\begin{enumerate}
\item 
On the path $\pi(a,u)$ all edges except the first one have associated coarse cover elements.

\item On the path $\pi(v,e)$ let $x$ be the last vertex visible from $ab$.  All edges on $\pi(v,x)$ have associated coarse cover elements for the polygon edge $e$.  Furthermore, if $x$ is a leaf then there is a central trapezoid associated with $e$, and otherwise there is an edge $xy$ in $\pi(v,e)$ and it is associated with a central triangle for $e$.
\end{enumerate}
A similar property holds for $T_b$.
\end{lem}

\begin{proof}
The first statement just depends on the fact that if $u$ is visible from $ab$ (i.e., has different parents in $T_a$ and $T_b$) then the same is true for every vertex on the path $\pi(a,u)$.

For the second statement, note that $e$ is the farthest edge from $v$ in the subtree of $v$.
Let $w$ be any vertex on $\pi(v,x)$, and $e'$ any polygon edge for which the terminal of the path $\pi(a,e')$ 
lies in the subtree of $w$.
We have $d(v,e') = d(v,w) + d(w,e')$.
Since $d(v,e) > d(v,e')$, the previous equality implies $d(w, e) > d(w,e')$.
So the farthest edge from $w$ is $e$,
and the coarse cover element for the tree edge joining $w$ to its parent $p_a(w)$ is also associated with $e$.

If $x$ is a leaf, the terminal points 
$t(a,e)$ and $t(b,e)$ 
are distinct (due to the introduction of 0-length segments).
From the definition of the coarse cover elements, this means there is a central trapezoid associated with $e$.
If $x$ is not a leaf, let $xy$ be the first edge on the path $\pi(x,e)$ (which is a subpath of $\pi(v,e)$).
The vertex (or terminal point) $y$ is not visible from $ab$ and $xy$ is a tree edge in both $T_a$ and $T_b$.
The coarse cover element for $xy$ is a central triangle associated with $e$.
\end{proof}

Lemma~\ref{lem:consecutive-coarse-cover} follows immediately from Lemma~\ref{lem:coarse-cover-path}.

\subparagraph*{Details on Constructing Tree $T$.}

\begin{enumerate}
\squeezelist
\item For each central trapezoid coarse cover element, say associated with polygon edge $e$, there is a leaf $l$
of $T$ corresponding to $e$. Attach a new edge in $T$ descending from $l$ and associate the central trapezoid element with it.

\item 
For each polygon edge $e$ that has $b$-side triangle elements associated with it, those triangles 
correspond to a  path $\pi$ in $T_b$, that is directed in leaf-to-root order by Lemma~\ref{lem:consecutive-coarse-cover}.  The tree $T$ currently has an edge, say $g$, associated with the central triangle/trapezoid for $e$.  Attach the path $\pi$ at end of the edge $g$.

\item 
Finally, we  contract any original edge of $T_a$ that is not associated with a coarse cover element.
\changed{These are edges
$uv$ such that $u$ is not visible from $ab$ plus the original edges incident to $a$.
}
\end{enumerate}

\section{Extra Material for Section~\ref{section:PhaseII}, Phase II
}

\subsection{Finding a Coarse Cover of the Polygon by Triangles} \label{appendix:coarse-cover}

This section contains the
first step of Phase II, which is 
to construct in linear time a \defn{coarse cover} 
of the polygon
as specified in Definition~\ref{defn:coarse-cover}.

Let $X$ be the set of vertices of the farthest edge Voronoi diagram on the boundary of $P$.  These are points on $\partial P$ that each have two farthest edges. The points of $X$ partition $\partial P$ into \defn{chains $C(e)$}, where $C(e)$ consists of the points on $\partial P$ whose farthest edge is $e$. We begin by expanding each chain $C(e)$  to a subpolygon
$Y(e)$
that contains the Voronoi region of the edge $e$.
After that, we will partition each polygon $Y(e)$ into triangles to obtain the coarse cover.

Suppose the chain $C(e)$ goes clockwise from $p(e) \in X$ to $q(e) \in X$.  The \defn{edge funnel} \defn{$Y(e)$} is the polygon  bounded by the chain $C(e)$, the \defn{walls} $\pi(p(e),e)$ and  $\pi(q(e),e)$, and the \defn{base} $b(e)$, which is the part
of $e$ between the terminal points $t(p(e),e)$ and $t(q(e),e)$.
See Figure~\ref{fig:flower}.
The \defn{size} of the edge funnel $Y(e)$ is its number of  vertices.
By Property~\ref{prop:path-to-same-edge}, the walls of an edge funnel may merge but never cross, so each edge funnel is a weakly simple polygon.

Edge funnels are an extension of the well-studied vertex funnels that are used for computing shortest paths (see~\cite{lee1984euclidean,SPT_linear}), and that were used by Ahn et al.~\cite{linear_time_geodesic} to compute the geodesic vertex center.
To build their coarse cover of the polygon, Ahn et al.~needed hourglasses as well as [vertex] funnels, so their method was more complicated. 
By contrast, our coarse cover is constructed from edge funnels alone because we did extra work ahead of time using hourglasses to compute the farthest edge Voronoi diagram restricted to the polygon boundary.

\begin{lem}
\label{lem:voronoi_containment}
For any point $p \in P$, if $e$ is a farthest edge from $p$, then $p \in Y(e)$.
\end{lem}
\begin{proof}
Consider the shortest path $\pi(p,e)$ from the 
point $p$ to the edge $e$
and extend the first segment of the path backwards 
until it intersects the boundary $\partial P$ at point $p'$.
Since the result is a locally shortest path, it must be the shortest path from $p'$ to $e$.
Thus the distance from $p'$ to $e$ is $|p'p| + d(p,e)$.
We now show that $e$ is a farthest edge from $p'$.
Consider any other edge $e'$. We have  $d(p,e') \leq d(p,e)$.
Then
$d(p',e') \leq |p'p|  + d(p,e') \leq |p'p| + d(p,e)  = d(p',e)$.
Thus $e$ is a farthest edge from $p'$ so  $p'$ 
lies on the chain $C(e)$.
By Property~\ref{prop:path-to-same-edge} of Lemma~\ref{lem:ordering-properties},
$\pi(p',e)$ 
does not cross the walls of the edge funnel $Y(e)$, so it lies inside $Y(e)$.
Therefore, the point $p$ lies in $Y(e)$ since $p$ is a point on $\pi(p',e)$.
\end{proof}

\begin{lem}
\label{lemma:flowers_total_linear}
The set of  edge funnels $Y(e)$ corresponding to all the edges $e$ of the polygon can be constructed in $O(n)$ time.  The sum of all their sizes is $O(n)$.
\end{lem}
\begin{proof}
The farthest edge Voronoi diagram on $\partial P$ gives us the chains $C(e)$, so we only need to find the walls of the edge funnels, which are the shortest paths from the endpoints of $C(e)$ to $e$.
Equivalently, we must find,
for each Voronoi vertex $p \in X$, the shortest paths to $p$'s farthest edges. Note that there are $O(n)$ points in $X$, and 
by Lemma~\ref{lem:1D-Voronoi-edges}, 
each $p \in X$ has two farthest edges.

Recall that by Lemma~\ref{lem:constant_splits} there is a \reviewerchange{linear-time algorithm} to find a set of five farthest edge separators such that for every point $p \in \partial P$, one of the separators has $p$ to its 
right  
and, by definition of a farthest edge separator, has all farthest edges of $p$ to the 
left. . 
It therefore suffices to focus on one of these farthest
edge separators $\gamma = \pi(a,b)$, 
and give a \reviewerchange{linear-time algorithm} to find the shortest path from each point $p \in X$ that lies to the right of $\gamma$ to each of $p$'s farthest edges. 
By Lemma~\ref{lem:separator-paths}, each such path $\pi(p,F(p))$ is contained (except for one edge) in the shortest path trees $T(a)$ and $T(b)$. By Lemma~\ref{lemma:walls_can_be_constructed_fast},  
after a linear time preprocessing of the trees $T(a)$ and $T(b)$, each path $\pi(p,F(p))$ can be found in time proportional to its
number of edges, which we denote by $| \pi(p,F(p)) |$. 
Finally, we note that 
the two walls of one edge funnel may share edges, but we claim that walls of different edge funnels do not share edges if they originate from points in $X$ to the right of $\gamma$. 
\newchanged{Consider $p, q \in X$ and the paths $\pi(p,F(p))$, $\pi(q,F(q))$ with $F(p) \ne F(q)$. By Claim~\ref{claim:across-separator}, the paths cross, and then Property~\ref{prop:no-shared-chord} implies that the paths do not share any edges (chords).
}

We therefore have 
$\sum | \pi(p,F(p)) | \in O(|T(a)| + |T(b)| + n)$, 
where the last term accounts for the one edge of each path that is not in the trees. Thus the total run time to find all the shortest paths is $O(n)$.
\end{proof}

\subparagraph*{Defining the coarse cover $\cal T$ of polygon $P$.}
The idea is to partition each funnel $Y(e)$ into triangles in time linear in $O(|Y(e)|)$, and then take the union over all funnels.
We first use Lemma~\ref{ShortestPathForest_lemma} to   partition  $Y(e)$ in linear time into its \emph{shortest path map} from its base edge $b(e)$.  Recall that
the shortest path map 
partitions $Y(e)$ into regions such that shortest paths to $e$  from points in the same region are combinatorially the same. 
In addition,
if a region of the shortest path map contains any vertex $v$ whose shortest path to $e$ splits the region, then we subdivide the region at the path. All these subdivisions can be found in linear time, and the resulting subdivided regions $T$  
are either triangles or trapezoids; see Figure~\ref{fig:coarse-cover}.

Next we define distance functions on the triangles and trapezoids.
If $T$ is a triangle, then the shortest path to $e$ from any point $p \in T$ goes through an apex $v$ of the triangle, and the distance from $p$ to $e$ is $f_T(p) := d_2(p,v) + \kappa$ where $d_2$ is Euclidean distance (ignoring $P$) and  $\kappa$ is $d(v,e)$ which is independent of $p$.
If $T$ is a trapezoid, then the shortest path to $e$ from any point $p$ of $T$ is a straight line segment meeting $e$ at right angles, and the distance from  $p $ to $e$ is $f_T(p) := d_2(p,{\bar e})$, where ${\bar e}$ is the line through $e$ and $d_2$ is Euclidean distance (ignoring $P$). 
For the convenience of not having to say ``triangles and trapezoids,'' we will further partition each trapezoid into two triangles, each inheriting a distance function of the form $d_2(p,{\bar e})$.

Define \defn{${\cal T}(e)$} to contain
the triple $(T, f_T, e)$ for each triangle $T$ in the partition of $Y(e)$.
Along with the triples we store the following:

\begin{enumerate}
\squeezelist
\item Each triangle $T$ 
is given by its three sides: one side is a subsegment of an edge and the other two are chords (recall that a chord may include, or be, an edge of $P$).  A chord is given by its two endpoints and the vertex/edge containing each endpoint.  
\item Furthermore, we store a list of chords used as  sides of triangles of ${\cal T}(e)$, and for each chord, list the one or two triangles it is a side of.
Each chord is given by its two endpoints on $\partial P$.
\end{enumerate}

\begin{claim}
\label{claim:cover-one-funnel}
${\cal T}(e)$ can be computed in time $O(|Y(e)|)$, and has size $O(|Y(e)|)$.
\end{claim}

Define 
\defn{${\cal T}$} 
to be $ \bigcup_e {\cal T}(e)$. 
We prove that $\cal T$ is a coarse cover of $P$ according to Definition~\ref{defn:coarse-cover}.
(Note that a chord may appear as a side of triangles in more than one ${\cal T}(e)$.  We could, in fact, identify these, but instead we simply allow duplicates in the list of chords.)

\begin{lem} 
\label{lem:coarse-cover}
The set $\cal T$ 
is a coarse cover of $P$.
Furthermore, $\cal T$
has size $O(n)$ and  can be computed in time $O(n)$.  
\end{lem}
\begin{proof}
From Lemma~\ref{lemma:flowers_total_linear}, we can construct all the edge funnels $Y(e)$ in time $\sum_e |Y(e)| \in O(n)$, and from Claim~\ref{claim:cover-one-funnel}, we
can compute ${\cal T}(e)$ in time $O|Y(e)|)$.  Thus we can compute $\cal T$ in time $O(n)$.

To prove that $\cal T$ is a coarse cover, first observe that the functions $f_T$ have the correct forms. 
By Lemma~\ref{lem:voronoi_containment}, for any point $p \in P$, if $e$ is a farthest edge from $p$, then $p$ is in the edge funnel $Y(e)$ so $p$ is contained in some triangle $T$ in the partition of $Y(e)$, and is therefore contained in a triple $(T,f_T,e)$ of ${\cal T}(e)$.
\end{proof}

\subsection{
Details for Section~\ref{section:final_edge_center}, Stage 1 Subproblems
}
\label{appendix:Q-bounds}
\changed{Stage 1 of the 
algorithm to find the edge center recurses on subproblems, each consisting of 
a subpolygon $Q$ that is a \newchanged{simple
$3$-anchor hull,} %
together with the coarse cover elements that intersect the interior of $Q$. 
We give some properties of these.}

\begin{obs}
\label{obs:3-geo-cell}
Let $Q$ be a 
\newchanged{$3$-anchor hull.}

\begin{enumerate}
\squeezelist

\item $Q$ is a closed connected weakly-simple polygon, and is \defn{geodesically convex in $P$}, meaning that for any two points $a$ and $b$ in $Q$, the geodesic path from $a$ to $b$ in $P$ is contained in $Q$.  This implies that the intersection of $Q$ with a chord [geodesic] of $P$ is a chord [geodesic] of $Q$.

\item The boundary of $Q$ consists of: 
\newchanged{the at most three anchors that are subchains of $\partial P$;
and at most three geodesic paths between pairs of anchors.
}

\item 
\newchanged{If $v$ is a vertex of $Q$ but not a vertex of $P$, then $v$ is a point anchor or the endpoint of an anchor chain that is interior to an edge of $P$.  In either case, $v$ is a convex vertex of $Q$.}

\end{enumerate}
\end{obs}

Each subproblem consists of the following.
\begin{enumerate}
\squeezelist
\item $Q$, a 
\newchanged{simple $3$-anchor hull} %
of $P$
that contains the geodesic edge center in its interior.
\item the set \defn{${\cal T}(Q)$}
of all coarse cover elements whose triangles intersect the interior of $Q$. 
Each triangle of the coarse cover is specified by its two defining chords of $P$ and the subsegment of an edge of $P$ that forms its third side.
Note: In this section we will
refer to triangles of the coarse cover  elements of ${\cal T}(Q)$ as
``triangles of ${\cal T}(Q)$.'' 

\item 
the set \newchanged{${\cal K}_T(Q)$ of 
defining chords of triangles of ${\cal T}(Q)$},
each given by its endpoints on the boundary of $P$, and each recording the one or two triangles of ${\cal T}(Q)$ that it is a side of.  
We also maintain the subset \newchanged{${\cal K}(Q)$ of chords that cross $Q$}, each given by its endpoints on the boundary of $Q$ (as well as its endpoints on the boundary of $P$).
\end{enumerate}

To solve a subproblem in Stage 1 means finding a point in $Q$ that minimizes the upper envelope of the functions of the coarse cover  $\mathcal{T}(Q)$, or reducing to $|Q| < 6$.

Define \defn{$t(Q)$} $:= |\mathcal{T}(Q)|$.
The \defn{size of a subproblem} is $|Q| + t(Q)$, where $|Q|$ is the number of vertices of $Q$ (as a polygon).
Initially, $Q$ is $P$,
$\mathcal{T}(Q)$ is all of $\cal T$, and ${\cal K}_T(Q)$ and ${\cal K}(Q) $ are all of $\mathcal{K}$.  
The size of the initial problem is $O(n)$, because $\cal T$ has linear size by Lemma~\ref{lem:coarse-cover}. 
Our goal is to  spend linear time in the size of a subproblem to reduce the size by a constant fraction.

We need some results about the size of $Q$.

\begin{lem}
\label{lem:floating_cells_constant_size}
If \newchanged{a simple $3$-anchor hull}
$Q$ is contained in a 
triangle of the coarse cover then $|Q| \le 6$.
\end{lem}
\begin{proof}
Let $T$ be a triangle of ${\cal T}(Q)$
that contains $Q$.
We claim that the boundary of $Q$ has at most three edges that are subsegments of edges of $P$.  Any such segment must lie on the boundary of $T$, and each of the three sides of $T$ can contain at most one such segment by our assumption that no three vertices of $P$ are collinear
(Assumption~\ref{assumption:no_three_points_collinear}).

We next claim that each of the at most three
geodesic chains on the boundary of $Q$ consists of a single segment.  This is because an internal vertex $v$ of a geodesic chain is a vertex of $P$, which must then be on the boundary of $T$ (since no point on the boundary  of $P$ lies in the interior of $T$). 
But then the internal angle of $Q$ at $v$ is $\le \pi$, so $v$ is not an internal vertex of a geodesic path.

Thus $Q$ has at most six edges.
\end{proof}

\begin{claim}
\label{observation:cell_triangles}
\newchanged{If $Q$ is  a simple $3$-anchor hull, then}
$|Q| \le 3 t(Q) + 6 \le 9 t(Q)$.
\end{claim}
\begin{proof}
\newchanged{Because $Q$ is simple, every vertex $v$ of $Q$ has interior points of $Q$ in its neighbourhood, so $v$}
must be contained in some triangle of ${\cal T}(Q)$ since ${\cal T}(Q)$ is 
a coarse cover of $Q$. 
By Observation~\ref{obs:3-geo-cell}, all but 6 of the vertices of $Q$ are vertices of $P$.

To complete the proof we show that each triangle $T$ of the coarse cover contains at most three vertices of $P$.
No vertex of $P$ is internal to $T$. Since $P$ does not have 3 collinear vertices (by Assumption~\ref{assumption:no_three_points_collinear}), each side of $T$ contains at most 2 vertices of $P$.  Furthermore, one side of $T$---call it $s_1$---is a subsegment of an edge of $P$, so it cannot contain vertices in its interior. 
Triangles of the coarse cover either have a vertex of $P$ as an apex opposite $s_1$, or arise from subdividing a trapezoid (see Appendix~\ref{appendix:coarse-cover}).  In the first case, $T$ has at most one more vertex on each side incident to the apex for a total of at most three vertices of $P$.  In the second case, $T$ has a side that is a diagonal of a trapezoid and contains no vertices in its interior, though it may have a vertex of $P$ at its intersection with $s_1$, and 
the third side of $T$ has at most two vertices of $P$, for a total of at most three vertices of $P$.
Thus $T$ contains at most three vertices of $P$.

This shows that $|Q| \le 3t(Q) + 6$.  For the second part of the inequality, note that $t(Q) \ge 1$.
\end{proof}

We note that the above Claim depends on the assumption that $Q$ is a $3$-anchor hull of $P$.
If we constructed $3$-anchor hulls of $3$-anchor hulls, then
the number of vertices that are not vertices of $P$ would grow.

We also need the following relationships between the number of  chords  and the 
number of coarse cover triangles.

\begin{claim}
\label{observation:chords_triangles}
$|\mathcal{K}(Q)| \leq |\mathcal{K}_T(Q)| \leq 2 t(Q)$.  If $Q$ is not contained in a triangle of ${\cal T}(Q)$, then $t(Q) \le 2|{\cal K}(Q)|$.
\end{claim}
\begin{proof}
For the first inequality, ${\cal K}(Q) \subseteq {\cal K}_T(Q)$ and every triangle of the coarse cover has two chords (the third side is a piece of a polygon edge).  For the second inequality, since no triangle of ${\cal T}(Q)$ contains $Q$, each one has at least one chord that crosses $Q$, and each chord comes from the coarse cover ${\cal T}(e)$ of a  
funnel $Y(e)$ and is a side of one or two coarse cover triangles in ${\cal T}(e)$.  (If a chord arises from more more than one $Y(e)$, we duplicate it in $\cal K$, see the definition of ${\cal T}(e)$ in
Appendix~\ref{appendix:coarse-cover}.)
\end{proof}

\subsection{Details for Section~\ref{section:largeQ}, Stage 1: Algorithm for Large $Q$}
\label{appendix:Algorithm-Stage1}

In this section we give an algorithm to handle a subproblem corresponding to a
subpolygon $Q$ 
\newchanged{(a simple 3-anchor hull)} with $|Q| >6$ 
and its associated sets ${\cal T}(Q)$, ${\cal K}_T(Q)$, and ${\cal K}(Q)$. 
By Lemma~\ref{lem:floating_cells_constant_size}, no triangle of ${\cal T}(Q)$
contains $Q$, so every triangle of ${\cal T}(Q)$ has a chord in ${\cal K}(Q)$.
The algorithm either finds the edge center or reduces to a subproblem with $|Q| \le 6$ which is handled in 
Appendix~\ref{section:constantQ}.
The idea was described in the main text.

\begin{enumerate}
\squeezelist

\item For 
$\epsilon = \frac{1}{80}$,
construct an $\epsilon$-net $N$ for the 
\newchanged{$3$-anchor}
range space
with ground set ${\cal K}(Q)$.  The range space is defined with respect to $3$-anchor hulls of $P$.

\item Compute the arrangement $A$ of the chords $N$
inside $Q$, and
use the Chord Oracle of Lemma~\ref{lem:generalized-chord-oracle} to find the face $F$ of $A$ that contains the edge center.

\item 
Partition  face $F$ into a constant number of
\newchanged{$3$-anchor hulls}
of $P$.

\item
Use the Geodesic Oracle (Lemma~\ref{lem:geodesic-oracle}) to find
\newchanged{which of these $3$-anchor hulls contains the edge center, and to reduce it to a simple $3$-anchor hull $Q'$.}

\item If $|Q'| \le 6$ then test each triangle of ${\cal T}(Q)$ to find ${\cal T}(Q')$ and ${\cal K}_T(Q')$, and switch to Stage 2 in the next subsection.

\item Otherwise $|Q'| > 6$.
Find ${\cal K}(Q')$  and ${\cal T}(Q')$, and recurse on the subproblem for $Q'$.

\end{enumerate}

We elaborate on these steps and their run-times below, but first we justify that our choice of $\epsilon$ in Step 2 guarantees that the size of the subproblem we recurse on is reduced by a fraction.
Recall that the size of the subproblem for $Q$ is $|Q| + t(Q)$.

\begin{lem}
For 
$\epsilon = \frac{1}{80}$,
if $|Q'| > 6$, then 
$|Q'| + t(Q') \le \frac{1}{2} (|Q| + t(Q))$.
\end{lem}
\begin{proof} 
Since $|Q'| > 6$, no triangle of ${\cal T}(Q')$ contains $Q'$.
Thus, since ${\cal K}(Q') \cap N = \phi$, the defining property of $\epsilon$-nets (equation~\ref{eqn:epsilon-net}), ensures that
$|{\cal K}(Q')| \le \frac{1}{80}  |{\cal K}(Q)|$,
which we relate to the subproblem sizes as follows.
\begin{align*}
|Q'| + t(Q') & \le 9 t(Q') + t(Q') & \text{by Claim~\ref{observation:cell_triangles}}\\
& = 10 t(Q') \le 20 |{\cal K}(Q')| & \text{by Claim~\ref{observation:chords_triangles} (no triangle of ${\cal T}(Q')$ contains $Q'$)}\\
& \le \tfrac{20}{80} |{\cal K}(Q)| = \tfrac{1}{4}  |{\cal K}(Q)|& \text{by the $\epsilon$-net property}\\
& \le \tfrac{1}{2} t(Q) & \text{by Claim~\ref{observation:chords_triangles}}\\
& \le \tfrac{1}{2}(|Q| + t(Q))\\
\end{align*}
\end{proof}

We now fill in more details of the steps of the algorithm, and justify that the runtime is $O(|Q| + t(Q))$.

\smallskip
\noindent{\bf 1. Construct an $\epsilon$-net.}
Lemma~\ref{lem:constant_shattering_dimension}
proves that the
$3$-anchor 
range space has bounded VC-dimension, and
Lemma~\ref{lem:subspace-oracle} proves that a subspace oracle exists.  This implies (see
Lemma~\ref{lem:constant_size_net})
that we can find a constant
sized $\epsilon$-net for this range space in time proportional to the size of the ground set, which is 
$O(| \mathcal{K}(Q) |)$ in our case.
By Claim~\ref{observation:chords_triangles} this is $O(t(Q))$.

\smallskip
\noindent{\bf 2. Compute the arrangement of $A$ in $Q$ and find the face $F$ that contains the edge center.}
Once the constant sized $\epsilon$-net  $N$ is determined, we can construct the arrangement of the chords in $O(|N|^2) = O(1)$ 
time, using the algorithm of Edelsbrunner et al~\cite{edelsbrunner1986constructing}.
Note that we know the endpoints of each chord of $N$ on $\partial Q$.
We run the chord oracle of Lemma~\ref{lem:generalized-chord-oracle} on each chord of $N$ inside polygon $Q$ to determine the face $F$ that contains the edge center (halting if we find the center on one of the chords).
This takes $O(t(Q))$ time for each chord of $N$. 
Since $N$ has constant size, 
this step takes $O(t(Q))$ time.

\smallskip
\noindent{\bf 3. Partition $F$ into 
\newchanged{$3$-anchor hulls.}
}
The boundary of 
$F$ consists of $O(1)$ segments of chords in $N$, $O(1)$ subchains of the geodesics bounding $Q$, and  $O(1)$ subchains of the polygon $P$. Let 
\newchanged{$V = \{v_0, \ldots, v_t \}$ be the points in order around $\partial F$ that join successive segments/subchains}.
Then $V$ has size $O(1)$.  %
Find shortest paths 
$\gamma_i = \pi(v_0, v_i)$, $i = 1, \ldots, t$
in $F$. 
This takes time $O(|F|)$, which is $O(|Q|)$.

Let $\Gamma$ be the set of these $O(1)$ shortest (geodesic) paths. 
Because $F$ is geodesically convex, each shortest path $\gamma_i \in \Gamma$ is a geodesic path in $P$ (the shortest path in $P$ from $v_0$ to $v_i$ lies inside $F$, and thus is equal to $\gamma$).
\newchanged{We claim that
the paths of $\Gamma$ subdivide $F$ into a constant number of 
$3$-anchor hulls (which need not be simple).  If the boundary of $\partial F$ between $v_i$ and $v_{i+1}$, $i = 1, \ldots, t-1$,  is a segment of a chord of $N$ or a subchain of a geodesic bounding $Q$, then take the $3$-anchor hull that is the geodesic hull of the three point anchors $v_0, v_i, v_{i+1}$.  If the boundary of $\partial F$ between $v_i$ and $v_{i+1}$
is a subchain of $\partial P$, then take the $3$-anchor hull that is the geodesic hull of $v_0$ and the polygon chain. 
Finally, if if the boundary of $\partial F$ between $v_0$ and $v_1$ or between $v_t$ and $v_0$ is a subchain of $\partial P$, then take the $3$-anchor hull of the polygon chain. 
}

\smallskip
\noindent{\bf 4. 
\newchanged{Find a simple $3$-anchor hull $Q' \subseteq F$ that contains the edge center.}}
Call the Geodesic Oracle (Lemma~\ref{lem:geodesic-oracle}) in $Q$ for each of the $O(1)$ geodesics
of $\Gamma$.
Halt if we find the center on one of the geodesics.
\newchanged{Otherwise, the geodesic oracle tells us which region of the partition by $\Gamma$ contains the edge center in its interior, and this gives us a \defn{simple} $3$-anchor hull $Q'$ with the edge center in its interior.}
Each of the constant number of calls to the geodesic oracle takes time 
$O(|Q| + t(Q))$.

\smallskip
\noindent{\bf 5. If $|Q'| \le 6$, find  ${\cal T}(Q')$ and ${\cal K}_T(Q')$.}
Since $Q'$ has constant size, we can find its intersection with each triangle in ${\cal T}(Q)$ in constant time, so we can find ${\cal T}(Q')$ and ${\cal K}_T(Q')$ in time $O(t(Q))$.

\smallskip
\noindent{\bf 6. If $|Q'| > 6$ find ${\cal K}(Q')$ and ${\cal T}(Q')$.}
We first find ${\cal K}(Q')$ by checking which chords of ${\cal K}(Q)$ cross $Q'$.
By Observation~\ref{obs:3-geo-cell}, the 
$3$-anchor hull
$Q'$ is bounded by at most three polygon chains and three geodesic chains.
A chord of ${\cal K}(Q)$ crosses $Q'$ if and only if it has an endpoint interior to one of polygon chains of $Q'$, or crosses one of the geodesic chains of $Q'$.  
We can test the former in constant time per chord because we know the endpoints of each chord on $\partial P$ (including knowing  which edge of $P$ contains the endpoint). %
We can test the latter by finding the intersections of the chords of ${\cal K}(Q)$ with each of the at most three geodesics bounding $Q'$ using Lemma~\ref{lem:intersection-with-geodesic} in $Q$. The runtime is $O(|{Q}| + |{\cal K} (Q)|)$ = $O(|{Q}| +  t(Q))$.  

Note that these tests also determine the endpoints of each chord of ${\cal K}(Q')$ on $\partial Q'$.

Finally, since each chord of ${\cal K}(Q)$ records the triangles of ${\cal T}(Q)$ that it bounds, we set ${\cal T} (Q')$ to be the triangles that are bounded by a chord of ${\cal K}(Q')$.
Note that this gives all triangles that intersect the interior of $Q'$ since no triangle contains $Q'$
by Lemma~\ref{lem:floating_cells_constant_size}.
This step takes
$O(t(Q))$ 
time.

\subsection{Problem with the Partitioning Scheme of Ahn et al.}
\label{section:counterexample}

In this section we explain the error in the step of the algorithm of Ahn et al.~\cite[Section 6]{linear_time_geodesic} where they take a 4-cell subdivided by chords of  
an $\epsilon$-net $N$ of constant size and partition the resulting faces into a constant number of $4$-cells. 
From the intersection points and endpoints of the chords in $N$, they shoot vertical rays up and down until either a chord of $N$ or the boundary of the outer 4-cell is reached. 
They claim that this subdivides each face into a constant number of 4-cells.  It is true that there are a constant number of regions, but not true that the regions are $4$-cells.
\reviewerchange{A counterexample is shown in Figure~\ref{fig:decomposition-1};}
there are five chords in $N$, and the construction of Ahn et al.~leaves a $5$-cell.

\reviewerchange{We briefly describe a way to fix their approach.
Find a trapezoidization of the faces of the arrangement of $N$ in the $4$-cell.  This can be done in time linear in the size of the $4$-cell.
The dual of the trapezoidization is a tree.  Working from the leaves of the tree, take a union of trapezoids until the resulting region is a $4$-cell, then chop it off and continue.
}

\subsection{Details for Section~\ref{section:epsilon-net-results}, $\epsilon$-Net Results for Stage 1}
\label{appendix:epsilon-nets}

\begin{claim}
\label{claim:4-cell}
If $Q$ is a 4-cell, then it is a $4$-anchor hull.
\end{claim}
\begin{proof}
Around the boundary of $Q$, there are four chords (or segments of chords), with two consecutive ones joined by a polygon chain or meeting at a point.  $Q$ is the geodesic hull of these $\le 4$ polygon chains and points.
\end{proof}

\begin{claim}
\label{claim:same-crossing-chords}
Let $Q$ be a 
\newchanged{$3$-anchor hull}
and $\psi(Q)$ be the corresponding expanded 3-anchor hull. Then 
a chord of $\cal{K}$ crosses $Q$ if and only if it crosses $\psi(Q)$,
i.e., ${\cal K}(Q) = {\cal K}(\psi(Q))$.
\label{claim:expansions}
\end{claim}

\begin{proof} 
One direction of the proof is simple: If a chord crosses $Q$, it must cross $\psi(Q)$ since $Q \subseteq \psi(Q)$.

For the other direction we prove
that if a chord $K \in {\cal K}$ does not cross the 
$3$-anchor hull $Q$, 
then it does not cross the 
expanded $3$-anchor hull $\psi(Q)$.
Suppose that a chord $K \in \mathcal{K}$ does not cross $Q$. 
Then $Q$ is contained in one of the closed half-polygons, say $H$, defined by $K$.
\newchanged{
This implies that the anchors of $Q$ are contained in $H$.
Since the corresponding expanded anchors were defined to not cross chords of $\cal K$, they are contained in $H$. Thus $\psi(Q)$, being the geodesic convex hull of sets in $H$, is also in $H$.  So $K$ does not cross $\psi(Q)$. 
} 
\end{proof}

\begin{lem}
\label{lem:subspace-oracle}
The $3$-anchor 
range space 
has a subspace oracle.
\end{lem}

\begin{proof} 
We must provide a deterministic algorithm that, given a subset ${\cal K}' \subseteq {\cal K}$ with $|{\cal K}'| = m$, computes the set
of ranges ${\cal R} = \{ {\cal K}'(Q) \mid Q \text{ is 
a $3$-anchor hull}\}$
in time $O(m^{d+1})$, where $d=6$ is the shattering dimension of the $3$-anchor range space.  

\newchanged{
We use the equivalence of 
the $3$-anchor range space and the expanded $3$-anchor range space (Lemma~\ref{lem:anchor-ranges}).
In Lemma~\ref{lem:constant_shattering_dimension} we proved that 
the number of expanded $3$-anchor hulls, $Q$, 
is $O(m^6)$.  
We must find these, and find, 
for each $Q$, the set of chords of $\cal K'$ that cross it. 

Recall that $A(\cal{K'})$ is the arrangement of the chords of ${\cal K}'$ plus the edges of $P$. 
This is an arrangement of line segments, with the special property that all segment endpoints are on the outer face. 
Recall also that $V({\cal K'})$ denotes the endpoints of the chords in $\cal K'$.
If $\cal{K'}$ has size $m$, then
$A(\cal{K'})$ has $O(m^2)$ faces, $O(m^2)$ internal vertices and edges, and $n+2m$ external vertices and edges on the boundary of $P$. 
In particular, the external vertices are the vertices of $P$ union $V({\cal K'})$.
For the algorithm we will avoid the dependence on $n$ by working with a combinatorial version of $A(\cal{K'})$ in which each minimal chain along $\partial P$ with endpoints in $V({\cal K'})$
is represented by a single ``dummy edge''. 
Note that the number of dummy edges is at most $2m$.  
Let \defn{$G(A({\cal K'}))$}
denote this planar graph, which has $O(m^2)$ vertices, edges, and faces.

We compute $G(A({\cal K'}))$ as follows. 
Compute 
the arrangement of the $m$ line segments $\cal K'$ in time $O(m^2)$.
Then traverse the outer face of the arrangement, adding dummy edges corresponding to subchains of $\partial P$ between vertices of  $V({\cal K}')$. 
We thus compute $G(A({\cal K'}))$ in time $O(m^2)$.

Next, we enumerate all of the possible expanded anchors: the $O(m^2)$ internal vertices, edges, and faces of $G(A({\cal K'}))$, and the $O(m^2)$ polygon chains, each represented by two endpoints in $V({\cal K'})$. 

For each of the $m$ chords $K$ of $\cal K'$ we enumerate the
anchors that lie in each of the two closed half-polygons $H$ defined by $K$.
In particular, we can traverse $G(A({\cal K'}))$ in time $O(m^2)$ to find the 
vertices, edges, and faces that lie in $H$.  We can also decide which of the $O(m^2)$  polygon chains lie entirely in $H$, based on where the endpoints lie.  This takes time $O(m^2)$ per chord, for a total of $O(m^3)$.

Finally, we can enumerate all the $O(m^6)$ choices of at most three expanded anchors that determine an expanded $3$-anchor hull $Q$.  For each choice we spend $O(m)$ time to find the set of chords crossing $Q$---begin with all of $\cal K'$ and eliminate chords that have all three anchors on the same side, since these are precisely the chords do not cross $Q$.
This gives us the set of chords crossing $Q$.}
\end{proof}

\newchanged{Designing a subspace oracle for the 
$4$-cell range space of Ahn et al.~seems problematic.  However, the above proof can be used to show that the $4$-anchor range space has a subspace oracle.  Thus constant-sized $\epsilon$-nets can be found in deterministic linear time. An $\epsilon$-net for the $4$-anchor range space is an $\epsilon$-net for the $4$-cell range space.  This repairs the approach of Ahn et al., modulo repairing their partition of a cell into $4$-cells (Appendix~\ref{section:counterexample}).}

\subsubsection{Overview of $\epsilon$-Nets}
\label{appendix:epsilon-net-overview}
This section contains background results on 
$\epsilon$-nets and their use in geometric divide-and-conquer algorithms.
For more details, we refer to the paper by Haussler and Welzl~\cite{haussler1987},  the survey by Mustafa and Varadarajan~\cite[Chapter 47]{toth2017handbook}, and the book by Mustafa~\cite{mustafa2022sampling}.
A \defn{range space} is a pair $(X,\mathcal{R})$ where $X$ is a \defn{ground set} of elements and $\mathcal{R}$ is a set of subsets of $X$.
We refer to the elements of $\mathcal{R}$ as the \defn{ranges} of the range space.
For any $\epsilon$ between 0 and 1, an \defn{$\epsilon$-net} for the range space $(X,\mathcal{R})$ is a subset $N \subseteq X$ with the following property: for every range $R$ in $\mathcal{R}$ with $|R| \geq \epsilon |X|$, we have $N \cap R \neq \phi$.
We use this as:

\begin{equation}
\text{if $N \cap R = \phi$, then $|R| < \epsilon |X|$}   \label{eqn:epsilon-net} 
\end{equation}

In many geometric settings, the ground set consists of hyperplanes.
In such cases, 
an $\epsilon$-net $N$ determines a hyperplane arrangement that partitions the space and suggests 
a natural divide and conquer approach based on the cells of this partition (called a \emph{cutting}~\cite{chazelle1993cutting,matouvsek1991cutting}).
We follow this approach, although our ground set consists of chords of the polygon rather than hyperplanes.

The size of the $\epsilon$-net 
directly controls the number of subproblems in the divide and conquer algorithm.
Efficient algorithms using this approach require an $\epsilon$-net of small size.
One way to guarantee constant sized $\epsilon$-nets is using combinatorial properties like the VC-dimension or shattering dimension.

Consider the range space $(X,\mathcal{R})$.
For a set $A \subseteq X$, the \defn{restriction of the ranges} to $A$, denoted \defn{${\cal R}_{|A} $}, is defined to be $\{A \cap R : R \in {\cal R} \}$.
A set $A$ is \defn{shattered} by the range space $(X,\mathcal{R})$ if ${\cal R}_{|A} $ is the power set of $A$.
The \defn{VC-dimension of a range space} $(X,\mathcal{R})$ is the maximum size of a set that can be shattered by the range space.
If a range space can shatter sets of arbitrarily large size, it has infinite VC-dimension.

The \defn{shattering dimension} of the range space $(X,\mathcal{R})$ is the minimum number $d$ such that for all $m$ and for all sets $A \subseteq X$ with $|A|=m$, we have $|{\cal R}_{|A}| \in O(m^d)$.
Equivalently, this says
that the number of ranges 
when restricted to
any subset of size $m$ is upper bounded by a polynomial in $m$ of degree equal to the shattering dimension.
Usually, upper bounds for the shattering dimension can be found more readily than those for the VC-dimension,
and upper bounds on the shattering dimension imply upper bounds on the VC-dimension, as expressed by
the following restatement of Lemma 5.14 from Har-Peled~\cite{har2011geometric}:

\begin{lem}
\label{lemma:shatter_vc_related}
If a range space has shattering dimension $d$, its VC-dimension is bounded by $O(d \log d)$, specifically by $12d \ln{(6d)}$.
\end{lem}

In the next lemma, we state the result of Haussler and Welzl~\cite{haussler1987} that a range space $(X,\mathcal{R})$ of  finite VC dimension has constant-sized $\epsilon$-nets.  
For a divide and conquer algorithm we also need an  algorithm to \emph{find} an $\epsilon$-net of constant size. 
A randomized algorithm is easier to obtain
but we need a deterministic algorithm.
Such a deterministic algorithm
was given by 
Matousek~\cite{matousek1991subspace} for any range space of shattering dimension $d$ that has a 
\defn{subspace oracle} which is defined to be a deterministic algorithm that, given a subset $X' \subseteq X$, computes the set $\mathcal{R}_{|X'}$ in time 
$O(|X'|^{(d+1)})$.

We summarize the results of Haussler and Welzl~\cite{haussler1987} and 
Matousek~\cite{matouvsek1989construction} in the following lemma.
Other sources for these results include
the survey  by Mustafa and Varadarajan~\cite[Chapter 47, Theorem 47.4.3]{toth2017handbook}, and the textbook by Mulmuley~\cite{mulmuley1994computational}.

\begin{lem}
\label{lem:constant_size_net}
A range space $(X,\mathcal{R})$ of finite VC-dimension has $\epsilon$-nets of size $O(\frac{1}{\epsilon} \log{\frac{1}{\epsilon}})$.
Furthermore, if the range space has a subspace oracle
then such an $\epsilon$-net can be found in deterministic time $O(|X|)$.
\end{lem}

\subsection{Stage 2: Algorithm for $Q$ a Triangle}
\label{section:constantQ}
\label{section:constantQ_Q}

In this section we 
\newchanged{outline the} algorithm to solve a subproblem for a
subpolygon $Q$ with $|Q| \le 6$ 
and its associated sets ${\cal T}(Q)$ and ${\cal K}_T(Q)$. 
Some of the triangles of  ${\cal T}(Q)$ may contain $Q$. 
We can triangulate $Q$ in constant time and apply the chord oracle to determine which triangle contains the center. Thus we will assume that $Q$ is a triangle.  

We must find the point that minimizes the upper envelope of the functions of the coarse cover ${\cal T}(Q)$.  
We crucially use the properties that the upper envelope is a geodesically convex function (Lemma~\ref{lem:geodesically-convex}) and that $Q$ is convex---together these imply that the upper envelope is a convex function.
We use a Megiddo-style prune-and-search technique, following the same approach as Ahn et al.~\cite[Section7]{linear_time_geodesic} but 
modified to deal with the edge center rather than the vertex center.

Each triangle $T$ of the coarse cover is the domain of a distance function to some edge $e$ of $P$.
Definition~\ref{defn:coarse-cover} tells us that functions associated with coarse cover elements have two different forms. Accordingly, we partition ${\cal T} (Q)$ into:
\begin{enumerate}
    \squeezelist
    \item \textbf{${\cal T}_1$}: Coarse cover elements whose associated functions have the form $d_2(x,v)+ \kappa$, where $v$ is a polygon vertex and $\kappa$ is a constant.
    \item \textbf{${\cal T}_2$}: Coarse cover elements whose associated functions have the form $d_2(x,{\bar e})$, where $\bar e$ is the line through polygon edge $e$.
\end{enumerate}

To determine the edge center, we must locate a point $x = (x_1, x_2)$ and a value $\rho$ to solve the following

\begin{equation}\label{program:optimization_floating}
\begin{array}{ll@{}ll}
\text{minimize}  & \rho &\\
\text{subject to}&
x \in Q\\
&d_2(x,v) + \kappa \le \rho \ \ \ &
\text{$x \in T \cap Q$; $v$, $\kappa$, and $T$
from 
an element of ${\cal T}_1$}\\
&d_2(x,{\bar e}) \le \rho & 
\text{$x \in T \cap Q$; $\bar e$ and $T$
from 
an element of  ${\cal T}_2$}\\
\end{array}
\end{equation}

We show how to solve 
Problem (\ref{program:optimization_floating}) 
\newchanged{in linear time} when the upper envelope of the coarse cover functions is convex.
(Without this condition 
the problem becomes hard since we then have unrelated convex constraints defined on different subdomains $T$.)

The constraints corresponding to ${\cal T}_1$ will be referred to as \textit{disk constraints}.
The constraints corresponding to ${\cal T}_2$ will be referred to as \textit{half-plane constraints}.
\newchanged{Ahn et al.~\cite[Section 7]{linear_time_geodesic} solve Problem (\ref{program:optimization_floating}) when there are no half-plane constraints.
Following their approach, we first describe previous work that handles the case when all triangles of the coarse cover contain $Q$.}

\subparagraph*{Special Case: All Triangles Contain $Q$.}
\newchanged{Note that in this case there is no need to assume that the upper envelope of the coarse cover functions is convex, since this}
follows immediately from the fact that each constraint
is convex on $Q$.
\begin{enumerate}
\squeezelist

\item \label{halfplane_constraints} 
Suppose all the constraints are half-plane constraints.
In this case, the problem 
is simply
linear programming in fixed
dimension 
which was solved in linear time by Megiddo~\cite{megiddo_linear} and Dyer~\cite{dyer1984linear}.
The idea 
is to pair up the lines that define the half-planes, and compute the angle bisector of each pair.
Knowing which side of the bisector contains the optimum point allows us to restrict the domain and discard one of the two constraints.
Find an appropriately-sized cutting of the 
bisectors.
\nnewchanged{
If we find which simplex of the cutting contains the optimum point, we can discard a constant fraction of the constraints.
The simplex can be found using an ``oracle'' that finds the optimum restricted to a line, i.e., in one lower dimension, 
and then testing whether this solution is the  global optimum, and if not, finding which side of the line contains the optimum.} 
The ``oracle'' on a line uses the prune-and-search technique applied repeatedly to the median point.

    \item \label{disk_constraints} 
    Suppose all the constraints are disk constraints.
This special case was also solved by Megiddo~\cite{megiddo_spanned_ball} and the solution was used in the geodesic center algorithm of Pollack et al.~\cite{pollack_sharir}.
The idea is again to pair up the constraints.
Although the constraints are non-linear, Megiddo showed that in the three-dimensional space of $x_1,x_2,\rho$, the locus of points where two constraints are equally tight is a plane that acts as the bisector between the two constraints.
\nnewchanged{The methods used to solve linear programming in three dimensions can then be applied to solve the problem in linear time.}

\item Finally, suppose there are both half-plane and disk constraints.
A linear-time algorithm for this case is given by Lubiw and Naredla~\cite{lubiw2021visibility} in their solution of the visibility center problem.
\newchanged{The idea is to pair up the half-plane constraints and separately pair up the disk constraints. After computing the bisector of each pair, the}
prune-and-search approach described above will %
prune away a constant fraction of 
both 
types of constraints in linear time.
\end{enumerate}

\subparagraph*{General Case.}
\newchanged{The new complication is that}
each constraint applies only in a triangular subdomain.
The idea for the solution one dimension down 
(with interval subdomains on a line) comes from the linear-time chord oracle of Pollack et al.~\cite{pollack_sharir}.  This was extended by Ahn et al.~\cite{linear_time_geodesic} to two dimensions. 
\newchanged{They dealt only with disk constraints, but we can extend the approach to handle both disk constraints and half-plane constraints, by pairing each constraint with another of the same type.

We outline the approach of Ahn et al.~\cite[Section 7.1]{linear_time_geodesic}.}
The basic idea is to add the subdomain boundary lines to the set of bisectors.
\newchanged{Each triangle of the coarse cover is bounded by two chords of $P$.
A pair of constraints (of the same type) then involves five linear constraints (a ``plane-set''): two for each triangular subdomain plus one bisecting plane.
Using cuttings and a ``side-decision'' algorithm
\nnewchanged{(which Megiddo called an ``oracle'')} 
we can in linear time restrict our search to a constant sized convex region $Q'\subseteq Q$ 
such that some constant fraction of the pairs of constraints have the property that no member of their plane-set
intersects $Q'$.
The claim is that at least one of each such pair can be eliminated. 
If $Q'$ is outside either of the two triangular domains, then the corresponding constraint is irrelevant
Otherwise, $Q'$ is inside both the domains.
In this case, we use the fact that it lies on one side of the bisector plane.
One constraint dominates over the other on this side of the bisector plane, and the other one may be ignored.
The last remaining ingredient is   
the ``side-decision'' algorithm which involves solving Problem~(\ref{program:optimization_floating}) 
restricted to a plane---this is the same problem down a dimension---and then testing
whether this solution is a local (hence global) solution and if not, finding which side of the plane contains the optimum.  
}

This completes the outline for solving Problem~\ref{program:optimization_floating} in linear-time.

\end{document}